\PassOptionsToPackage{unicode}{hyperref}
\PassOptionsToPackage{hyphens}{url}
\PassOptionsToPackage{dvipsnames,svgnames,x11names}{xcolor}
\documentclass[
  12pt]{article}

\usepackage{amsmath,amssymb,amsthm,amsfonts}
\usepackage{iftex}
\ifPDFTeX
  \usepackage[T1]{fontenc}
  \usepackage[utf8]{inputenc}
  \usepackage{textcomp} 
\else 
  \usepackage{unicode-math}
  \defaultfontfeatures{Scale=MatchLowercase}
  \defaultfontfeatures[\rmfamily]{Ligatures=TeX,Scale=1}
\fi
\usepackage{lmodern}
\ifPDFTeX\else  
\fi
\IfFileExists{upquote.sty}{\usepackage{upquote}}{}
\IfFileExists{microtype.sty}{
  \usepackage[]{microtype}
  \UseMicrotypeSet[protrusion]{basicmath} 
}{}
\makeatletter
\@ifundefined{KOMAClassName}{
  \IfFileExists{parskip.sty}{%
    \usepackage{parskip}
  }{
    \setlength{\parindent}{0pt}
    \setlength{\parskip}{6pt plus 2pt minus 1pt}}
}{
  \KOMAoptions{parskip=half}}
\makeatother
\usepackage{xcolor}
\makeatletter
\ifx\paragraph\undefined\else
  \let\oldparagraph\paragraph
  \renewcommand{\paragraph}{
    \@ifstar
      \xxxParagraphStar
      \xxxParagraphNoStar
  }
  \newcommand{\xxxParagraphStar}[1]{\oldparagraph*{#1}\mbox{}}
  \newcommand{\xxxParagraphNoStar}[1]{\oldparagraph{#1}\mbox{}}
\fi
\ifx\subparagraph\undefined\else
  \let\oldsubparagraph\subparagraph
  \renewcommand{\subparagraph}{
    \@ifstar
      \xxxSubParagraphStar
      \xxxSubParagraphNoStar
  }
  \newcommand{\xxxSubParagraphStar}[1]{\oldsubparagraph*{#1}\mbox{}}
  \newcommand{\xxxSubParagraphNoStar}[1]{\oldsubparagraph{#1}\mbox{}}
\fi
\makeatother

\usepackage{longtable,booktabs,array}
\usepackage{calc} 
\usepackage{etoolbox}
\makeatletter
\patchcmd\longtable{\par}{\if@noskipsec\mbox{}\fi\par}{}{}
\makeatother
\IfFileExists{footnotehyper.sty}{\usepackage{footnotehyper}}{\usepackage{footnote}}
\makesavenoteenv{longtable}
\usepackage{graphicx}
\makeatletter
\def\maxwidth{\ifdim\Gin@nat@width>\linewidth\linewidth\else\Gin@nat@width\fi}
\def\maxheight{\ifdim\Gin@nat@height>\textheight\textheight\else\Gin@nat@height\fi}
\makeatother
\setkeys{Gin}{width=\maxwidth,height=\maxheight,keepaspectratio}
\makeatletter
\def\fps@figure{htbp}
\makeatother

\makeatletter
\@ifpackageloaded{caption}{}{\usepackage{caption}}
\AtBeginDocument{%
\ifdefined\contentsname
  \renewcommand*\contentsname{Table of contents}
\else
  \newcommand\contentsname{Table of contents}
\fi
\ifdefined\listfigurename
  \renewcommand*\listfigurename{List of Figures}
\else
  \newcommand\listfigurename{List of Figures}
\fi
\ifdefined\listtablename
  \renewcommand*\listtablename{List of Tables}
\else
  \newcommand\listtablename{List of Tables}
\fi
\ifdefined\figurename
  \renewcommand*\figurename{Figure}
\else
  \newcommand\figurename{Figure}
\fi
\ifdefined\tablename
  \renewcommand*\tablename{Table}
\else
  \newcommand\tablename{Table}
\fi
}
\@ifpackageloaded{float}{}{\usepackage{float}}
\floatstyle{ruled}
\@ifundefined{c@chapter}{\newfloat{codelisting}{h}{lop}}{\newfloat{codelisting}{h}{lop}[chapter]}
\floatname{codelisting}{Listing}

\makeatother
\makeatletter
\@ifpackageloaded{caption}{}{\usepackage{caption}}
\@ifpackageloaded{subcaption}{}{\usepackage{subcaption}}
\makeatother

\ifLuaTeX
  \usepackage{selnolig}  
\fi
\usepackage[authoryear]{natbib}
\usepackage{bookmark}

\IfFileExists{xurl.sty}{\usepackage{xurl}}{} 
\hypersetup{
  pdftitle={Fingerprint Analysis for Climate Change Detection and Attribution under a Latent Factor Model},
  pdfauthor={Haoran Li; Yan Li},
  pdfkeywords={Factor models; Measurement error; Random matrix theory; Regularization},
  colorlinks=true,
  linkcolor={blue},
  filecolor={Maroon},
  citecolor={Blue},
  urlcolor={Blue},
  pdfcreator={LaTeX via pandoc}}

\usepackage{algorithm}
\usepackage{algorithmic}
\usepackage{mathtools}
\usepackage{multirow}

\usepackage{bm}
\usepackage{bbm}

\usepackage{float}
\usepackage{enumitem}
\usepackage{diagbox} 

\newtheorem{theorem}{Theorem}[section]
\newtheorem{lemma}{Lemma}[section]
\newtheorem{proposition}{Proposition}[section]

\numberwithin{theorem}{section}
\numberwithin{lemma}{section}
\numberwithin{proposition}{section}
\numberwithin{corollary}{section}
\numberwithin{definition}{section}
\numberwithin{cons}{section}
\numberwithin{remark}{section}
\numberwithin{exa}{section}
\numberwithin{table}{section}
\numberwithin{figure}{section}

\newcommand{\bs}[1]{\boldsymbol{#1}}
\newcommand{\mb}[1]{\mathbb{#1}}
\newcommand{\cov}[1]{\operatorname{Cov}(#1)}
\newcommand{\mc}[1]{\mathcal{#1}}
\newcommand{\inner}[2]{\langle#1, #2 \rangle}
\newcommand{\tr}{\operatorname{Tr}}
\newcommand{\mE}{\mathbb{E}}

\allowdisplaybreaks

\newcommand{\anon}{1}

\begin{document}

\def\spacingset#1{\renewcommand{\baselinestretch}%
{#1}\small\normalsize} \spacingset{1}

\if1\anon
{
  \title{\bf Fingerprint Analysis for Climate Change Detection and Attribution under a Latent Factor Model}
  \author{Haoran Li\thanks{These authors contribute equally to this work. Haoran Li is an Assistant Professor, Department of Mathematics and Statistics, Auburn University, 221 Parker Hall, Auburn, AL, 36849. \url{hzl0152@auburn.edu}} { and Yan Li}\thanks{Corresponding author. Yan Li is an Assistant Professor, Department of Mathematics and Statistics, Auburn University, 221 Parker Hall, Auburn, AL, 36849. \url{yzl0317@auburn.edu}} \\ Department of Mathematics and Statistics \\ Auburn University}
  \maketitle
} \fi

\if0\anon
{
  \bigskip
  \bigskip
  \bigskip
  \begin{center}
    {\LARGE\bf Fingerprint Analysis for Climate Change Detection and Attribution under a Latent Factor Model\par}
\end{center}
  \medskip
} \fi

\bigskip
\begin{abstract}
Detection and attribution (D\&A) analyses provide a statistical
framework for quantifying the contribution of external forcings to
observed climate change. Optimal fingerprinting, the primary approach
for D\&A, is formulated as an errors-in-variables regression with a
high-dimensional covariance structure. Reliable inference is challenging
because covariance matrices must be estimated from a limited number of
control simulations, and climate models may exhibit variability patterns
that differ from those of the observed climate system.
We develop a spiked optimal fingerprinting framework that exploits the
spiked structure of dominant climate variability modes while providing
stable covariance estimation in high-dimensional settings. The proposed
method develops bias-corrected spiked covariance estimation for
constructing adaptive weight matrices and incorporates variability
inflation adjustments to account for model--observation differences in
internal variability. We further develop a valid uncertainty 
quantification procedure for the scaling factor estimators and 
a residual consistency test for assessing model adequacy.
Numerical studies demonstrate improved estimation accuracy and
uncertainty quantification compared with existing practical approaches.
Applied to annual mean near-surface temperature data, the proposed
framework produces shorter and better-calibrated confidence intervals
and leads to different detection and attribution conclusions in several
regions, providing new insights into the interpretation of climate
attribution results.
\end{abstract}

\noindent%
\vfill

\vfill

\newpage
\spacingset{1.8} 

\section{Introduction}
\label{sec:introduction}

The assessment reports of the Intergovernmental Panel on Climate
Change (IPCC) have established that more than half of the observed
warming in recent decades is attributable to anthropogenic forcings,
primarily greenhouse gas emissions
\citep{hegerl2007understanding, Bind:etal:dete:2013,
  eyring2021human}. Detection and attribution (D\&A) analyses provide
the statistical basis for these conclusions. The primary framework is
optimal fingerprinting (OF), a multiple linear regression model
\citep{hegerl1996detecting, Alle:Tett:chec:1999,
  Alle:Stot:esti:2003}, in which the observed climate variable is
regressed on fingerprints representing expected responses to external
forcings. The resulting regression coefficients, called scaling
factors, are used to assess detection and attribution. A forcing is
detected when the confidence interval for its scaling factor lies above
zero; if the interval also contains one, the observed response is
considered consistent with the simulated response. 

\textcolor{black}{Unlike standard regression, OF has several distinguishing features. 
First, the observed climate response is a high-dimensional spatiotemporal vector 
with complex dependence characterized by an unknown covariance matrix $\Sigma$ \citep{Alle:Tett:chec:1999, Alle:Stot:esti:2003, Ribe:Plan:Terr:appl:2013}. Second, the fingerprints (predictors) are estimated from climate-model simulations and 
therefore are subject to simulation uncertainty, leading to an errors-in-variables (EIV) framework. 
Third, preindustrial control runs, which represent
internal climate variability under constant forcing, are available for estimating $\Sigma$. 
However, because only a limited number of control runs are available, the resulting sample covariance estimator may be unstable or singular. The traditional OF framework is formulated as follows:
\begin{equation}\label{eq:model_OF}
\begin{split}
\bs{Y}
&= \sum_{i=1}^p {X}_i \beta_i + \bs{\epsilon},
\qquad \bs{\epsilon} \sim \mathcal{N}(0,\Sigma), \\
\bs{\bar{X}}_i &= \frac{1}{n_i}\sum_{j=1}^{n_i} \bs{\tilde{X}}_{ij},  \quad \bs{\tilde{X}}_{ij} = {X}_i+ \bs{\eta}_{ij},
\quad \bs{\eta}_{ij}\sim\mathcal{N}(0,\Sigma),
\quad j=1,\dots, n_i, ~ i=1,\ldots,p, \\
\bs{Z}_j
&\stackrel{\mathrm{iid}}{\sim}\mathcal{N}(0,\Sigma),
\qquad j=1,\ldots,m,
\end{split}
\end{equation}
where $\bs{Y}\in\mathbb{R}^N$ denotes the observed climate response; ${X}_1,\ldots,{X}_p\in\mathbb{R}^N$ are the expected (unknown) responses to external forcings; $\{\bs{\tilde{X}}_{ij},\, j=1,\dots,n_i\}$ is the simulated fingerprint for the $i$th forcing from $n_i$ climate-model simulations; $\bs{Z}_1,\ldots,\bs{Z}_m$ are preindustrial control simulations used to estimate $\Sigma$. Under the classical OF formulation, the fingerprint simulation errors
and internal variability in the observed response and control
simulations are assumed to share a common covariance structure and are mutually independent. 
See \citet{chen2024statistical} and \citet{chen2026valid} for statistical reviews of this framework.}

\textcolor{black}{Consistent estimation and valid confidence intervals for the scaling factors $\beta_j$ are essential to D\&A analyses, but estimating $\Sigma$ from a limited number of control runs remains challenging. Early approaches projected the data onto a truncated basis of empirical orthogonal functions \citep{hegerl1996detecting, Alle:Tett:chec:1999}, potentially discarding important spatial structure. Subsequent methods have employed linear shrinkage \citep{ledoit2004well, ribes2009adaptation, li2025regularized}, Bayesian model averaging \citep{katzfuss2017bayesian}, integrated likelihood \citep{hannart2016integrated}, structural assumptions on $\Sigma$ \citep{ma2023optimal}, or localized dependence \citep{chen2026valid}. These methods may be computationally intensive, sensitive to prior choices, or limited by restrictive assumptions or suboptimal weighting. Moreover, failure to account for high-dimensional effects can produce overly narrow confidence intervals \citep{fuller1980properties, pesta2012total, delsole2019, li2021confidence, li2023regularized}; see \citet{li2025regularized} and \citet{chen2026valid} for recent calibration methods.}

In summary, existing methods for constructing valid confidence
intervals either impose strong structural assumptions on the covariance
matrix or restrict the weight matrix to linear shrinkage. Climate
responses, however, may exhibit complex and nonlocal dependence,
including long-range dependence and teleconnections across widely
separated regions, which may not be adequately captured by
localization-based approaches. At the same time, climate variability
often exhibits a limited number of dominant large-scale modes, such as
ENSO-related variability and other low-frequency coupled
ocean--atmosphere patterns. Empirical orthogonal function analyses and
studies of major climate variability modes suggest that a substantial
fraction of the variability in high-dimensional climate fields can be
represented by a relatively small number of leading spatiotemporal
patterns \citep{monahan2009empirical, fasullo2020evaluation}. These
dominant modes correspond to leading eigendirections of the covariance
structure and motivate the spiked covariance representation adopted in
this work. Consistently, the ensemble covariance matrices examined in
Section~\ref{sec:real} also exhibit rapidly decaying spectra with a few
dominant eigenvalues.

Motivated by these features, we propose a new flexible OF framework with three main innovations. First, we adopt
a spiked covariance structure and construct a class of spiked weight
matrices that flexibly weight the dominant eigendirections without
imposing localized dependence assumptions. The spike weights are
selected in a data-driven manner by minimizing the asymptotic MSE of the scaling factor estimator, and we develop
consistent variance estimators and valid confidence intervals. Second,
we allow different variability inflation levels for the observed
response, simulated fingerprints, and control runs, thereby
accommodating model--observation differences in internal variability.
Third, we introduce a residual consistency test for assessing model
adequacy in high-dimensional settings. We establish the theoretical
properties of the proposed method and evaluate its finite-sample and
practical performance through simulations and real-data applications.

The remainder of the paper is organized as follows.
Section~\ref{sec:setup} introduces the proposed framework and the structure of datasets. 
Section~\ref{sec:method} develops the estimation and inference procedures.
Section~\ref{subsec:analysis_selection_theta} illustrates the selection of the optimal weight through toy examples.
Section~\ref{sec:simulation} reports simulation results under a range of settings. 
Section~\ref{sec:real} presents D\&A analyses of annual mean temperature across spatial and temporal domains. Section~\ref{sec:disc} concludes with a discussion. Additional details are provided in the Supplementary Materials.
\section{Problem Setup and Data Structure}
\label{sec:setup}

In this section, we extend the OF framework introduced in Section~\ref{sec:introduction} to incorporate latent factors and allow the scale of internal variability to differ between climate-model simulations and the observed climate system.

Recall regression error $\bs{\epsilon}$, the fingerprint noise $\{\bs{\eta}_i\}$, and the control runs $\{\bs{Z}_j\}$.  We assume that the regression error $\bs{\epsilon}$ is decomposed into latent factor components and iid noise:
\begin{equation}
\bs{\epsilon} = \sum_{k=1}^K \sqrt{\ell_k} {U}_k \bs{f}_k^{(y)} +\sigma\bs{e}^{(y)},
\quad\text{with}\quad
\bs{f}_k^{(y)}
\stackrel{\mathrm{iid}}{\sim}\mathcal{N}(0,1),
\quad
\bs{e}^{(y)}\sim\mathcal{N}(0,I_N).
\label{eq:model_y}
\end{equation}
Here, $\bs{f}_k$ represents the $k$th stochastic latent factor, whose effect across spatiotemporal dimension is determined by the deterministic loading vector $U\in\mathbb{R}^N $ and the scalar multiplier $\ell_k>0$. To ensure identifiability, we assume that the loading vectors $U_k$ are orthonormal, so that $U_k^T U_i=1$ if $k=i$ and $U_k^T U_i=0$ otherwise. Under this normalization, the magnitude of the loading effect is determined solely by $\ell_k$, henceforth referred to as the factor strength.


\textcolor{black}{Although it is standard in OF to assume that the fingerprint simulation errors $\bs{\eta}_{ij}$ and control runs $\bs{Z}_j$ share the same covariance structure as the response variability $\bs{\epsilon}$, 
recent climate studies have highlighted potential
mismatches between model-simulated and observed internal variability \citep{jain2023importance, rateb2026structural}. 
Recent work by \citet{li2026adaptable} also provides evidence that model-simulated variability 
tends to be inflated relative to observed variability across different regional scales.
We therefore allow $\bs{\eta}_{ij}$ and $\bs{Z}_j$ to have different variability levels while sharing the same 
latent-factor structure as $\bs{\epsilon}$, with variability multipliers $s_x>0$ and $s_z>0$, respectively. In particular, we assume
the following model:}
\begin{align}
\bs{\eta}_{ij} &= s_x \Big(\sum_{k=1}^K \sqrt{\ell_k} U_k \bs{f}_{ijk}^{(x)} + \sigma \bs{e}^{(x)}_{ij} \Big), \quad j=1,\dots, n_i,\; i=1,\dots, p;\label{eq:model_x} \\
\bs{Z}_j  &= s_z\Big(\sum_{k=1}^K \sqrt{\ell_k} U_k \bs{f}_{jk}^{(z)} + \sigma \bs{e}^{(z)}_j \Big), \quad j=1,\dots, m. \label{eq:model_z}
\end{align}
Here, $\{\bs{f}^{(x)}_{ijk}\}$ and $ \{\bs{f}^{(z)}_{jk}\}$ are iid counterparts of the latent factors $\{\bs{f}_k^{(y)}\}$; $\{\bs{e}_{ij}^{(x)}\}$ and $\{\bs{e}_j^{(z)}\}$ are iid counterparts of the noise $\bs{e}^{(y)}$. 

In what follows, we adopt a covariance-based perspective. Under the latent-factor models \eqref{eq:model_y}--\eqref{eq:model_z}, the covariance matrix of $\bs{Y}$ is
\begin{equation}
\label{eq:spike}
\Sigma= \sum_{k=1}^K \ell_k U_k U_k^T+\sigma^2 I_N= 
U\Lambda U^T+\sigma^2 I_N,
\end{equation}
where $U=(U_1,\dots,U_K)$ and $\Lambda=\operatorname{Diag}(\ell_1,\dots,\ell_K)$. Moreover, $\cov{\bs{\bar{X}}_{i}}=s_x^2n_i^{-1}\Sigma$ and  $\cov{\bs{Z}_{j}}=s_z^2\Sigma$.

The matrix $\Sigma$ is a finite-rank perturbation of a scaled identity matrix, commonly known as the \emph{spiked covariance model} \citep{johnstone2001spiked}. This model has been widely used in applications such as finance, signal processing, and biomedical research; see \citet{JohnstoneP2018} for a comprehensive review. The spectral properties of sample covariance matrices under the spiked covariance model have been extensively studied in the high-dimensional regime $m\asymp N$. We discuss the results relevant to our methodology in Section~\ref{sec:method}.



Figure~\ref{fig:data_structure} illustrates the data structure used in
D\&A analyses through global mean near-surface
temperature anomalies over 1966--2020. The data are aggregated into
11 nonoverlapping five-year means. The observed response $\bs{Y}$,
ensemble mean fingerprints $\bs{\bar X}_1$, $\bs{\bar X}_2$, and
$\bs{\bar X}_3$ corresponding to greenhouse gas (GHG), aerosol (AER),
and natural (NAT) forcings, respectively, are shown together with the
ensemble mean of the control simulations. The control ensemble mean
remains close to zero, consistent with the mean-zero assumption.

The shaded bands represent variability across the forcing simulations,
defined by the 2.5\% and 97.5\% empirical quantiles across ensemble
members. The ensemble sizes for GHG, AER, and NAT are $n_1=43$, $n_2=42$, and $n_3=40$, 
respectively, and the control ensemble contains $m=181$ simulations. This example
corresponds to a single spatial scale represented by the global mean.
In regional D\&A analyses, multiple grid boxes are pooled across space and time to form
high-dimensional spatiotemporal vectors
\citep{Ribe:Plan:Terr:appl:2013}. Their spatial and temporal dependence
then jointly determines the covariance structure.

\begin{figure}[htbp]
  \centering
  \includegraphics[width=4.5in]{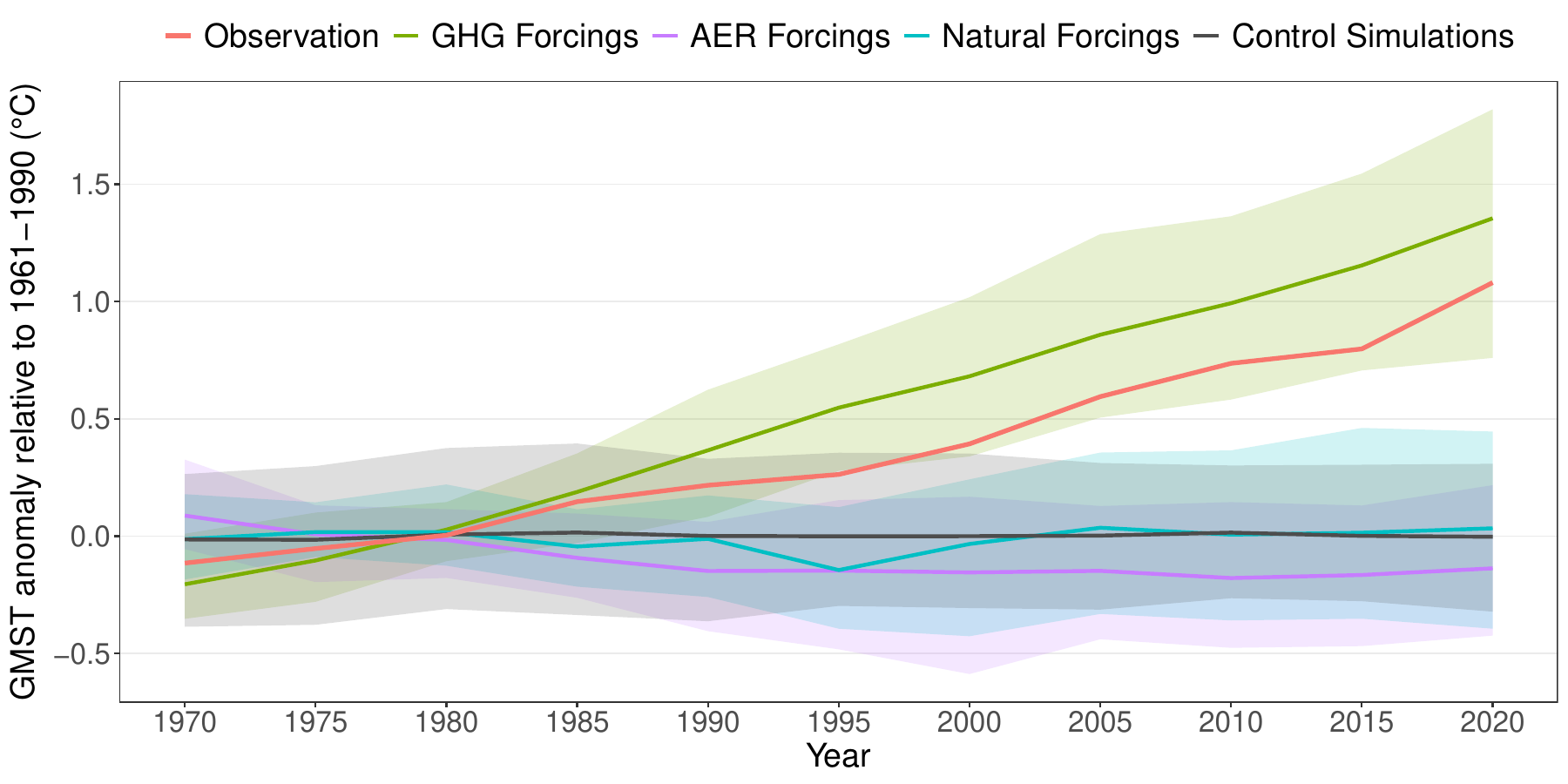}
\caption{Observed global mean temperature anomalies and
  ensemble mean responses to GHG, AER, and NAT forcings over
  1966--2020, together with the ensemble mean of the control
  simulations. Time points are nonoverlapping five-year means labeled
  by their ending years; shaded bands show the 2.5\% and 97.5\%
  empirical quantiles across forcing ensembles.}
  \label{fig:data_structure}
\end{figure}


Figure~\ref{fig:sample_spike} further presents the leading sample
eigenvalues of the covariance matrices estimated from the control
simulations for three representative regions: global, North
America, and eastern North America. These regions represent
global, continental, and subcontinental spatial scales, respectively.
Across the regions, the eigenvalues decay rapidly, with a small
number of dominant eigenvalues separated from the rest.
This pattern suggests that much of the dependence in internal climate
variability is concentrated along a few leading eigendirections,
providing empirical support for the assumed covariance structure.

\begin{figure}[htbp]
  \centering
  \includegraphics[width=5in]{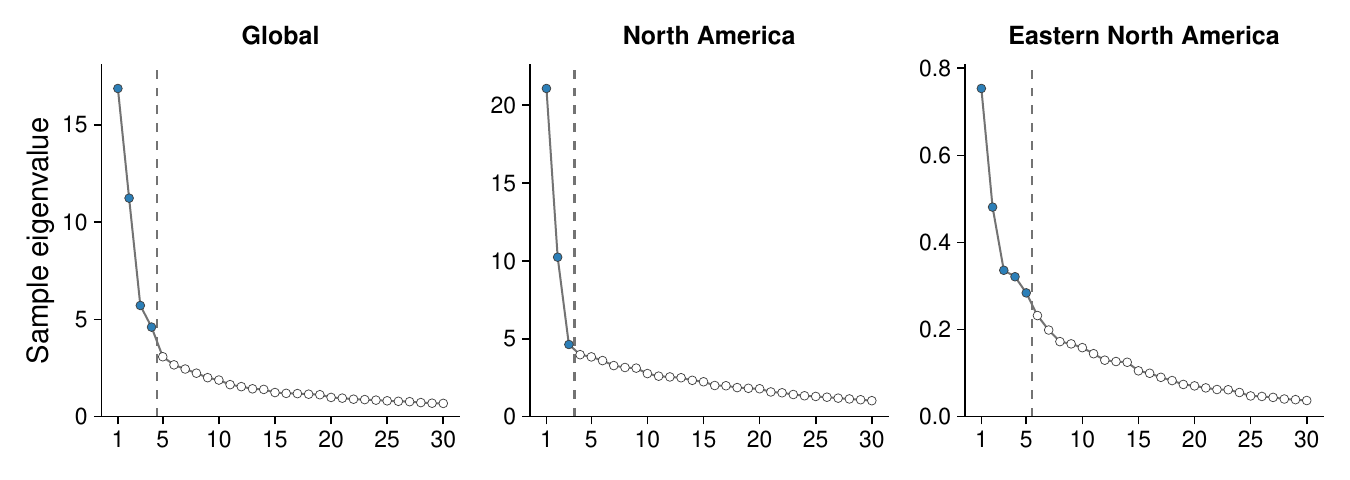}
  \caption{Leading sample eigenvalues of the covariance matrices
    estimated from the control simulations for the global, North
    American, and eastern North American regions.}
  \label{fig:sample_spike}
\end{figure}
\section{Methodology}\label{sec:method}

Under the framework of EIV, the method of \emph{total least squares} (TLS) is one of the broadly used estimation techniques. Denote
\begin{equation}\label{eq:ensemble_means}
X = (X_1,\dots, X_p), \quad \beta = (\beta_1,\dots, \beta_p)^T, \quad \bs{\bar{X}} = (\bs{\bar{X}}_1 ,\dots, \bs{\bar{X}}_p). 
\end{equation}
Further, call $D = \operatorname{Diag}(1/n_1,\dots, 1/n_p)$. Under the model specification, when the variability multiplier $s_x$ is given, the TLS estimator of $\beta$ with a given weight matrix $W$ is obtained as 
\begin{equation}
    \label{eq:TLS}
    \bs{\hat{\beta}} = \bs{\hat{\beta}}(s_x, W) = \arg\min_{\beta \in \mathbb{R}^p} \frac{\| W^{1/2} (\bs{Y} - \bs{\bar{X}} \beta) \|_2^2}{1 + s_x^2 \beta^T D\beta}. 
\end{equation}
The formula is adapted from standard results in the errors-in-variables literature \citep[e.g.,][]{fuller2009measurement}, and the derivation is omitted for brevity. Indeed, $\bs{\hat{\beta}}$ admits a closed-form expression, detailed in Lemma S.1.1 in the Supplementary Material. Ideally, $W$ is proportional to $\Sigma^{-1}$, so that the transformed observations $W^{1/2}\bs{Y}$ and fingerprints $W^{1/2}\bs{\bar{X}}$ have covariance matrices proportional to the identity matrix.

To implement the method, a consistent estimator of $s_x$ and an appropriate empirical choice of $W$ are required. Further, for the purpose of uncertainty quantification, consistent estimators of the factor strengths $\ell_k$, the noise variance $\sigma^2$, and the multiplier $s_z$ are additionally required. 

We assume
throughout that the number of factors $K$ is known. In the numerical
implementation, $K$ is selected using the data-driven procedures of
\citet{kritchman2008determining} and
\citet{passemier2017estimation}; a detailed comparison and practical
recommendations are provided in \citet{li2024testing}. 


Throughout the analysis, we assume the following technical conditions. 

\begin{enumerate}[label ={\bf C\arabic*}]\itemsep0.3em

    \item \label{enum:HD_regime} (Asymptotic regime) While the number of predictors $p$ is fixed, we assume that $N, m \to \infty$ simultaneously such that $\sqrt{N} |N/m- \gamma| \to 0$ for some $\gamma\in (0,\infty)$. In addition, the ensemble sizes $n_i$ are such that $n_i  = O (\log N)$ for $i=1,\dots,p$.  
  
    \item \label{enum:spike_strength} (Significant spikes) As \(N\to\infty\), the number of spikes \(K\), the factor strengths \(\{\ell_k\}\), the noise variance \(\sigma^2\), and the inflation multipliers \(s_x\) and $s_z$ remain fixed constants. Moreover, the spikes are assumed to be detectable in the sense that $\min_{k}{\ell_k}/{\sigma^2}>\sqrt{\gamma}$.

    \item \label{enum:asymptotic_full_rank} (Well-conditioning $X$) 
    {\small $0<\liminf_{N\to\infty} \lambda_{p}(X^\top X/N) \leq \limsup_{N\to\infty} \lambda_{1}(X^\top X/N)<\infty$}. Here, $\lambda_j(\cdot)$ represents the $j$th largest eigenvalue of a symmetric matrix. 
  \end{enumerate}


\subsection{Asymptotic Theory of Empirical Covariance Structure} \label{subsec:asymptotic_spikes}

For convenience, we write the inflated population factor strength for the control runs as $\pi_k = s_z^2 \ell_k$, $k=1,\dots, K$ and write the inflated noise variance as $\sigma_z^2 = s_z^2 \sigma^2$. 
In the context of D\&A analysis, the covariance of observational variability is commonly estimated by the sample covariance matrix of the control runs, defined as
\[
\bs{S}
=
\frac{1}{m}\sum_{j=1}^m \bs{Z}_j\bs{Z}_j^T.
\]
Let the ordered eigenvalues of \(\bs{S}\) be $\lambda_1(\bs{S})\geq \lambda_2(\bs{S})\geq \cdots \geq \lambda_N(\bs{S})$, and denote by \(\bs{\widetilde U}_k\) the eigenvector associated with the \(k\)-th largest eigenvalue, \(k=1,\dots,K\).

Under the spiked covariance model, we estimate the bulk noise level by averaging the trailing eigenvalues: $\bs{\tilde{\sigma}}_z^2= (N-K)^{-1}\sum_{j=K+1}^{N}\lambda_j(\bs{S})$. 
Further, define $\bs{\tilde{\pi}}_k = \lambda_k(\bs{S})-\bs{\tilde{\sigma}}_z^2$, for $k=1,\dots,K$.
Then \(\bs{\tilde{\pi}}_k\) is the \emph{maximum likelihood estimator} (MLE) of \(\pi_k\), \(k=1,\dots,K\), and \(\bs{\tilde{\sigma}}_z^2\) is the MLE of \(\sigma^2_z\) \citep{anderson1956statistical,li2024testing}. Consequently, the MLE of the inflated population covariance matrix \(s_z^2\Sigma\) is
\[
\bs{\widetilde{\Sigma}}=\sum_{k=1}^K \bs{\tilde{\pi}}_k\,\bs{\widetilde U}_k\bs{\widetilde U}_k^T+ \bs{\tilde{\sigma}}_z^2 I_N.
\]

Under the classical regime where the number of control runs \(m\) is substantially larger than the dimension \(N\), these estimators are consistent. However, when \(m\) is limited relative to \(N\), as in Condition~\ref{enum:HD_regime}, their estimation accuracy deteriorates.

In the statistical literature, the asymptotic behavior of the empirical spikes and noise variance under the regime of Condition~\ref{enum:HD_regime} has been studied extensively. We summarize several key results from \cite{paul2007asymptotics}, \cite{onatski2012asymptotics}, and \cite{passemier2017estimation}, which form the theoretical foundation of the current work.

Define two auxiliary functions 
\begin{align*}
    &\psi(x,\gamma) = x+ \frac{\gamma (x+1) }{x}, \quad x\in (0, \infty);\qquad \zeta(x, \gamma) = \left[ \frac{1- \gamma/x^2}{1+ \gamma/x} \right]^{1/2}, \quad x\in [\sqrt{\gamma}, \infty).
\end{align*}
Since eigenvectors are identifiable only up to a sign change, following the convention, we assume without loss of generality that $\inner{\bs{\widetilde{U}}_k}{U_k} >0$, $k=1,\dots, K$. 

\begin{theorem}
    \label{thm:asymptotics_eigenvalues_eigenvectors}
    Suppose that Model \eqref{eq:model_z} and Conditions \ref{enum:HD_regime}--\ref{enum:spike_strength} hold. 
    \begin{itemize}
        \item[1.] The MLE $\bs{\tilde{\pi}}_k$ is biased upwards in the sense that 
        \begin{equation}\label{eq:bias_spikes}
        \bs{\tilde{\pi}}_k \xrightarrow[]{a.s.} \sigma_z^2 \psi\Big(\pi_k/\sigma_z^2, \gamma\Big) = \pi_k +\gamma \frac{\pi_k + \sigma_z^2}{\pi_k/\sigma_z^2}, \quad k=1,\dots, K.
        \end{equation}
        \item[2.] The angle between an empirical spiked eigenvector $\bs{\widetilde{U}}_k$ and a population spiked eigenvector $U_d$ satisfies
        \[ |\inner{\bs{\widetilde{U}}_{k} }{ U_{d} }|  \xrightarrow{a.s.} \delta_{kd} \zeta\Big(\pi_k/\sigma_z^2, \gamma\Big), \qquad  k,d \in \{1,\dots,K\} ,\]
        where $\delta_{kd} =1$ if $k=d$; $0$ otherwise. 
        \item[3.] The following decomposition of the spiked eigenvector $\bs{\widetilde{U}}_k$, $k=1,\dots, K$ holds. 
        \[ \bs{\widetilde{U}}_k = \bs{\omega}_k U_k + \frac{1}{\sqrt{N}} \bs{\omega}_k U \bs{q}_k + \sqrt{1- \bs{\tilde{\omega}}_k^2} U_{\perp} \bs{s}_k, \quad \mbox{where} \]
    \begin{itemize}
        \item[(i)] $\bs{\omega}_k = \zeta\big(\pi_k/\sigma_z^2, \gamma\big) + O_p({1}/{\sqrt{N}})$. 
        \item[(ii)] $\bs{q}_k \in \mb{R}^K$ is a random vector such that  $\bs{q}_k \stackrel{d}{\longrightarrow} \mc{N}(0, A_k)$, where the expression of the covariance matrix $A_k$ is not of direct interest in our context. See \cite{paul2007asymptotics} for the details on $A_k$. 
         
        \item[(iii)] $\bs{\tilde{\omega}}_k$ is such that 
        \[ \bs{\tilde{\omega}}_k^2 = \bs{\omega}_k^2 \left( 1+ \frac{2}{\sqrt{N}} \bs{q}_{kk} + \frac{1}{m} \|\bs{q}_k \|^2_2 \right),\]
        where $\bs{q}_{kk}$ is the $k$th element of $\bs{q}_k$ and $\|\cdot\|_2$ is the vector Euclidean norm. 
        \item[(iv)] $U_{\perp}$ denotes an orthonormal basis for the orthogonal complement of the column space of $U$, and $\bs{s}_k$ is a $(N-K)$-variate random vector following the uniform distribution on the unit $(N-K-1)$-sphere. 
    \end{itemize}
        \item[4.] The empirical noise variance $\bs{\tilde{\sigma}}_x^2$ is biased downward as 
        \[ \frac{N-K}{\sigma_z^2 \sqrt{2\gamma}} \left(\bs{\tilde{\sigma}}_z^2 - \sigma_z^2 \right)  + b(\sigma_z^2) \stackrel{d}{\longrightarrow}N(0,1), \quad \mbox{where} \quad b(\sigma_z^2) = \sqrt{(\gamma/2)} \Big(K + \sum_{k=1}^K \frac{\sigma_z^2}{\pi_k} \Big).\]
    \end{itemize}
\end{theorem}
To further interpret these results, Result 1 shows that the empirical factor strengths exhibit an asymptotic positive bias, that is, \(\bs{\tilde{\pi}}_k-\pi_k>0\) asymptotically. Result 2 shows that the empirical loading directions \(\bs{\widetilde{U}}_k\) are not perfectly aligned with the true loading directions \(U_k\), but instead form a systematic nonzero angle with them. Result 3 further characterizes the structure of \(\bs{\widetilde{U}}_k\). It shows that \(\bs{\widetilde{U}}_k\) can be decomposed into a component along the true loading direction \(U_k\) with approximate mass \(\zeta(\pi_k/\sigma_z^2, \gamma)\), together with a component in the noise eigenspace with approximate mass \(\big(1-\zeta^2(\pi_k/\sigma_z^2,\gamma)\big)^{1/2}\). Its projection onto other loading directions \(U_d\), \(d\neq k\), is asymptotically negligible. Finally, the empirical noise variance \(\bs{\tilde{\sigma}}_z^2\) exhibits a negative asymptotic bias after normalization. 

It is generally challenging to correct the bias in the empirical loading directions \(\bs{\widetilde{U}}_k\) unless additional structural assumptions on the population loadings \(U_k\), such as sparsity, are imposed. By contrast, bias-corrected estimators for the factor strengths and the noise level have been developed in the literature. The procedure, originally proposed in \citet{passemier2017estimation} and \citet{li2024testing}, is summarized in Algorithm~1 in Supplementary Materials.

Denote the bias-corrected estimators from Algorithm~1  as $\bs{\hat{\pi}}_k$ and $\bs{\hat{\sigma}}_z^2$. The following asymptotic results follow from \citet{passemier2017estimation} and \citet{li2024testing}.
\begin{proposition}
    \label{prop:consistency_adjusted_estimator}
    Suppose that Model \eqref{eq:model_z} and Conditions \ref{enum:HD_regime}--\ref{enum:spike_strength} hold. Then, 
    \[\frac{N-K}{\sigma_z^2 \sqrt{2\gamma}}(\bs{\hat{\sigma}}_z^2  - \sigma_z^2) \stackrel{d}{\longrightarrow}N(0,1)\qquad \text{and} \qquad \bs{\hat{\pi}}_k - \pi_k = O_p(1/\sqrt{m}), \quad k =1,\dots, K.\]
\end{proposition}

\subsection{Design of Weight Matrices}\label{subsec:design_weight}
The asymptotic behavior of the empirical covariance structure established in Section~\ref{subsec:asymptotic_spikes} indicates that the empirical loading directions \(\bs{\widetilde{U}}_k\), although biased, still retain substantial information about the population directions \(U_k\), provided that \(N/m\) is not excessively large. Motivated by this observation, we propose the following class of weight matrices:
\begin{equation}
\label{eq:weight_matrix_family}
W_{\Theta}
=
\sum_{k=1}^K \theta_k \bs{\widetilde{U}}_k \bs{\widetilde{U}}_k^T + I_N  = \bs{\widetilde{U}} \Theta \bs{\widetilde{U}}^T  + I_N ,  
\end{equation}
where $\theta_k\in[-1,\infty)$, $k=1,\dots,K$ are regularization parameters. Here, we set \(\Theta=\operatorname{Diag}(\theta_1,\dots,\theta_K)\). Note that $W_{\Theta}$ is guaranteed to be nonnegative definite. 

A special choice of $W_{\Theta}$ is obtained by setting 
$\theta_k= -\bs{\hat{\pi}}_k /( \bs{\hat{\pi}}_k + \bs{\hat{\sigma}}_z^2 )$, where \(\bs{\hat{\pi}}_k\) and \(\bs{\hat{\sigma}}_z^2\) are the bias-corrected estimators of \(\pi_k\) and \(\sigma_z^2\). Under this specification, \(W_{\Theta}\) serves as a naive estimator of \(\sigma^{2}\Sigma^{-1}\). Although this estimator remains biased due to the biased empirical directions \(\bs{\widetilde{U}}_k\), it provides a reasonable initial choice for subsequent analysis. Henceforth, we denote this particular choice of $\Theta$ as $\Theta_0$ and the resulting $W_{\Theta_0}$ as $\widehat{W}$. 

We introduce some objects that play a central role in the proposed procedure. Define 
\[ \Pi = \operatorname{Diag}(\bs{\hat{\pi}}_1, \dots, \bs{\hat{\pi}}_K) \qquad \text{and} \qquad \mc{Z} = \operatorname{Diag}\left( \zeta \Big(\frac{\bs{\hat{\pi}}_1}{\bs{\hat{\sigma}}_z^2}, \hat{\gamma} \Big), \dots, \zeta\Big( \frac{\bs{\hat{\pi}}_K}{\bs{\hat{\sigma}}_z^2}, \hat{\gamma} \Big)   \right). \]
Here, $\hat{\gamma} = N/m$. Let 
\[ \alpha =   \frac{ m\sum_{i=1}^p \sum_{j=1}^{n_i} (\bs{\tilde{X}}_{ij} - \bs{\bar{X}}_i)^T (\bs{\tilde{X}}_{ij} - \bs{\bar{X}}_i)  }  { (\sum_{i=1}^p n_i-p) \sum_{j=1}^m \bs{Z}_j^T \bs{Z}_j} .\]

For any admissible $\Theta$,  define 
{\footnotesize
\begin{align}
&\Delta_{\Theta}  =\frac{1}{N} \tr\Big(\Theta (\Pi \mc{Z}^2 + \bs{\hat{\sigma}}_z^2 I_K )  \Big)  + \frac{1}{N} \tr(\Pi) + \bs{\hat{\sigma}}_z^2; \label{eq:def_Delta}\\
&\Phi_{\Theta} = \frac{1}{N}  \tr \Big(\Theta^2 [\mc{Z}^2 \Pi + \bs{\hat{\sigma}}_z^2 I_K]^2 + 2 \Theta[\mc{Z}^2 \Pi^2 + 2\bs{\hat{\sigma}}_z^2\mc{Z}^2 \Pi + \bs{\hat{\sigma}}_z^4 I_K ] + [\Pi + \bs{\hat{\sigma}}_z^2 I_K]^2 \Big) + (1-\frac{K}{N})\bs{\hat{\sigma}}_z^4. \label{eq:def_Phi}\\
&\Psi_\Theta = \bs{\hat{\sigma}}_z^4 [1+ N^{-1}\tr (2\Theta + \Theta^2) ] + \frac{1}{N} \tr[\Pi^2 (\mc{Z}^2 \Theta + I_K)^2] \\
&\qquad + \bs{\hat{\sigma}}_z^2 N^{-1}\tr[\Pi (\mc{Z}^2 \Theta^2 + 2\Theta + \mc{Z}^{-2})]+\bs{\hat{\sigma}}_z^2 N^{-1} \tr [ \Pi \{ I_K + \mc{Z}^2 (2\Theta + \Theta^2)\}].\\
&\Omega_{\Theta} = \frac{1}{N}\bs{\bar{X}}^T \widetilde{W} \bs{\bar{X}} - \alpha \Psi_\Theta, \quad \text{with }\widetilde{W} = \bs{\widetilde{U}} \mc{Z}^{-2} \Pi( \mc{Z}^{2} \Theta+  I_K)^2\bs{\widetilde{U}}^T + \bs{\hat{\sigma}}_z^2 W_\Theta^2  .\label{eq:def_Omega}\\
& G_\Theta = \frac{1}{N}\bs{\bar{X}}^T W_\Theta \bs{\bar{X}} -  \alpha \Delta_{\Theta} D.\label{eq:def_G}
\end{align}}
Note that these functions can be computed efficiently from $\{\bs{\tilde{X}}_{ij}\}$ and $\{\bs{Z}_j\}$ for any $\Theta$.

\begin{lemma}
    \label{lemma:consistency_Delta_Phi}
    Suppose the Models \eqref{eq:model_OF}--\eqref{eq:model_z} and Conditions \ref{enum:HD_regime}--\ref{enum:asymptotic_full_rank} hold. Then, 
    for any fixed $\Theta$ such that $\theta_k \geq  -1$, we have 
    \[ \Delta_{\Theta}  -\frac{s_z^2}{N} \tr[ W_\Theta \Sigma ] = O_p(1/\sqrt{N}), \qquad
     \Phi_{\Theta} -\frac{s_z^4}{N} \tr[ W_\Theta \Sigma W_\Theta \Sigma ] =o_p(1),\]
     \[ \Big\|G_\Theta -\frac{1}{N}X^T W_\Theta X\Big\|_2 = o_p(1), \qquad \Big\|\Omega_{\Theta} - \frac{s_z^2}{N} X^T W_\Theta \Sigma W_\Theta X\Big\|_2 =o_p(1).\]
\end{lemma}

\subsection{Estimation of Variability Inflation Multipliers}\label{subsec:estimation_s}

To estimate the inflation multipliers $s_x$ and $s_z$, we consider a
user-specified interval $[0,\Upsilon]$, where $\Upsilon>0$ is a
sufficiently large upper bound. The subsequent analysis is insensitive
to the choice of $\Upsilon$ provided that the interval contains the true
value of $s_x$. 

Suppose a weight matrix $W = W_\Theta$ is given and consider the TLS estimator $\bs{\hat{\beta}}(s, W)$ when $s\in [0, \Upsilon]$. A candidate choice of $W$ is $\widehat{W}$.
 Define the loss function 
\begin{equation}
    \label{eq:loss_estimation_s}
   L(s,W)= \left|\frac{s^2}{N} \frac{\big\|W^{1/2} (\bs{Y} -\bs{\bar{X}} \bs{\hat{\beta}}(s,W) ) \big\|_2^2  }{1+ s^2 \big(\bs{\hat{\beta}}(s,W)\big)^T D \bs{\hat{\beta}}(s,W)}   - \alpha \Delta_{\Theta} \right|.
\end{equation}
We estimate $s_x$ and $s_z$ by 
\begin{equation}
    \label{eq:def_s_hat}
    \bs{\hat{s}}_x(W) = \arg\min_{s\in [0,\Upsilon]} L(s,W), \qquad \text{and} \qquad \bs{\hat{s}}_z(W) = \bs{\hat{s}}_x(W)/\sqrt{\alpha}.
\end{equation}

\begin{theorem}
    \label{thm:consistency_s_hat}
    Suppose that the Models \eqref{eq:model_OF}--\eqref{eq:model_z} and Conditions \ref{enum:HD_regime}--\ref{enum:asymptotic_full_rank} hold. Assume that $\Upsilon > s_x$. Then, for any admissible $\Theta$ as designed in Section \ref{subsec:design_weight} and $W = W_\Theta$, as $N\to\infty$, we have  $\bs{\hat{s}}_x(W_\Theta) - s_x = O_p(1/\sqrt{N})$ and $\bs{\hat{s}}_z(W_\Theta) - s_z  = O_p(1/\sqrt{N})$.
\end{theorem}

\subsection{Selection of Weight Matrices and Confidence Intervals}\label{subsec:CI_select_weights}
In this section, we derive the asymptotic distribution of $\bs{\hat{\beta}}_\Theta =\bs{\hat{\beta}}(\bs{\hat{s}}_x, W_{\Theta})$ for any admissible $\Theta$ as designed in Section \ref{subsec:design_weight}. 

\begin{theorem}
    \label{thm:normality}
Suppose that the Models \eqref{eq:model_OF}--\eqref{eq:model_z} and Conditions \ref{enum:HD_regime}--\ref{enum:asymptotic_full_rank} hold. Let $\Theta$ be as designed in Section \ref{subsec:design_weight} and consider $\bs{\hat{\beta}}_\Theta = \bs{\hat{\beta}}(\bs{\hat{s}}_x, W_\Theta)$. Then, as $N\to\infty$, 
\[  \sqrt{N} \Xi_\Theta^{-1/2} \Big(\bs{\hat{\beta}}_\Theta - \beta \Big) \stackrel{d}{\longrightarrow} N(0, I_p),\]
where $\Xi^{-1/2}_\Theta$ denotes the symmetric inverse square root of the matrix $\Xi_\Theta$ given as 
\[ \Xi_\Theta =  \bs{\hat{s}}_z^{-2} (1+\bs{\hat{s}}_x^2 \bs{\hat{\beta}}^T_\Theta D \bs{\hat{\beta}}_\Theta )G^{-1}_{\Theta} \Big\{ \Omega_{\Theta} + \alpha \Phi_{\Theta} (D^{-1} + \bs{\hat{s}}_x^2 \bs{\hat{\beta}}_{\Theta}\bs{\hat{\beta}}^T_\Theta)^{-1}  \Big\} G^{-1}_{\Theta}. \]
\end{theorem}


We now address the problem of selecting the regularization parameters $\Theta$. According to Theorem \ref{thm:normality}, the asymptotic total variance of $\bs{\hat{\beta}}_{\Theta}$ is estimated by the trace of $\Xi_\Theta$. We then propose to select the regularization parameters in a data-driven manner as:
\begin{equation}
    \label{eq:def_optimal_Theta}
    \bs{\widehat{\Theta}}  =\arg \min_{\Theta} \tr[\Xi_\Theta], \quad \text{subject to }~\theta_k \geq - 1 ~\text{ for all }k=1,\dots, K. 
\end{equation}

We next discuss the uncertainty quantification. We write $\bs{\hat{\beta}} = \bs{\hat{\beta}} (\bs{\hat{s}},  W_{\bs{\widehat{\Theta}}})$ in the rest of this section. 
In practice, accurate and reliable confidence intervals for $\beta_i$ play a critical role in detection and attribution analysis of climate variables. For each $\beta_i$, $i=1,\dots, p$, a marginal confidence interval at the asymptotic level $(1-\alpha)$ can be constructed as 
\begin{equation}
    \label{eq:CI_marginal}
    \beta_i \in  \Big[\bs{\hat{\beta}}_i -  \frac{z_{1-\alpha/2}}{\sqrt{N}} \Xi_{ii}^{1/2}, ~~\bs{\hat{\beta}}_i + \frac{z_{1-\alpha/2} }{\sqrt{N}} \Xi_{ii}^{1/2} \Big],
\end{equation}
where $z_{1-\alpha/2}$ denotes the upper $(1-\alpha/2)$-quantile of the standard normal distribution, $\bs{\hat{\beta}}_i$ is the $i$th element of $\bs{\hat{\beta}}$, and $\Xi_{ii}$ is the $i$th diagonal element of $\Xi_{\bs{\widehat{\Theta}}}$. To control the family-wise error rate across the marginal intervals, the Bonferroni correction can be applied by adjusting the individual confidence level to $1-\alpha/p$.


\subsection{Residual Diagnostics} \label{subsec:adequacy_test}

Assessing model adequacy is critical in statistical modeling. The problem is particularly challenging in the current context because of the strong correlation in the residuals and the lack of a consistent estimator for the covariance structure. We propose a residual-based diagnostic procedure. 
Specifically, define the residuals of the regression model as 
\begin{equation}\label{eq:residual}
\bs{\hat{\epsilon}} = \bs{Y} - \bs{\bar{X}}\bs{\hat{\beta}}.
\end{equation}

As an initial step, informal diagnostic plots of the residuals can provide useful insights into potential departures from the model specification. To this end, we standardize the residuals to have approximate unit variance by 
\begin{equation}
    \label{eq:std_residual}
    \bs{\hat{\epsilon}}_{\rm std} = \bs{\hat{s}}_z [\operatorname{Diag}(\bs{S})]^{-1/2} \bs{\hat{\epsilon}}.
\end{equation}
Plotting these standardized residuals against their indices, spatial or temporal coordinates, or the observed response $\bs{Y}$ may reveal systematic patterns. Substantial deviation from a pattern centered around zero with constant variance may suggest violations of model assumptions or evidence of model misspecification. 

Secondly, we develop a formal residual consistency test. Let 
\[ \bs{Q}(\bs{\hat{\epsilon}}) = \frac{1}{N} \bs{\hat{\epsilon}}^T W_{\bs{\widehat{\Theta}}} \bs{\hat{\epsilon}}  - [1+ \bs{\hat{s}}^2_x \bs{\hat{\beta}}^T D\bs{\hat{\beta}}] \bs{\hat{s}}_z^{-2} \Delta_{\bs{\widehat{\Theta}}}.\]

\begin{theorem}
    \label{thm:adequacy_test}
    Suppose that the Models \eqref{eq:model_OF}--\eqref{eq:model_z} and Conditions \ref{enum:HD_regime}--\ref{enum:asymptotic_full_rank} hold. We have 
    \[ \frac{\bs{Q}(\bs{\hat{\epsilon}})}{\sqrt{V_Q}}\stackrel{d}{\longrightarrow} N(0,1), \quad \text{where} \quad V_Q = \frac{2}{N\bs{\hat{s}}_z^4 }  (1+ \bs{\hat{s}}_x^2 \bs{\hat{\beta}}^T D \bs{\hat{\beta}})^2 \Phi_{\bs{\widehat{\Theta}}}.\]
\end{theorem}
The model assumptions are rejected as misspecified at asymptotic significance level $(1-\alpha)$, if  $\bs{Q}(\bs{\hat{\epsilon}})/\sqrt{V_Q} > z_{1-\alpha}$, where $z_{1-\alpha}$ denotes the $1-\alpha$ quantile of $N(0,1)$. The p-value is obtained accordingly.

\section{Analysis of the Selection Procedure of $\Theta$}\label{subsec:analysis_selection_theta}

In general, the optimal value of $\Theta$ is strongly influenced by
the degree to which the true design matrix $X \in \mathbb{R}^{N \times p}$ aligns 
with the leading spiked eigen-directions $U \in \mathbb{R}^{N \times K}$. To demonstrate this effect more clearly, we consider a simplified 
case with $K=1$ and $p=1$. In this setting, the selection problem reduces to choosing a single scalar parameter $\Theta$, and the alignment 
between $X$ and the leading spiked direction $U=U_1$ can be summarized by the proportion of masses in the true design matrix $X$ onto the 
spiked eigen-direction $u_1 \in \mathbb{R}^N$, i.e., the spike-direction proportion (SDP), defined as
\begin{equation}
\operatorname{SDP}\label{eq:SDP}
= \frac{\operatorname{tr}(X^T U U^T X)} {\operatorname{tr}(X^T X)}.    
\end{equation}



\begin{figure}[!ht]
    \centering
    \begin{subfigure}[t]{3.8in}
        \centering
        \includegraphics[width=\linewidth]{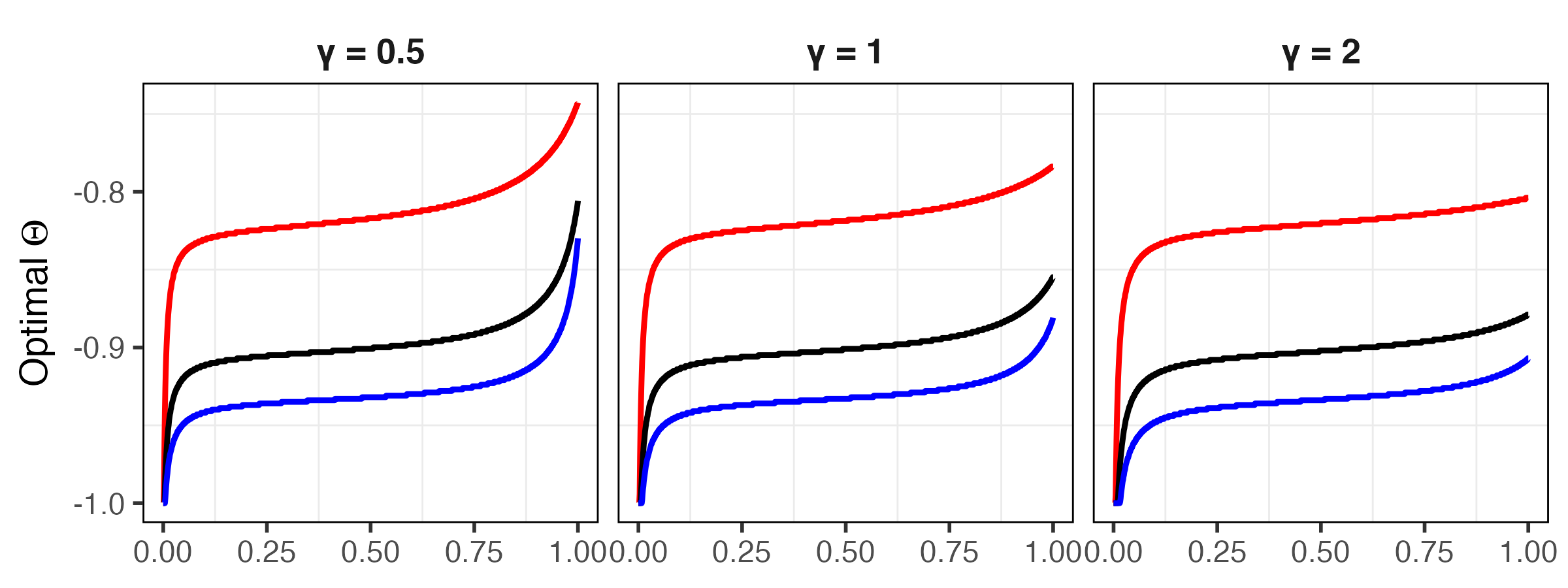}
        \caption{Optimal $\Theta$ as a function of SDP.}
        \label{fig:optimal_theta_panel}
    \end{subfigure}
    ~
    \begin{subfigure}[t]{5in}
        \centering
        \includegraphics[width=\linewidth]{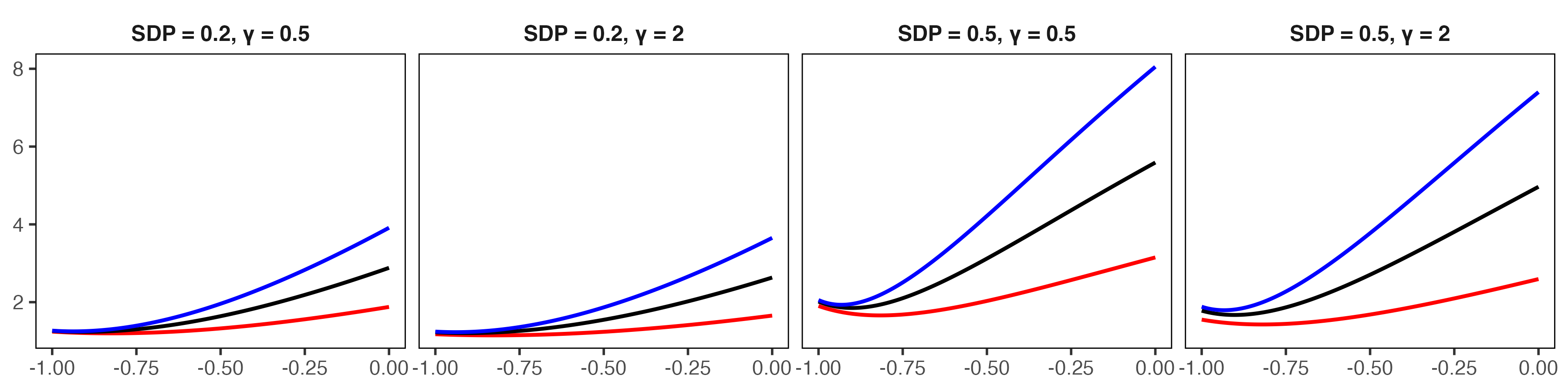}
        \caption{Expected objective $\mathbb{E}\{\Xi(\Theta)\}$ as a
        function of $\Theta$.}
        \label{fig:Xi_theta_panel}
    \end{subfigure}

    \caption{Selection of the optimal weight for $K=J=1$.
    The spike strengths are $\pi_1=5$ (red), $10$ (black), and
    $15$ (blue).}
    \label{fig:optimal_theta}
\end{figure}

We set $N=100$, $\sigma^2=1$, $s_x^2=s_z^2=1$, $\beta=2$, and
$D=1/4$, so that $\beta^2D=1$. Treating all true parameters as known,
we obtained the optimal $\Theta$ by minimizing $\mathbb{E}(\Xi)$.
Figure~\ref{fig:optimal_theta} shows the optimal $\Theta$ as a function
of SDP under different $\pi_1$ and $\gamma$ values, together with the
corresponding objective $\mathbb{E}(\Xi)$ as a function of $\Theta$.

\textcolor{black}{
The results illustrate how the optimal choice of $\Theta$ adapts to the
interaction between the signal structure and the spiked covariance
structure. The optimal values of $\Theta$ are consistently negative,
indicating that the identity weighting case is not optimal in these
settings and that incorporating covariance information improves
efficiency. As the spike strength $\pi_1$ increases, the optimal
$\Theta$ becomes more negative, indicating a stronger reduction of the
influence of the high-variance spiked direction in the weight
construction. In contrast, as SDP increases, the optimal $\Theta$
increases because the spiked direction contains more signal information
and therefore should retain greater influence. In the extreme case of
$\mathrm{SDP}=0$, the optimal choice is $\Theta=-1$, corresponding to
eliminating the influence of the leading spike in the weight
construction, since this direction contains no useful signal
information.
}

The lower panel further shows how the uncertainty of the
scaling factor estimator varies with SDP, spike strength, and 
$\gamma$. The minimum value of $\mathbb{E}\{\Xi(\Theta^*)\}$ increases 
with SDP, indicating larger estimation uncertainty when the signal 
overlaps more strongly with high-variance spiked directions. The qualitative 
patterns across different $\pi_1$ and $\gamma$ values are consistent with 
those observed in the subsequent simulation studies.

\section{Simulation Studies}\label{sec:simulation}

We conducted simulation studies to compare the proposed framework 
with existing approaches, under settings to mimic high-dimensional detection 
and attribution analyses at regional and larger spatial scales. We considered
$n_S\in\{20,40\}$ spatial grid boxes and $n_T=13$ temporal points,
giving overall dimensions $N\in\{260,520\}$. These dimensions are
representative of those encountered in detection and attribution
analyses of near-surface air temperature
\citep{eyring2021human} and in the real-data
applications of Section~\ref{sec:real}.

To generate the key components of the OF model in
\eqref{eq:model_y}--\eqref{eq:model_z}, we first specified the covariance
matrix $\Sigma$ and inflation multipliers $s_x^2$ and $s_z^2$ for the
simulated fingerprints and control runs. We used the spiked covariance
model $\Sigma=U\Lambda U^\top+\sigma^2 I_N$ in \eqref{eq:spike}, with
$\sigma^2=1$ and $s_x^2=s_z^2\in\{1,2\}$, so that
$\alpha=s_x^2/s_z^2=1$. For the spikes, 
we considered $K\in\{3,5\}$ spikes, with
$\Lambda=\operatorname{Diag}(8,7,5)$ for $K=3$ and
$\Lambda=\operatorname{Diag}(10,9,8,7,6)$ for $K=5$. For each
configuration, the population eigenvectors $U_1,\ldots,U_K$ were
generated to be orthonormal.

We next specified the true fingerprints in a two-forcing setting with
$X_1,X_2\in\mathbb{R}^N$, mimicking the real D\&A analysis of 
anthropogenic (ANT) and natural (NAT) forcings, and set $\beta_1=\beta_2=1$. 
The true fingerprints were
constructed to control the overall signal-to-noise ratio (SNR) and
their alignment with the spiked covariance directions, measured by the
spike-direction proportion (SDP) defined by~\eqref{eq:SDP}.
For a prescribed $\mathrm{SNR} \in \{2,5\}$, we set the total signal mass of the two
true fingerprints as $SS_X=\operatorname{tr}(X^\top X)=\mathrm{SNR}\times \operatorname{tr}(\Sigma)$,
where $X=(X_1,X_2)$. This allows us to vary the relative strength of
the forced signal and internal variability across different
signal-strength regimes.

For $\mathrm{SDP}=w\in\{0.2,0.5\}$, we
generated each fingerprint as the sum of two orthogonal components:
one lying in the leading spiked subspace spanned by the columns of
$U$, and the other lying in its orthogonal complement. The relative
magnitudes of these two components were chosen so that the prescribed
SDP was attained. To make the two forcings distinct, the first
fingerprint placed more weight on the strongest spike directions,
whereas the second distributed its spiked-subspace signal more evenly
across the leading $K$ directions. The two fingerprints were then
scaled so that $X_1^\top X_1=SS_X/3$ and
$X_2^\top X_2=2SS_X/3$, mimicking the real climate data in which the
NAT fingerprint has larger variation than the ANT fingerprint. This
construction allows the SNR and SDP to be controlled at their
prescribed levels while maintaining distinct signal structures for
the two forcings.

Given $\Sigma$, $s_x^2$, $s_z^2$, and the true fingerprints $X_1$ and
$X_2$, we generated the observed response, model-simulated
fingerprints, and preindustrial control runs. The observed response
$\bs{Y}$ was generated from~\eqref{eq:model_y}. For each forcing
$i=1,2$, noisy realizations $\bs{\tilde{X}}_{ij}$ were generated
from~\eqref{eq:model_x} with ensemble size $n_i=20$. We also generated
$m\in\{100,200,400\}$ independent control simulations from
\eqref{eq:model_z} for estimating $\Sigma$. All error terms were
normally distributed.

We compared the proposed spiked-weight OF framework, denoted by SOF,
with three competing methods for estimating the scaling factors 
and their confidence intervals. All methods use the TLS
estimator in~\eqref{eq:TLS} but differ in the construction of the
weight matrix $W$ and the inflation multipliers $s_x^2$ and $s_z^2$.
The number of spikes was selected using the method of
\citet{kritchman2008determining}. The first competitor is the
regularized OF method with a linear shrinkage weight estimator
\citep{ribes2009adaptation, li2025regularized}, denoted by ROF. The second competitor 
is an empirical spiked-weight OF method, denoted by EOF, which uses the MLE-based 
empirical spiked covariance estimator to construct the weight matrix and estimate $\Xi$.
For SOF, ROF, and EOF, the inflation multipliers are
estimated as in Section~\ref{subsec:estimation_s}. We also considered
an oracle benchmark using the true spiked eigenvectors, spike
strengths, and inflation multipliers.

We obtained estimates of the scaling factors $\beta_1$ and
$\beta_2$ and constructed 90\% confidence intervals (CIs) using the
corresponding asymptotic distributions of the weighted TLS estimators.
For the SOF, EOF, and Oracle methods, the construction is detailed in
Section~\ref{subsec:CI_select_weights}. For the shrinkage-based ROF
method, we followed the confidence interval construction in
\citet{li2025regularized}. For each configuration, the simulation was
repeated 1000 times. The random-number seeds and simulation scripts in
R are provided in the Supplementary Materials to ensure
reproducibility.

Table~\ref{tab:simulation_RMSE} reports the empirical biases and the
root mean square errors (RMSEs) for $\beta_1$ under normally distributed
errors with $K=5$, and Table~\ref{tab:simulation_CI} summarizes the
corresponding coverage rates and average lengths of the CIs.
Full simulation results are reported in Section~S.3 of the
Supplementary Materials and show similar patterns.

\begin{table}[!ht]
\scriptsize
\setlength{\tabcolsep}{2.4pt}
\renewcommand{\arraystretch}{0.8}
\centering
\caption{Empirical bias (RMSE) of the scaling factor estimator
$\hat{\beta}_1$ obtained from the TLS estimator in
\eqref{eq:TLS} under normal errors with $K=5$. Results are reported
for the proposed SOF method, ROF, EOF, and the Oracle benchmark across
different numbers of control runs $m$, spatiotemporal dimensions $N$,
signal-to-noise ratios (SNRs), and spike-direction proportions (SDPs).}
\label{tab:simulation_RMSE}
\resizebox{\textwidth}{!}{
\begin{tabular}{@{}llrrrr@{\hspace{7pt}}rrrr@{}}
\toprule
 &  & \multicolumn{4}{c}{$N=260$} & \multicolumn{4}{c}{$N=520$} \\
\cmidrule(lr){3-6}\cmidrule(lr){7-10}
 &  & \multicolumn{2}{c}{SDP = 0.2} & \multicolumn{2}{c}{SDP = 0.5} & \multicolumn{2}{c}{SDP = 0.2} & \multicolumn{2}{c}{SDP = 0.5} \\
\cmidrule(lr){3-4}\cmidrule(lr){5-6}\cmidrule(lr){7-8}\cmidrule(lr){9-10}
$m$ & Method & SNR = 2 & SNR = 5 & SNR = 2 & SNR = 5 & SNR = 2 & SNR = 5 & SNR = 2 & SNR = 5 \\
\midrule
\multicolumn{10}{c}{Case 1: $s^2_x = s^2_z=1$} \\
\multirow{4}{*}{100} & SOF & -0.003 (0.091) & -0.004 (0.054) & 0.004 (0.120) & 0.002 (0.075) & 0.001 (0.070) & -0.000 (0.043) & -0.007 (0.091) & -0.001 (0.056)
\\
 & ROF & -0.003 (0.096) & -0.004 (0.058) & 0.002 (0.126) & 0.003 (0.080) & 0.002 (0.073) & -0.000 (0.046) & -0.006 (0.094) & -0.000 (0.058)
\\
 & EOF & -0.017 (0.092) & -0.009 (0.055) & -0.018 (0.120) & -0.006 (0.076) & -0.009 (0.071) & -0.004 (0.044) & -0.022 (0.093) & -0.007 (0.057)
\\
 & Oracle & -0.002 (0.085) & -0.004 (0.050) & 0.002 (0.108) & 0.003 (0.068) & -0.000 (0.065) & 0.001 (0.040) & -0.004 (0.080) & -0.002 (0.047)
\\
\addlinespace[4pt]
\multirow{4}{*}{200} & SOF & 0.004 (0.090) & 0.001 (0.054) & 0.003 (0.112) & -0.001 (0.067) & -0.000 (0.066) & 0.001 (0.042) & 0.005 (0.083) & 0.002 (0.053)
\\
 & ROF & 0.003 (0.097) & 0.001 (0.057) & 0.002 (0.118) & 0.000 (0.070) & 0.000 (0.070) & 0.001 (0.045) & 0.005 (0.088) & 0.003 (0.056)
\\
 & EOF & -0.009 (0.090) & -0.003 (0.054) & -0.014 (0.112) & -0.008 (0.067) & -0.008 (0.066) & -0.002 (0.043) & -0.007 (0.083) & -0.002 (0.053)
\\
 & Oracle & 0.004 (0.087) & 0.001 (0.053) & 0.002 (0.106) & -0.002 (0.064) & -0.000 (0.063) & -0.000 (0.040) & 0.005 (0.076) & 0.001 (0.048)
\\
\addlinespace[4pt]
\multirow{4}{*}{400} & SOF & 0.001 (0.082) & 0.001 (0.052) & 0.002 (0.111) & 0.005 (0.066) & -0.003 (0.064) & -0.001 (0.039) & 0.005 (0.083) & 0.001 (0.048)
\\
 & ROF & 0.000 (0.086) & 0.001 (0.055) & 0.003 (0.119) & 0.003 (0.070) & -0.004 (0.069) & -0.001 (0.041) & 0.005 (0.089) & 0.002 (0.052)
\\
 & EOF & -0.010 (0.082) & -0.004 (0.053) & -0.013 (0.111) & -0.002 (0.066) & -0.010 (0.065) & -0.004 (0.039) & -0.005 (0.082) & -0.003 (0.048)
\\
 & Oracle & -0.001 (0.081) & 0.000 (0.052) & 0.002 (0.108) & 0.003 (0.064) & -0.004 (0.063) & -0.001 (0.038) & 0.005 (0.078) & 0.001 (0.045)
\\
\midrule
\multicolumn{10}{c}{Case 2: $s^2_x = s^2_z=2$} \\
\multirow{4}{*}{100} & SOF & 0.007 (0.099) & 0.002 (0.059) & 0.001 (0.131) & 0.001 (0.075) & 0.009 (0.079) & -0.002 (0.044) & -0.001 (0.095) & 0.003 (0.060)
\\
 & ROF & 0.007 (0.105) & 0.002 (0.063) & 0.004 (0.139) & 0.000 (0.080) & 0.010 (0.081) & -0.001 (0.046) & 0.001 (0.098) & 0.003 (0.062)
\\
 & EOF & -0.021 (0.099) & -0.009 (0.059) & -0.039 (0.133) & -0.016 (0.076) & -0.010 (0.079) & -0.010 (0.045) & -0.030 (0.097) & -0.009 (0.061)
\\
 & Oracle & 0.006 (0.093) & 0.001 (0.055) & -0.002 (0.117) & 0.000 (0.068) & 0.007 (0.070) & 0.000 (0.040) & -0.002 (0.080) & 0.003 (0.052)
\\
\addlinespace[4pt]
\multirow{4}{*}{200} & SOF & 0.003 (0.099) & 0.002 (0.058) & 0.009 (0.125) & -0.000 (0.075) & 0.002 (0.075) & 0.000 (0.042) & -0.004 (0.091) & 0.004 (0.052)
\\
 & ROF & 0.002 (0.105) & 0.001 (0.062) & 0.005 (0.133) & 0.001 (0.080) & 0.002 (0.079) & 0.001 (0.046) & -0.002 (0.096) & 0.004 (0.056)
\\
 & EOF & -0.020 (0.099) & -0.007 (0.058) & -0.026 (0.125) & -0.013 (0.076) & -0.014 (0.075) & -0.006 (0.043) & -0.027 (0.093) & -0.006 (0.053)
\\
 & Oracle & 0.002 (0.095) & 0.002 (0.056) & 0.005 (0.117) & -0.001 (0.071) & 0.003 (0.071) & -0.000 (0.041) & -0.006 (0.083) & 0.003 (0.047)
\\
\addlinespace[4pt]
\multirow{4}{*}{400} & SOF & 0.005 (0.095) & 0.002 (0.056) & -0.002 (0.123) & -0.001 (0.073) & -0.002 (0.071) & -0.001 (0.042) & 0.009 (0.086) & 0.000 (0.053)
\\
 & ROF & 0.005 (0.101) & 0.003 (0.059) & -0.002 (0.131) & -0.001 (0.077) & -0.002 (0.074) & -0.000 (0.045) & 0.007 (0.090) & 0.001 (0.056)
\\
 & EOF & -0.016 (0.095) & -0.006 (0.056) & -0.032 (0.124) & -0.013 (0.073) & -0.015 (0.071) & -0.007 (0.043) & -0.011 (0.085) & -0.007 (0.054)
\\
 & Oracle & 0.004 (0.092) & 0.001 (0.055) & -0.003 (0.117) & -0.002 (0.070) & -0.003 (0.068) & -0.001 (0.041) & 0.004 (0.082) & -0.000 (0.050)
\\
\bottomrule
\end{tabular}%
}
\end{table}

\begin{table}[!ht]
\scriptsize
\setlength{\tabcolsep}{2.0pt}
\renewcommand{\arraystretch}{0.8}
\centering
\caption{Empirical coverage rates (average lengths) of the 90\%
confidence intervals for the scaling factor estimator $\hat{\beta}_1$
constructed using the proposed SOF method, ROF, EOF, and the Oracle
benchmark under normal errors with $K=5$.}
\label{tab:simulation_CI}
\begin{tabular}{@{}llrrrr@{\hspace{7pt}}rrrr@{}}
\toprule
 &  & \multicolumn{4}{c}{$N=260$} & \multicolumn{4}{c}{$N=520$} \\
\cmidrule(lr){3-6}\cmidrule(lr){7-10}
 &  & \multicolumn{2}{c}{SDP = 0.2} & \multicolumn{2}{c}{SDP = 0.5} & \multicolumn{2}{c}{SDP = 0.2} & \multicolumn{2}{c}{SDP = 0.5} \\
\cmidrule(lr){3-4}\cmidrule(lr){5-6}\cmidrule(lr){7-8}\cmidrule(lr){9-10}
$m$ & Method & SNR = 2 & SNR = 5 & SNR = 2 & SNR = 5 & SNR = 2 & SNR = 5 & SNR = 2 & SNR = 5 \\
\midrule
\multicolumn{10}{c}{Case 1: $s^2_x = s^2_z=1$} \\
\multirow{4}{*}{100} & SOF & 88.8 (0.292) & 91.5 (0.181) & 89.0 (0.378) & 87.2 (0.232) & 90.4 (0.227) & 89.6 (0.140) & 88.9 (0.295) & 89.6 (0.183)
\\
 & ROF & 89.3 (0.310) & 90.9 (0.193) & 88.6 (0.401) & 87.5 (0.245) & 89.1 (0.235) & 89.5 (0.145) & 89.2 (0.303) & 90.2 (0.188)
\\
 & EOF & 87.7 (0.289) & 90.5 (0.180) & 84.5 (0.343) & 83.9 (0.211) & 85.0 (0.203) & 85.3 (0.126) & 80.6 (0.236) & 80.0 (0.146)
\\
 & Oracle & 90.3 (0.279) & 91.2 (0.173) & 89.1 (0.349) & 88.1 (0.213) & 89.8 (0.205) & 89.1 (0.127) & 89.5 (0.256) & 91.0 (0.156)
\\
\addlinespace[4pt]
\multirow{4}{*}{200} & SOF & 88.0 (0.283) & 89.6 (0.176) & 89.2 (0.362) & 89.8 (0.222) & 89.5 (0.213) & 88.0 (0.132) & 90.5 (0.275) & 88.1 (0.169)
\\
 & ROF & 88.8 (0.305) & 90.5 (0.188) & 90.2 (0.389) & 91.6 (0.239) & 88.6 (0.228) & 87.7 (0.141) & 89.5 (0.291) & 88.1 (0.180)
\\
 & EOF & 89.7 (0.293) & 90.3 (0.182) & 88.0 (0.353) & 89.2 (0.217) & 88.0 (0.206) & 87.0 (0.128) & 84.7 (0.245) & 83.4 (0.150)
\\
 & Oracle & 88.0 (0.280) & 90.1 (0.173) & 89.9 (0.348) & 91.1 (0.213) & 89.9 (0.205) & 87.5 (0.127) & 91.0 (0.256) & 89.3 (0.156)
\\
\addlinespace[4pt]
\multirow{4}{*}{400} & SOF & 91.7 (0.280) & 89.5 (0.174) & 89.2 (0.354) & 88.6 (0.217) & 90.1 (0.208) & 89.6 (0.129) & 89.3 (0.265) & 91.8 (0.163)
\\
 & ROF & 91.5 (0.297) & 90.1 (0.183) & 89.0 (0.379) & 90.1 (0.232) & 89.1 (0.221) & 90.8 (0.137) & 88.8 (0.282) & 90.7 (0.174)
\\
 & EOF & 92.8 (0.295) & 91.1 (0.183) & 89.2 (0.359) & 89.7 (0.220) & 89.9 (0.208) & 90.3 (0.129) & 87.6 (0.252) & 89.6 (0.154)
\\
 & Oracle & 92.6 (0.279) & 89.0 (0.173) & 90.1 (0.348) & 88.9 (0.213) & 90.5 (0.205) & 90.1 (0.127) & 89.5 (0.256) & 90.8 (0.156)
\\
\midrule
\multicolumn{10}{c}{Case 2: $s^2_x = s^2_z=2$} \\
\multirow{4}{*}{100} & SOF & 89.1 (0.315) & 89.2 (0.192) & 87.6 (0.408) & 89.0 (0.247) & 87.5 (0.243) & 91.3 (0.148) & 90.3 (0.317) & 89.2 (0.193)
\\
 & ROF & 88.6 (0.334) & 87.8 (0.205) & 86.4 (0.431) & 89.5 (0.262) & 87.8 (0.251) & 91.7 (0.154) & 89.8 (0.326) & 90.1 (0.199)
\\
 & EOF & 86.8 (0.310) & 87.6 (0.190) & 83.8 (0.368) & 85.2 (0.224) & 82.9 (0.216) & 87.0 (0.133) & 81.5 (0.255) & 80.0 (0.154)
\\
 & Oracle & 89.1 (0.302) & 90.6 (0.183) & 90.1 (0.380) & 90.9 (0.227) & 88.6 (0.221) & 90.5 (0.134) & 92.5 (0.280) & 89.2 (0.167)
\\
\addlinespace[4pt]
\multirow{4}{*}{200} & SOF & 86.4 (0.305) & 88.9 (0.186) & 87.8 (0.393) & 87.7 (0.237) & 88.2 (0.230) & 90.3 (0.139) & 89.3 (0.297) & 91.3 (0.180)
\\
 & ROF & 88.5 (0.328) & 89.5 (0.199) & 87.5 (0.420) & 88.0 (0.253) & 87.3 (0.245) & 89.1 (0.149) & 89.4 (0.314) & 90.4 (0.191)
\\
 & EOF & 87.7 (0.314) & 90.4 (0.192) & 87.4 (0.382) & 85.7 (0.230) & 85.5 (0.222) & 89.1 (0.135) & 84.0 (0.265) & 86.8 (0.160)
\\
 & Oracle & 87.9 (0.301) & 89.9 (0.183) & 89.9 (0.382) & 88.8 (0.227) & 88.3 (0.222) & 90.7 (0.134) & 91.1 (0.280) & 91.7 (0.167)
\\
\addlinespace[4pt]
\multirow{4}{*}{400} & SOF & 88.0 (0.303) & 89.3 (0.183) & 88.7 (0.387) & 88.5 (0.232) & 87.6 (0.224) & 90.2 (0.136) & 91.6 (0.289) & 89.1 (0.173)
\\
 & ROF & 88.4 (0.320) & 89.4 (0.194) & 87.3 (0.410) & 88.6 (0.248) & 88.9 (0.238) & 89.9 (0.145) & 91.5 (0.307) & 88.9 (0.185)
\\
 & EOF & 90.1 (0.318) & 90.6 (0.193) & 87.9 (0.391) & 88.3 (0.235) & 88.5 (0.225) & 89.2 (0.137) & 89.9 (0.274) & 86.6 (0.164)
\\
 & Oracle & 89.8 (0.302) & 89.6 (0.183) & 89.6 (0.381) & 89.2 (0.227) & 89.7 (0.222) & 89.6 (0.134) & 92.2 (0.280) & 90.6 (0.167)
\\
\bottomrule
\end{tabular}
\end{table}

The results in Table~\ref{tab:simulation_RMSE} show that the biases
of the point estimators were generally small for all methods, although
EOF exhibited slightly larger bias. The bias and RMSE generally
decreased with increasing numbers of control runs $m$, dimension $N$,
and SNR, and with decreasing SDP. These results indicate improved
estimation accuracy as the covariance structure is estimated more
reliably and the forced signal becomes stronger. Across different
methods, the Oracle estimator achieved the smallest RMSE, as expected
from using the truth. Among the practical methods, SOF generally produced slightly 
smaller RMSE than EOF, and both substantially outperformed ROF. 
Since all methods differ primarily in the construction of the weight
matrix $W$, these results highlight the benefit of using a
bias-corrected spiked covariance estimator for constructing the weight
matrix.

For the confidence intervals, Table~\ref{tab:simulation_CI} shows that
SOF, ROF, and the Oracle method maintained coverage rates close to the
nominal level across all settings, with average interval lengths
comparable to the empirical standard deviations of the estimators. 
In contrast, EOF exhibited noticeable undercoverage in most settings,
which is likely caused by the bias in the empirical spike estimates
used for estimating the asymptotic
variance of the scaling factors. The issue was less severe when the
number of control runs $m$ was comparable to the spatiotemporal
dimension $N$ and the SDP was low, where the empirical 
estimator can be estimated more accurately. When the internal
variability in the fingerprints exceeded that in the observed response,
namely $s_x^2=s_z^2=2$, the undercoverage became more severe. 
Among the practical methods, SOF achieved coverage rates closest to
the nominal level while producing shorter confidence intervals. 
The interval lengths decreased as the number of control runs $m$
and the spatiotemporal dimension $N$ increased, and as the signal
strength increased, reflecting improved estimation precision under more
informative settings. This behavior is also consistent with the 
intervals obtained in the real-data application.

Lastly, the proposed confidence intervals are constructed under the
assumption that the underlying error terms follow a normal
distribution. To assess their robustness to departures from this
assumption, we conducted additional simulations using heavier-tailed
multivariate $t_\nu$ distributions with degrees of freedom $\nu=5$ for
the random errors, while maintaining the same covariance matrix
$\Sigma$ as in the previous normal settings. Detailed results 
are presented in Section~S.3 of Supplementary Materials.
The results showed a similar pattern to the normal case: SOF outperformed the other
practical methods and approached the performance of the Oracle method
as the number of control runs $m$ increased. Overall, these results
highlight the benefit of the proposed bias-corrected spiked covariance
estimator for accurate uncertainty quantification and efficient
inference.


\section{Detection and Attribution of Temperature Change}
\label{sec:real}

\subsection{Regional Detection and Attribution Analysis}

To demonstrate the practical performance of the proposed approach, we
conducted detection and attribution analyses of annual mean near-surface
air temperature over 1951--2020 across different spatial scales using
HadCRUT5 observations and CMIP6 simulations. Following
\citet{zhang2006multimodel} and \citet{li2025regularized}, we
considered the global region (GL), the Northern Hemisphere (NH), the
Northern Hemisphere midlatitudes between $30^\circ$N and $70^\circ$N
(NHM), Eurasia (EA), North America (NA), and three subcontinental
North American regions: western (WNA), central (CNA), and eastern
(ENA) North America. These regions span a range of spatial scales and
potential signal-to-noise ratios. For each region, we performed a
two-signal analysis of anthropogenic (ANT) and natural (NAT) forcings
under Models~\eqref{eq:model_y}--\eqref{eq:model_z}. The three key
ingredients, namely, the observed response $\bs{Y}$, the simulated
fingerprints
$\{\widetilde{\bs{X}}_{ij}:j=1,\ldots,n_i\}$ for $i=1,2$, and the
control runs $\bs{Z}_1,\ldots,\bs{Z}_m$, were obtained as follows.

Observations $\bs{Y}$ were obtained from HadCRUT5
\citep{morice2021updated}, which provides monthly anomalies on a
$5^\circ \times 5^\circ$ grid relative to 1961--1990. Annual
anomalies were computed when at least nine months were available and
then averaged into nonoverlapping five-year means. Because the
anomalies are centered relative to the 1961--1990 climatology, 
we omit the 1961--1965 block, leaving 13 time points. To reduce the computational burden for large
spatial domains, GL and the continental regions were aggregated to
coarser grids using area-weighted averages, whereas the original
$5^\circ \times 5^\circ$ resolution was retained for WNA, CNA, and
ENA. Table~\ref{tab:regions} and Figure~\ref{fig:region_boundary}
summarize the spatial domains and dimensions.

\begin{table}[htbp]
  \scriptsize
  \setlength{\tabcolsep}{3.5pt}
  \renewcommand{\arraystretch}{0.6}
  \centering
  \caption{Summary of the spatial domains, coordinate ranges,
  spatiotemporal dimensions ($S$ and $T$), and observation dimensions
  $N$ after removing missing values.}
  \label{tab:regions}
  \begin{tabular}{llcccccc}
    \toprule
    Acronym & Regions & Longitude  & Latitude   & Grid size & $S$ & $T$ & $N$ \\
            &         & ($^\circ$E) & ($^\circ$N) & ($1^\circ \times 1^\circ$) & & &\\
    \midrule
    \multicolumn{8}{c}{Global and Continental Regions}\\
    GL & Global & $-$180 / 180 & $-$90 / 90 & $40 \times 30$ & 54 & 13 & 696 \\
    NH & Northern Hemisphere & $-$180 / 180 & 0 / 90 & $40 \times 30$ & 27 & 13 & 351 \\
    NHM & Northern $30^\circ N$ to $70^\circ N$ & $-$180 / 180 & 30 / 70 & $40 \times 10$ & 36 & 13 & 468 \\
    EA & Eurasia & $-$10 / 180 & 30 / 70 & $10 \times 20$ & 38 & 13 & 494 \\
    NA & North America & $-$130 / $-$50 & 30 / 60 & $10 \times 5$ & 48 & 13 & 624 \\
    \addlinespace[1ex]
    \multicolumn{8}{c}{Subcontinental Regions}\\
    \midrule
    WNA & Western North America & $-$130 / $-$105 & 30 / 60 & $5 \times 5$ & 30 & 13 & 390 \\
    CNA & Central North America & $-$105 / $-$85 & 30 / 50 & $5 \times 5$ & 16 & 13 & 208 \\
    ENA & Eastern North America & $-$85 / $-$50 & 15 / 30 & $5 \times 5$ & 21 & 13 & 273 \\
    \bottomrule
  \end{tabular}
\end{table}

\begin{figure}[!ht]
  \centering
  \includegraphics[width=4in]{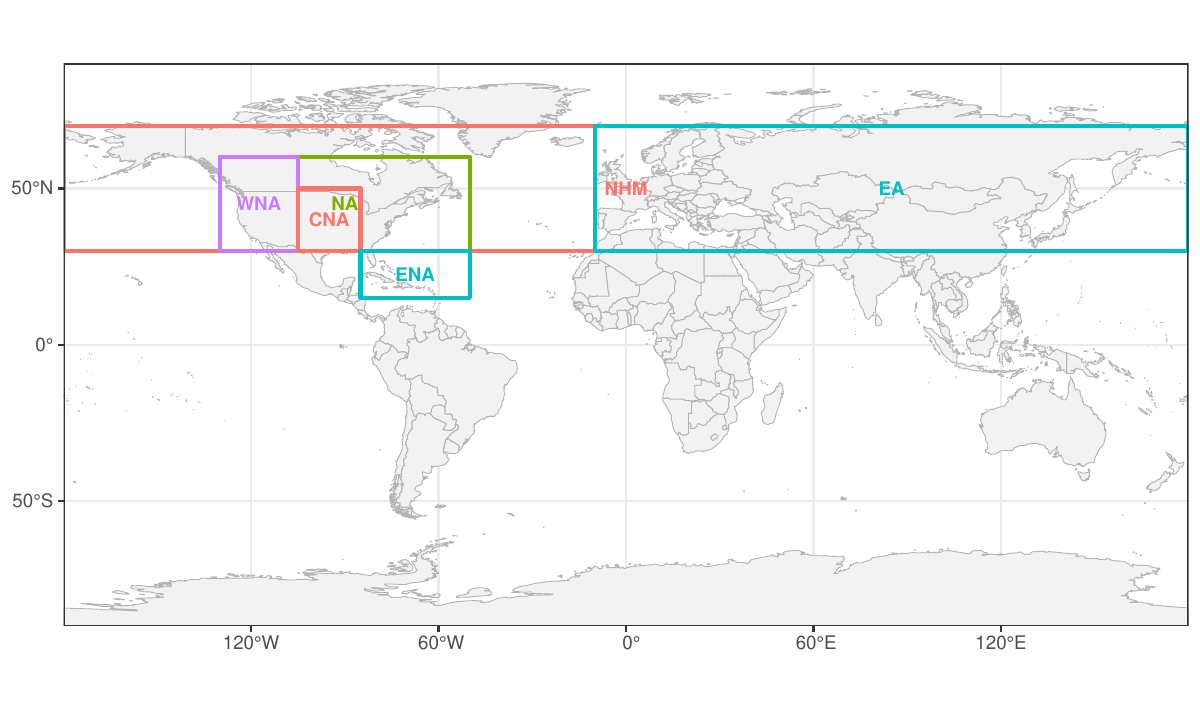}
  \caption{Geographic boundaries of the continental and subcontinental
  regions summarized in Table~\ref{tab:regions}. GL and NH are omitted
  for visual clarity.}
  \label{fig:region_boundary}
\end{figure}

The simulated fingerprints and control runs were derived from
multimodel simulations under the CMIP6 Detection and
Attribution Model Intercomparison Project
\citep{eyring2016overview}. Because the number of simulations available
from each individual climate model was limited, we did not consider
structural differences among models and instead treated the multimodel
ensembles as arising from a common forced-response framework. The
simulations included 43 runs from the hist-GHG experiment, driven by
changes in well-mixed greenhouse gas concentrations; 43 runs from the
hist-aer experiment, driven by changes in anthropogenic aerosols; 40
runs from the hist-nat experiment, driven by natural forcings; and
preindustrial control runs of varying lengths representing internal
climate variability.

For the ensemble means $\bs{\bar{X}}$ used in the TLS estimator in
\eqref{eq:ensemble_means}--\eqref{eq:TLS}, the NAT fingerprint
$\bs{\bar{X}}_{\mathrm{NAT}}$ was obtained by averaging the 40
available hist-nat runs. Because direct ANT simulations were not
available from CMIP6 simulations, we followed the standard linear additivity
assumption
$X_{\mathrm{ANT}}=X_{\mathrm{GHG}}+X_{\mathrm{AER}}$
\citep{zhang2006multimodel}. We constructed
$\bs{\bar{X}}_{\mathrm{ANT}}$ by summing the GHG and AER fingerprints
within each model and then averaging the resulting ANT fingerprints
across models. The corresponding effective ensemble size
$n_{\mathrm{ANT}}$ was used in the TLS estimation. The same temporal
averaging, spatial aggregation, and missing-data mask used for the
observations were applied to all simulated fingerprints.

Control runs were obtained from 29 CMIP6 preindustrial simulations,
ranging from approximately 100 to 1200 years in length. To account for
model drift, a long-term linear trend was removed separately at each
grid box. Assuming temporal stationarity, the detrended runs were
partitioned into nonoverlapping 70-year blocks corresponding to
1951--2020, yielding $m=181$ replicates for estimating
$\Sigma$. Each block underwent the same temporal averaging, spatial
aggregation, and missing-data masking as the observed response
$\bs{Y}$. Additional details on the CMIP6 models and control
simulations are available in \citet{li2025regularized}.

For estimation of the variability inflation multipliers $s_x^2$ and
$s_z^2$ in Section~\ref{subsec:estimation_s}, replicates are required
for each forcing. Because direct ANT replicates were unavailable, we
first estimated the multipliers using a three-signal analysis of GHG,
AER, and NAT. Assuming a common fingerprint variability multiplier
across forcings, the resulting estimates were then used in the
two-signal analysis. Figure~\ref{fig:scales} shows the
estimated multipliers.

\begin{figure}[!ht]
  \centering
  \includegraphics[width=4.5in,height=2in]{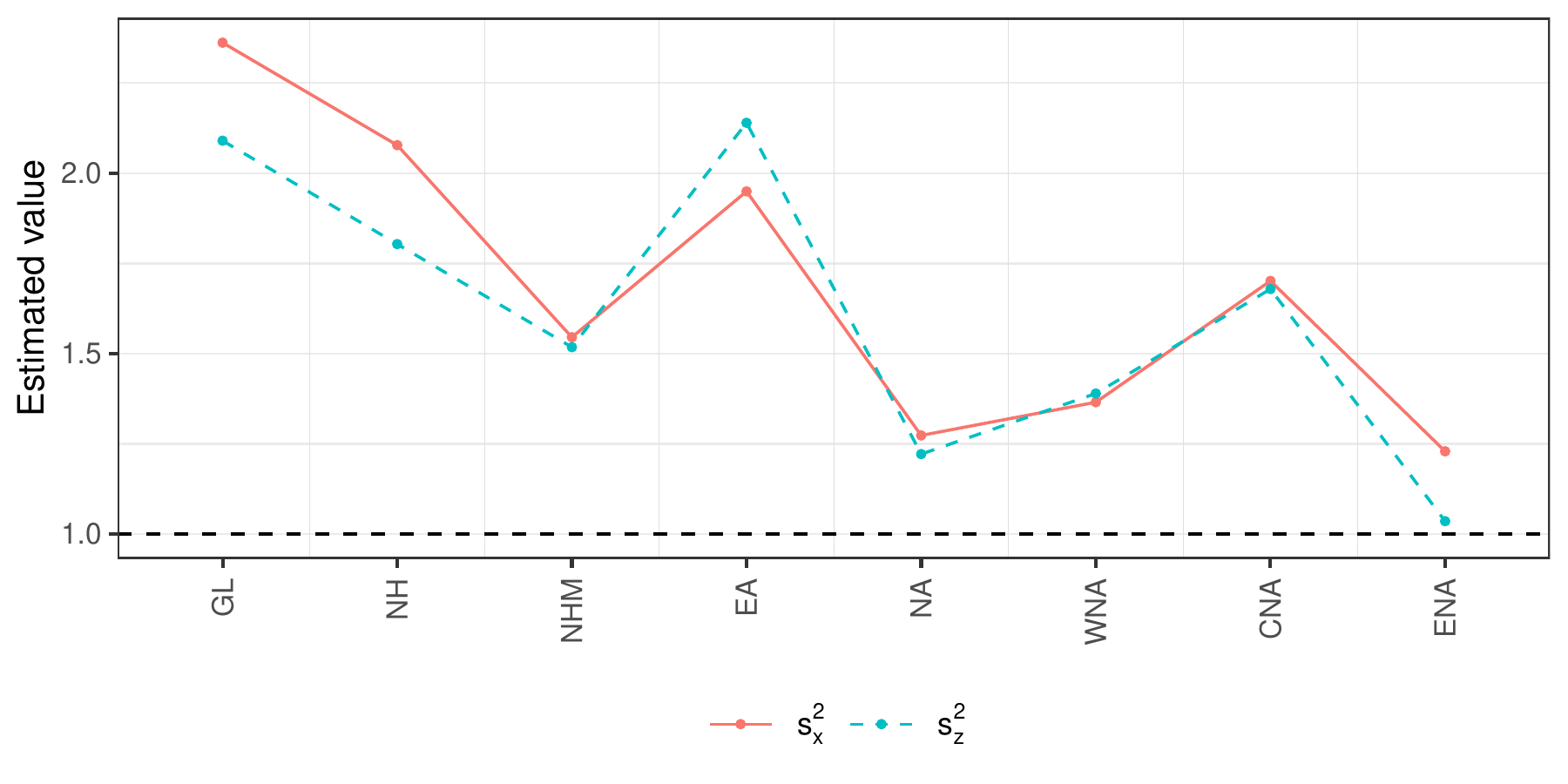}
  \caption{Estimated variance inflation multipliers $\hat s_x^2$ and
    $\hat s_z^2$ across the spatial domains. }
  \label{fig:scales}
\end{figure}

For most regions, both estimated multipliers exceeded one, indicating
greater internal variability in the simulated fingerprints and control
runs than in the observed response. This supports estimating the
variability multipliers rather than fixing them \emph{a priori} at one,
as is commonly done in conventional analyses. Under the estimated
multipliers, the residual consistency test yielded relatively large
p-values for the smaller domains NA, WNA, CNA, and ENA, suggesting
an adequate characterization of variability, but p-values close to
zero for the larger domains GL, NH, NHM, and EA. One possible
explanation is that internal variability differs across subregions
within these broad domains, so that a single multiplier $s_x^2$ or
$s_z^2$ cannot adequately capture the spatial heterogeneity.

We then estimated the scaling factors for the ANT and NAT forcings and
constructed the corresponding 90\% confidence intervals using the
proposed SOF method with an optimally selected spiked weight matrix;
EOF with a weight matrix constructed from empirical spike estimates;
and ROF with a weight matrix based on linear shrinkage. The number of
spikes was selected using the procedure in
Section~\ref{sec:simulation}, with at most ten spikes allowed. For
each method, we considered both the estimated inflation multipliers and
the conventional specification $s_x^2=s_z^2=1$. The results are shown
in Figure~\ref{fig:real_data}.

\begin{figure}[!ht]
  \centering
  \includegraphics[width=4.5in, height =3.5in]{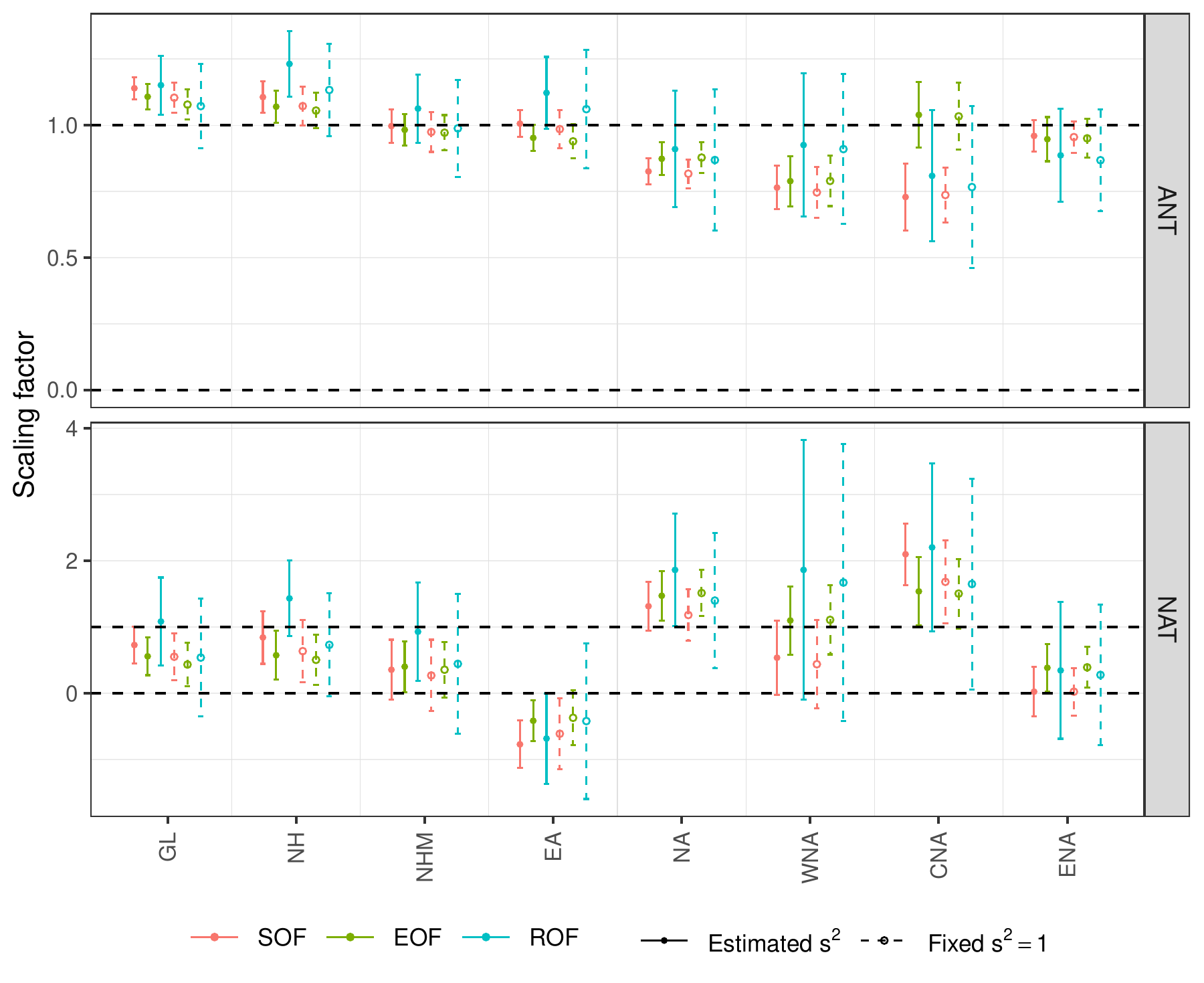}
  \caption{Estimated ANT and NAT scaling factors with 90\% confidence
  intervals across spatial domains over 1951--2020. Results are shown
  for SOF with bias-corrected spike estimation, EOF with empirical
  spike estimation, and ROF with linearly optimal weighting, using
  either the estimated inflation multipliers $s_x^2$ and $s_z^2$ or
  $s_x^2=s_z^2=1$.}
  \label{fig:real_data}
\end{figure}

Under the same multiplier specification, the three methods produced
noticeably different estimates, particularly for NAT,
whereas point estimates within each method were generally similar
between the estimated and unit multiplier specifications. The
confidence intervals differed more substantially. Using the estimated
multipliers generally produced shorter intervals than fixing them at
one, especially in the smaller regions where the residual consistency
test favored the estimated specification. Across methods, SOF and EOF
generally produced shorter intervals than ROF, with SOF slightly
shorter than EOF in most analyses. Together with the simulation
results, which showed substantial bias for EOF but accurate estimation
and valid uncertainty quantification for SOF, these findings support
the improved efficiency and reliability of SOF.

For the ANT forcing, all three methods detected the signal in every
region, as all corresponding intervals excluded zero.
Attribution conclusions, however, varied across methods and regions.
ROF supported both detection and attribution in most regions, except
GL and NH under the estimated multiplier specification, where the
confidence intervals lay entirely above one, suggesting that the
simulated ANT response underestimated the observed response. SOF and
EOF showed similar underestimation in GL and NH under both multiplier
specifications. In contrast, the intervals for NA and WNA lay entirely
below one, suggesting overestimation of the simulated ANT response;
SOF also indicated overestimation in CNA.

The NAT forcing showed a weaker and less stable pattern. The confidence
intervals were generally wider, and conclusions varied more across
methods and multiplier specifications. Under the unit multiplier, 
ROF supported both detection and attribution only in NA,
in agreement with the other methods. With estimated multipliers, ROF
additionally supported detection and attribution in GL, NH, and NHM.
SOF supported both detection and attribution in GL, NH, and NA, and
detection in CNA, where the interval lay entirely above one. EOF gave
a similar pattern but additionally supported detection and attribution
in WNA and provided weak evidence of detection in ENA. The intervals
for EA lay entirely below zero, which may reflect the difficulty of
estimating the weak NAT signal over this broad and climatically
heterogeneous region. Together with the simulation results, these
findings illustrate the practical impact of SOF on detection and
attribution conclusions.


\subsection{Retrospective Attribution Across IPCC Assessment Periods}

Given the central role of detection and attribution analyses in
successive IPCC assessment reports, we examined how conclusions for
global mean temperature change across observational periods
corresponding to these assessments when analyzed using modern
observations, climate model simulations, and statistical methods. We
considered periods beginning in 1951 and ending in the assessment
years from the First to the Sixth Assessment Report. Using HadCRUT5 observations and CMIP6 simulations, 
we performed two-forcing analyses of ANT and NAT and compared the proposed SOF method with the commonly used ROF approach, implemented with the
corrected inferential procedure of \citet{li2025regularized}. EOF was
excluded because of concerns regarding its inferential reliability.
The estimated scaling factors and corresponding 90\% confidence intervals are shown
in Figure~\ref{fig:results_ipcc}.

\begin{figure}[!ht]
  \centering
  \includegraphics[width=5in, height=2in]{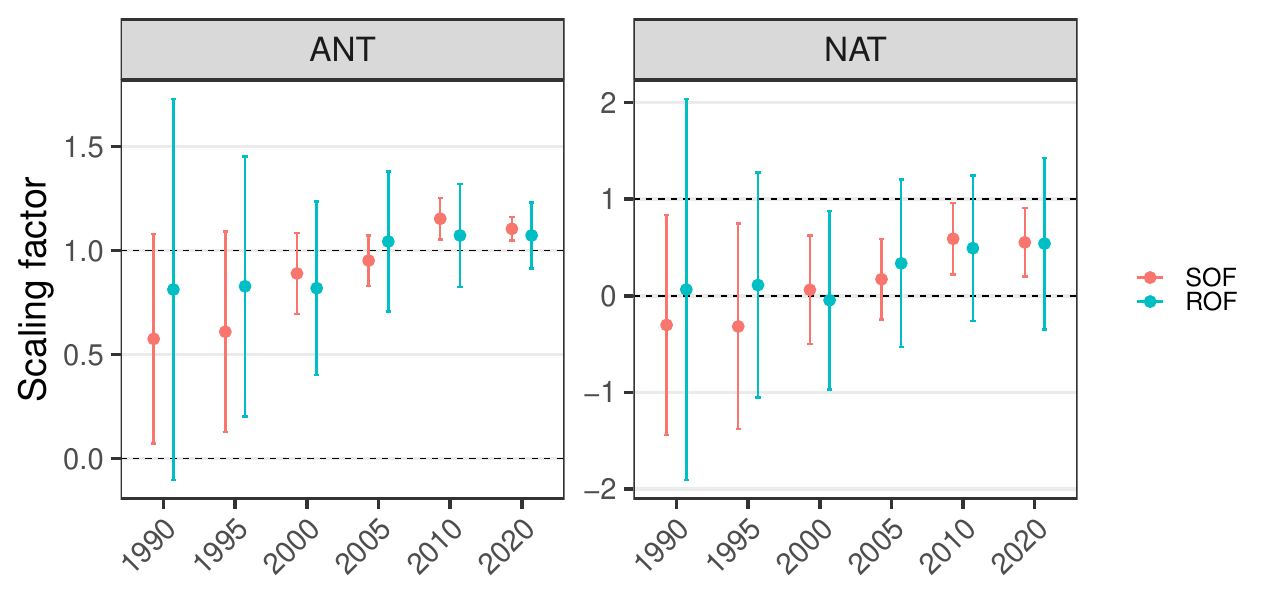}
  \caption{Estimated scaling factors for ANT and NAT forcings, 
  with 90\% CIs, for analyses from 1951 to cutoff years.
  Results are shown for SOF and ROF with $s_x^2=s_z^2=1$.}
  \label{fig:results_ipcc}
\end{figure}

Both methods produced progressively narrower confidence intervals
across successive assessment periods, reflecting the increasing
strength of the forced signal, but led to notably different conclusions
for the earliest periods. For the period ending with the First
Assessment Report, which stated that ``the unequivocal detection of
the enhanced greenhouse effect from observations is not likely''
\citep{IPCC1990}, SOF already supported both detection and attribution
of the ANT signal, whereas ROF did not. Beginning with the period
ending with the Second Assessment Report, which stated that ``the
balance of evidence suggests a discernible human influence on global
climate'' \citep{IPCC1995}, both methods supported detection of the
ANT signal. For the more recent assessment periods, however, the SOF
confidence intervals lay entirely above one, suggesting that the ANT
response simulated by CMIP6 models underestimated the magnitude of the
observed anthropogenic response.

For NAT, the signal was substantially weaker and the conclusions were
less robust. ROF did not detect the signal in any assessment period,
whereas SOF provided at most weak evidence of detection, primarily for
the 2010 and subsequent periods. Because these findings are based on a
single analysis, it is difficult to determine whether the weak NAT
detection reflects a reproducible signal or finite-sample variation.
Thus, unlike ANT, NAT was not robustly detected in this two-forcing
analysis of global mean temperature.

These retrospective findings should not be viewed as contradicting the
historical assessments, which necessarily reflected the
observations, climate model simulations, and statistical methods
available at the time. Rather, together with the simulation studies
and preceding regional analyses, they highlight the potential
importance of the proposed method when the forced signal is weak, 
and motivate revisiting existing D\&A analyses under such settings.


\section{Conclusion}\label{sec:disc}

We develop a spiked optimal fingerprinting framework for reliable detection and 
attribution analysis with limited control simulations and 
potential mismatch between model-simulated and observed variability.
The proposed framework exploits the spiked structure of internal variability 
while regularizing the remaining covariance spectrum to construct stable
weight matrices. It incorporates 
variability inflation, allowing for systematic differences between
model-simulated and observed variability. Together with the proposed
uncertainty quantification procedure and residual diagnostic, the
framework provides a practical approach for reliable inference without
imposing overly restrictive covariance structures.

Simulation studies demonstrate that the proposed method improves 
both estimation accuracy and uncertainty quantification 
relative to existing practical approaches, achieving lower estimation 
error while maintaining confidence interval coverage close to the nominal 
level and producing shorter intervals.
Our analyses of annual mean near-surface temperature changes further
illustrate the practical impact of the proposed framework. The
estimated variability inflation multipliers differed substantially from
the conventional unit specification, indicating that climate model
variability may not fully represent observed variability. The proposed framework led to
shorter confidence intervals and different detection and attribution
conclusions in several regions. These findings highlight the
importance of flexible uncertainty quantification for interpreting
climate attribution analyses and motivate further reassessment of D\&A
results under more comprehensive statistical frameworks.

Several extensions remain of interest. The current framework assumes
constant scaling factors and scalar variability inflation across the
analysis domain. Future work could incorporate spatially varying
scaling factors and region-specific variability adjustments to better
capture heterogeneous climate responses. 
In addition, the current
framework does not explicitly account for structural differences among
climate models within multimodel ensembles. Developing methods that
integrate model structural uncertainty into the fingerprinting
framework may further improve inference for large-scale climate
applications. 

\section{Disclosure statement}\label{disclosure-statement}

The authors declare that they have no conflicts of interest.

\section{Data Availability Statement}\label{data-availability-statement}

The climate datasets used in this study are publicly available at: 1) CMIP6, \url{https://aims2.llnl.gov/search/cmip6/} 2)
HadCRUT5, \url{https://www.metoffice.gov.uk/hadobs/hadcrut5/}.

\phantomsection\label{supplementary-material}
\bigskip

\begin{center}
{\large\bf SUPPLEMENTARY MATERIAL}
\end{center}

(A) Proofs of the Main Text; (B) Detailed Results on Simulation Studies; (C) The data and code used to reproduce the numerical studies are 
  contained in a zipped file named \texttt{reproducibility\_materials.zip}.

\bibliographystyle{apalike}
\setlength{\bibsep}{0pt}
\bibliography{cited}

@article{onatski2012asymptotics,
  title={Asymptotics of the principal components estimator of large factor models with weakly influential factors},
  author={Onatski, Alexei},
  journal={Journal of Econometrics},
  volume={168},
  number={2},
  pages={244--258},
  year={2012},
  publisher={Elsevier}
}

@article{passemier2017estimation,
  title={On estimation of the noise variance in high dimensional probabilistic principal component analysis},
  author={Passemier, Damien and Li, Zhaoyuan and Yao, Jianfeng},
  journal={Journal of the Royal Statistical Society: Series B (Statistical Methodology)},
  volume={79},
  number={1},
  pages={51--67},
  year={2017},
  publisher={Wiley Online Library}
}

@inproceedings{anderson1956statistical,
  title={Statistical inference in factor analysis},
  author={Anderson, Theodore W and Rubin, Herman},
  booktitle={Proceedings of the third Berkeley symposium on mathematical statistics and probability},
  volume={5},
  pages={111--150},
  year={1956}
}

@Article{JohnstoneP2018,
	author    = {Johnstone, Iain M and Paul, Debashis},
	title     = {{PCA} in high dimensions: An orientation},
	journal   = {Proceedings of the IEEE},
	year      = {2018},
	volume    = {106},
	number = {8},
	pages     = {1277--1292}
}

@article{li2024testing,
  title={Testing general linear hypotheses under a high-dimensional multivariate regression model with spiked noise covariance},
  author={Li, Haoran and Aue, Alexander and Paul, Debashis and Peng, Jie},
  journal={Journal of the American Statistical Association},
  volume={119},
  number={548},
  pages={2799--2810},
  year={2024},
  publisher={Taylor \& Francis}
}

@article{johnstone2001spiked,
  title={On the distribution of the largest eigenvalue in principal components analysis},
  author={Johnstone, Iain M},
  journal={The Annals of Statistics},
  volume={29},
  number={2},
  pages={295--327},
  year={2001},
  publisher={Institute of Mathematical Statistics}
}

@book{fuller2009measurement,
  title={Measurement error models},
  author={Fuller, Wayne A},
  year={2009},
  publisher={John Wiley \& Sons}
}

@article{bai1998no,
  title={No eigenvalues outside the support of the limiting spectral distribution of large-dimensional sample covariance matrices},
  author={Bai, Zhi-Dong and Silverstein, Jack W},
  journal={The Annals of Probability},
  volume={26},
  number={1},
  pages={316--345},
  year={1998},
  publisher={Institute of Mathematical Statistics}
}

@incollection{Bind:etal:dete:2013,
  author = {Bindoff, N. L. and Stott, P. A. and AchutaRao, K. M. and others},
  title     = {Detection and Attribution of Climate Change: From Global to Regional},
  booktitle = {Climate Change 2013: The Physical Science Basis},
  publisher = {Cambridge University Press},
  pages     = {867--952},
  year      = {2013},
  doi       = {10.1017/CBO9781107415324.022}
}

@incollection{eyring2021human,
  author    = {Eyring, V. and Gillett, N. P. and Achuta Rao, K. M. and others},
  title     = {Human Influence on the Climate System},
  booktitle = {Climate Change 2021: The Physical Science Basis},
  publisher = {Cambridge University Press},
  pages     = {423--552},
  year      = {2021}
}

@article{hegerl2007understanding,
  title={Understanding and attributing climate change},
  author={Hegerl, Gabriele C and Zwiers, Francis W and others},
  journal={Contribution of Working Group I to the Fourth Assessment Report of the Intergovernmental Panel on Climate Change (IPCC)},
  year={2007}
}

@article{hegerl1996detecting,
  title =	 {Detecting greenhouse-gas-induced climate change with
                  an optimal fingerprint method},
  author =	 {Hegerl, Gabriele C and von Storch, Hans and
                  Hasselmann, Klaus and Santer, Benjamin D and
                  Cubasch, Ulrich and Jones, Philip D},
  journal =	 {Journal of Climate},
  volume =	 9,
  number =	 10,
  pages =	 {2281--2306},
  year =	 1996
}

@article{Alle:Tett:chec:1999,
  author =	 {Allen, M. R. and Tett, S. F. B.},
  journal =	 {Climate Dynamics},
  pages =	 {419--434},
  title =	 {Checking for Model Consistency in Optimal
                  Fingerprinting},
  volume =	 15,
  year =	 1999,
}

@article{Alle:Stot:esti:2003,
  author =	 {Allen, M. R. and Stott, P. A.},
  title =	 {Estimating Signal Amplitudes in Optimal
                  Fingerprinting, Part {I}: {T}heory},
  journal =	 {Climate Dynamics},
  volume =	 21,
  year =	 2003,
  pages =	 {477--491},
  issue =	 5
}

@article{Ribe:Plan:Terr:appl:2013,
  title =	 {Application of Regularised Optimal Fingerprinting to
                  Attribution. {P}art {I}: {M}ethod, Properties and
                  Idealised Analysis},
  author =	 {Ribes, Aur{\'e}lien and Planton, Serge and Terray,
                  Laurent},
  journal =	 {Climate Dynamics},
  volume =	 41,
  number =	 {11-12},
  pages =	 {2817--2836},
  year =	 2013,
  publisher =	 {Springer}
}

@article{ribes2009adaptation,
  title =	 {Adaptation of the Optimal Fingerprint Method for
                  Climate Change Detection Using a Well-Conditioned
                  Covariance Matrix Estimate},
  author =	 {Ribes, Aur{\'e}lien and Aza{\"\i}s, Jean-Marc and
                  Planton, Serge},
  journal =	 {Climate Dynamics},
  volume =	 33,
  number =	 5,
  pages =	 {707--722},
  year =	 2009,
  publisher =	 {Springer}
}

@article{ledoit2004well,
  title =	 {A well-conditioned estimator for large-dimensional
                  covariance matrices},
  author =	 {Ledoit, Olivier and Wolf, Michael},
  journal =	 {Journal of Multivariate Analysis},
  volume =	 88,
  number =	 2,
  pages =	 {365--411},
  year =	 2004,
  publisher =	 {Elsevier}
}

@article{delsole2019,
  Title =	 {Confidence Intervals in Optimal Fingerprinting},
  Author =	 {Timothy DelSole and Laurie Trenary and Xiaoqin Yan
                  and Michael K. Tippett},
  Year =	 2019,
  Journal =	 {Climate Dynamics},
  volume =	 52,
  pages =	 {4111--4126}
}

@article{zhang2006multimodel,
  Title =	 {Multimodel Multisignal Climate Change Detection at
                  Regional Scale},
  Author =	 {Xuebin Zhang and Francis Zwiers and P. A. Stott},
  Journal =	 {Journal of Climate},
  Year =	 2006,
  volume =	 19,
  number =	 17,
  pages =	 {4294--4307}
}

@article{li2023regularized,
  author =        {Li, Yan and Chen, Kun and Yan, Jun and Zhang, Xuebin},
  journal =       {Annals of Applied Statistics},
  number =        {1},
  pages =         {225--239},
  title =         {Regularized Fingerprinting in Detection and
                   Attribution of Climate Change with Weight Matrix
                   Optimizing the Efficiency in Scaling Factor
                   Estimation},
  volume =        {17},
  year =          {2023},
  doi =           {10.1214/22-AOAS1624},
}

@article{li2021confidence,
  title =	 {Uncertainty in optimal fingerprinting is underestimated},
  author =	 {Li, Yan and Chen, Kun and Yan, Jun and Zhang,
                  Xuebin},
  journal =	 {Environmental Research Letters},
  year =	 2021,
  volume =       16,
  number =       8,
  pages =         {084043}
}

@article{katzfuss2017bayesian,
  title={A {B}ayesian Hierarchical Model for Climate Change Detection and Attribution},
  author={Katzfuss, Matthias and Hammerling, Dorit and Smith, Richard L},
  journal={Geophysical Research Letters},
  volume={44},
  number={11},
  pages={5720--5728},
  year={2017},
  publisher={Wiley Online Library}
}

@Article{hannart2016integrated,
  Title =	 {Integrated Optimal Fingerprinting: Method
                  Description and Illustration},
  Author =	 {Alexis Hannart},
  Journal =	 {Journal of Climate},
  Year =	 2016,
  volume =	 29,
  number =	 6,
  pages =	 {1977--1998}
}

@article{ma2023optimal,
  title={Optimal fingerprinting with estimating equations},
  author={Ma, Sai and Wang, Tianying and Yan, Jun and Zhang, Xuebin},
  journal={Journal of Climate},
  volume={36},
  number={20},
  pages={7109--7122},
  year={2023}
}

@article{morice2021updated,
  author =        {Morice, Colin P and Kennedy, John J and others},
  journal =       {Journal of Geophysical Research: Atmospheres},
  number =        {3},
  pages =         {e2019JD032361},
  publisher =     {Wiley Online Library},
  title =         {An updated assessment of near-surface temperature
                   change from 1850: {T}he {HadCRUT5} data set},
  volume =        {126},
  year =          {2021},
}

@article{eyring2016overview,
  author =        {Eyring, Veronika and Bony, Sandrine and others},
  journal =       {Geoscientific Model Development (Online)},
  pages =         {1937--1958},
  publisher =     {Lawrence Livermore National Lab.(LLNL), Livermore, CA
                   (United States)},
  title =         {Overview of the {C}oupled {M}odel {I}ntercomparison
                   {P}roject {P}hase 6 ({CMIP6}) Experimental Design and
                   Organization},
  volume =        {9},
  number =        {5},
  year =          {2016},
}

@article{paul2007asymptotics,
  title={Asymptotics of sample eigenstructure for a large dimensional spiked covariance model},
  author={Paul, Debashis},
  journal={Statistica Sinica},
  pages={1617--1642},
  year={2007},
  publisher={JSTOR}
}

@article{chen2024statistical,
  title={A statistical review on the optimal fingerprinting approach in climate change studies},
  author={Chen, Hanyue and Chen, Song Xi and Mu, Mu},
  journal={Climate Dynamics},
  volume={62},
  number={2},
  pages={1439--1446},
  year={2024},
  publisher={Springer}
}

@misc{li2025regularized,
      title={Regularized Fingerprinting with Linearly Optimal Weight Matrix in Detection and Attribution of Climate Change}, 
      author={Haoran Li and Yan Li},
      year={2025},
      eprint={2505.04070},
      archivePrefix={arXiv},
      primaryClass={stat.ME},
      url={https://arxiv.org/abs/2505.04070}, 
}

@article{fuller1980properties,
  title =	 {Properties of some estimators for the
                  errors-in-variables model},
  author =	 {Fuller, Wayne A},
  journal =	 {Annals of Statistics},
  volume =       8,
  pages =	 {407--422},
  year =	 1980,
  publisher =	 {JSTOR}
}

@article{pesta2012total,
  Title =	 {Total Least Squares and Bootstrapping with
                  Applications in Calibration},
  Author =	 {Michal Pesta},
  Year =	 2013,
  Journal =	 {Statistics},
  volume =	 47,
  number =	 5,
  pages =	 {966--991}
}

@article{jain2023importance,
  title={Importance of internal variability for climate model assessment},
  author={Jain, Shipra and Scaife, Adam A and Shepherd, Theodore G and Deser, Clara and Dunstone, Nick and Schmidt, Gavin A and Trenberth, Kevin E and Turkington, Thea},
  journal={npj Climate and Atmospheric Science},
  volume={6},
  number={1},
  pages={68},
  year={2023},
  publisher={Nature Publishing Group UK London}
}

@article{kritchman2008determining,
  title={Determining the number of components in a factor model from limited noisy data},
  author={Kritchman, Shira and Nadler, Boaz},
  journal={Chemometrics and Intelligent Laboratory Systems},
  volume={94},
  number={1},
  pages={19--32},
  year={2008},
  publisher={Elsevier}
}

@incollection{IPCC1990,
  author =        {{IPCC}},
  booktitle =     {The IPCC Scientific Assessment Contribution of
                   Working Group I to the First Assessment Report of the
                   Intergovernmental Panel on Climate Change},
  editor =        {Houghton, J. T. and Jenkins, G. J. and Ephraums, J. J.},
  pages =         {365},
  publisher =     {Cambridge University Press},
  title =         {Climate Change},
  year =          {1990},
}

@incollection{IPCC1995,
  author =        {{IPCC}},
  booktitle =     {Contribution of Working Group I to the Second
                   Assessment Report of the Intergovernmental Panel on
                   Climate Change},
  editor =        {Houghton, J. T. and Meira Filho, L. G. and
                   Callander, B. A. and Harris, N. and Kattenberg, A. and
                   Maskell, K.},
  pages =         {588},
  publisher =     {Cambridge University Press},
  title =         {Climate Change 1995 the science of climate change},
  year =          {1996},
}

@unpublished{chen2026valid,
  author = {Chen, Hanyue and Chen, Song Xi and Qiu, Jingkun},
  title  = {Valid Confidence Intervals for Detection and Attribution
            in Climate Change Studies},
  year   = {2026},
  note   = {Manuscript}
}

@article{monahan2009empirical,
  title={Empirical orthogonal functions: The medium is the message},
  author={Monahan, Adam H and Fyfe, John C and Ambaum, Maarten HP and Stephenson, David B and North, Gerald R},
  journal={Journal of Climate},
  volume={22},
  number={24},
  pages={6501--6514},
  year={2009}
}

@article{fasullo2020evaluation,
  title={Evaluation of leading modes of climate variability in the {CMIP} archives},
  author={Fasullo, John T and Phillips, AS and Deser, C},
  journal={Journal of Climate},
  volume={33},
  number={13},
  pages={5527--5545},
  year={2020}
}

@article{li2026adaptable,
  title={Adaptable Fingerprinting with Nonlinear Shrinkage for Climate Change Detection and Attribution under Variance Heterogeneity},
  author={Li, Haoran and Li, Yan},
  journal={arXiv preprint arXiv:2608.12496},
  year={2026}
}

@article{rateb2026structural,
  title={Structural deficits in large ensembles limit detection and attribution of terrestrial water storage},
  author={Rateb, Ashraf and R. Scanlon, Bridget and Buzzanga, Brett},
  journal={Communications Earth \& Environment},
  year={2026},
  publisher={Nature Publishing Group UK London}
}

\appendix

\clearpage
\setcounter{section}{0}
\renewcommand{\thesection}{S.\arabic{section}}
\setcounter{equation}{0}
\renewcommand{\theequation}{S.\arabic{equation}}
\setcounter{subsection}{0}
\renewcommand{\thesubsection}{S.\arabic{section}.\arabic{subsection}}
\setcounter{table}{0}
\renewcommand{\thetable}{S.\arabic{table}}
\setcounter{figure}{0}
\renewcommand{\thefigure}{S.\arabic{figure}}

\setcounter{theorem}{0}
\renewcommand\thetheorem{S.\arabic{section}.\arabic{theorem}}
\renewcommand\thelemma{S.\arabic{section}.\arabic{lemma}}

\setcounter{proposition}{0}
\renewcommand\theproposition{S.\arabic{proposition}}

\begin{center}
    {\Large\bf {Supplementary Material to\\ ``Fingerprint Analysis for Climate Change Detection and Attribution under a Latent Factor Model''}}
\end{center}

\section{Implementation Details}\label{sec:implement_details}

\subsection{Bias-Corrected Spike Estimation}

The bias-corrected estimation procedure for the spike strengths and
noise variance, following \citet{passemier2017estimation} and
\citet{li2024testing}, is summarized below.

\begin{algorithm}[h]
\caption{Adjusted estimation of $\pi_k$'s and $\sigma_z^2$}
\label{algo:estimation_spikes_var}
\begin{algorithmic}
\STATE {\bf Input:} $K$,  $\hat\gamma = N/m$, and ordered eigenvalues $\lambda_1(\bs{S})\geq\lambda_2(\bs{S}) \geq\cdots \geq \lambda_N(\bs{S})$ of $\bs{S}$.
\STATE {\bf Step 1:} Compute $\bs{\tilde{\sigma}}_z^2 = (N-K)^{-1} \sum_{j=K+1}^N \lambda_j(\bs{S})$ and $\bs{\tilde{\pi}}_k = \lambda_k (\bs{S}) - \bs{\tilde{\sigma}}_z^2$.
\STATE {\bf Step 2:} Utilizing \eqref{eq:bias_spikes}, for each $k= 1,\dots, K$, calculate a first-iteration estimator of $\pi_k$, say $\bs{\hat{\pi}}_k^{(1)}$, by solving the following equation
\[ \bs{\tilde{\pi}}_k = \bs{\tilde{\sigma}}_z^2 \psi\Big(\bs{\hat{\pi}}_k^{(1)}/ \bs{\tilde{\sigma}}_z^2, \hat\gamma\Big) = \bs{\hat{\pi}}_k^{(1)}+ \frac{\hat\gamma(\bs{\hat{\pi}}_k^{(1)}+ \bs{\tilde{\sigma}}_z^2 )}{\bs{\hat{\pi}}_k^{(1)} /\bs{\tilde{\sigma}}_z^2} .  \]

Note it is a quadratic equation of $\bs{\hat{\pi}}_k^{(1)}$. The equation has two real-valued roots if $\bs{\tilde{\pi}}_k/\bs{\tilde{\sigma}}_z^2 \geq \hat\gamma + 2\sqrt{\hat\gamma}$ and two complex-valued roots otherwise. We take $\bs{\hat{\pi}}_k^{(1)} = \max\{\operatorname{Re}(r_1), \operatorname{Re}(r_2)\}$, where $r_1$ and $r_2$ are the two roots and $\operatorname{Re}(r)$ is the real part of $r$.   
\STATE {\bf Step 3:} Denote $b_0 = K + \sum_{k=1}^K \bs{\tilde{\sigma}}_z^2/\bs{\hat{\pi}}_k^{(1)}$. The adjusted estimator of $\sigma^2_z$ is 
\[ \bs{\hat{\sigma}}_z^2 = \bs{\tilde{\sigma}}_z^2 \left(1+ \frac{\hat\gamma b_0}{N-K}\right).\]
\STATE {\bf Step 4:} Replace $\bs{\tilde{\sigma}}_z^2$ with $\bs{\hat{\sigma}}_z^2$ and repeat \textbf{Step 2} to solve for the second-iteration estimator $\bs{\hat{\pi}}_k^{(2)}$ of $\pi_k$. In case $\bs{\hat{\pi}}_k^{(2)} /\bs{\hat{\sigma}}_z^2 < \sqrt{\hat\gamma}$, set $\bs{\hat{\pi}}_k^{(2)} = \bs{\hat{\sigma}}^2_z \sqrt{\hat\gamma}$. 
\STATE {\bf Output:} The adjusted estimators $\bs{\hat{\pi}}_k = \bs{\hat{\pi}}_k^{(2)}$, $k =1,\dots, K$, and the noise variance $\bs{\hat{\sigma}}_z^2$. 
\end{algorithmic}
\end{algorithm}

\subsection{Closed-Form Expression for the TLS Estimator}

A closed-form expression for the TLS estimator is given in
Lemma~\ref{lemma:closed_form_TLS}.

\begin{lemma}
    \label{lemma:closed_form_TLS}
    The TLS estimator $\hat\beta(s_x, W)$ defined in \eqref{eq:TLS} of the main text admits a closed-form expression as 
    \[ \hat{\beta} = \Big[ \frac{1}{N} \bs{\bar{X}}^T W^{-1}\bs{\bar{X}} - \lambda(s_x^2) D \Big]^{-1} \frac{1}{N} \bs{\bar{X}}^T W^{-1} Y,\]
    where $\lambda(s_x^2)$ is the smallest eigenvalue of the matrix 
    \[ A (s_x^2) = \frac{1}{N} \begin{bmatrix}
        I_p & 0\\
        0 & s_x
    \end{bmatrix} 
    \begin{bmatrix}
        D^{-1/2} \bs{\bar{X}}^T \\
        \bs{Y}^T 
    \end{bmatrix} W^{-1}  
    \begin{bmatrix}
        \bs{\bar{X}} D^{-1/2} & \bs{Y}  
    \end{bmatrix}
    \begin{bmatrix}
        I_p & 0 \\ 0 & s_x
    \end{bmatrix}
    \] 
\end{lemma}

\section{Technical Proofs}
\label{sec:technical_proof}

In this section, we present the proofs of all technical results stated in the manuscript. Throughout the analysis, we use the notation $O_{\prec}(N^{a})$ and $o_{\prec}(N^{a})$
to denote a random matrix whose spectral norm is of stochastic order \(O_p(N^{a})\) and $o_p(N^{a})$, for any \(a\in\mathbb{R}\). The dimension of the matrix is determined implicitly by the surrounding context and may vary from line to line.

We start with a collection of technical lemmas. 
\begin{lemma}\label{lemma:expectation_quadratic_form}
Suppose that $Z_1 \sim N\left(0, I_N\right), Z_2 \sim N\left(0, I_N\right)$, and $Z_1$ is independent of $Z_2$. Suppose $D_1$ and $D_2$ are $N \times N$ deterministic symmetric matrices. Then,
$$
\begin{aligned}
& \mathbb{E} Z_1^T D_1 Z_2=0 \\
& \mathbb{E} Z_1^T D_1 Z_1=\operatorname{Tr}\left(D_1\right) \\
& \mathbb{E}\left(Z_1^T D_1 Z_2\right)\left(Z_1^T D_2 Z_2\right)=\operatorname{Tr}\left(D_1 D_2\right) \\
& \mathbb{E}\left(Z_1^T D_1 Z_1\right)\left(Z_1^T D_2 Z_1\right)=2 \operatorname{Tr}\left(D_1 D_2\right)+\operatorname{Tr}\left(D_1\right) \operatorname{Tr}\left(D_2\right)
\end{aligned}
$$  
\end{lemma}
The proof of this lemma is straightforward, and we omit the details.

\begin{lemma}[Lemma 2.7 of \cite{bai1998no}]
    \label{lemma:concentration_quadratic_form}
    Let $Z=\left(z_1, z_2, \ldots, z_p\right)^T$, where $z_i$'s are i.i.d. random variables with mean 0 and variance 1. Let $A$ be a deterministic matrix. Then for any $m \geq 2$, we have
$$
\mathbb{E}\left|Z^T A Z-\operatorname{Tr} A\right|^m \leqslant c_m\left(\mathbb{E} z_1^4 \operatorname{Tr} A A^T\right)^{m / 2}+c_m \mathbb{E} z_1^{2 m} \operatorname{Tr}\left[\left(AA^T\right)^{m / 2}\right],
$$
where $c_m$ is a constant only depending on $m$.
\end{lemma}

\begin{lemma}
    \label{prop:angle_tilde_U_U}
    Under our model assumptions, the empirical spiked eigenvectors and the true leading eigenvectors are such that 
    \begin{align*} 
    \bs{\widetilde{U}}^T U & =  \operatorname{Diag}\Big( \zeta \Big(\frac{\pi_1}{\sigma_z^2} , \gamma \Big), \dots, \zeta \Big(\frac{\pi_K}{\sigma_z^2} , \gamma \Big)\Big) + O_{\prec}\big(\frac{1}{\sqrt{N}}\big) \\
    & =\mc{Z} + O_{\prec}\big(\frac{1}{\sqrt{N}} \big).
    \end{align*}
\end{lemma}
The results follow directly from Result (3) of Theorem \ref{thm:asymptotics_eigenvalues_eigenvectors} and Proposition \ref{prop:consistency_adjusted_estimator}.

\begin{lemma}
    \label{lemma:converge_alpha}
    Under our model assumptions,  
    \[\alpha = s_x^2/s_z^2 + O_p\big( \cfrac{1}{ \sqrt{N\log N}}\big).\]
\end{lemma}
\begin{proof}
    
First of all, as
\[\bs{\tilde{X}}_{ij} \stackrel{iid}{\sim} N(X_i, s_x^2\Sigma), \quad j =1,\dots, n_i;~~ i =1,2,\dots, p,\]
we find that 
\[ \frac{1}{N}\frac{1}{\sum n_i -p} \sum_{i=1}^p \sum_{j=1}^{n_i} (\bs{\bar{X}}_{ij} - \bs{\bar{X}}_i)^T (\bs{\bar{X}}_{ij} - \bs{\bar{X}}_{i})\]
has the same distribution as 
\[  \frac{s_x^2}{n_T} \sum_{j=1}^{n_T} \frac{1}{N} \bs{Q}_j^T \Sigma \bs{Q}_j, \]
where $n_T = (\sum_{i} n_i) -p$ and $\bs{Q}_j \stackrel{iid}{\sim} N(0, I_N)$. By Lindeberg's central limit theorem, 
\[ \frac{1}{N} \bs{Q}_j^T \Sigma \bs{Q}_j = \frac{1}{N}\tr[\Sigma] + O_p(1/\sqrt{N}).\]
While $n_T = O(\log N)$, it follows then,
\[  \frac{s_x^2}{n_T} \sum_{j=1}^{n_T} \frac{1}{N} \bs{Q}_j^T \Sigma \bs{Q}_j = \frac{s_x^2}{N}\tr[\Sigma] + O_p(1/\sqrt{N\log N}).\]
Similarly, 
\[ \frac{1}{N} \frac{1}{m} \sum_{j=1}^m  \bs{Z}_j^T \bs{Z}_j = \frac{s_z^2}{N}\tr(\Sigma) + O_p(1/\sqrt{Nm}). \]
Then, 
\[ \alpha = \frac{m \sum_{i=1}^p \sum_{j=1}^{n_i} (\bs{\tilde{X}}_{ij}  - \bs{\bar{X}}_i)^T (\bs{\tilde{X}}_{ij} - \bs{\bar{X}}_i)}{ (\sum_{i=1}^p n_i - p) \sum_{j=1}^m \bs{Z}_j^T \bs{Z}_j} = \frac{s_x^2}{s_z^2} + O_p(1/\sqrt{N\log N}).\]

\end{proof}

\subsection{Proof of Lemma \ref{lemma:consistency_Delta_Phi}}\label{subsec:proof_lemma_consistency_Delta_Phi}

Recall that 
\[ W_\Theta =\bs{\widetilde{U}} \Theta \bs{\widetilde{U}}^T  + I_N   \qquad \text{and} \qquad  \Sigma = U \Lambda U^T + \sigma^2 I_N.\]
Using Proposition \ref{prop:consistency_adjusted_estimator} and Lemma \ref{prop:angle_tilde_U_U}, we obtain that 
\begin{align*}
    s_z^2 W_\Theta \Sigma &= s_z^2 \bs{\widetilde{U}} \Theta \mc{Z} \Lambda U^T +  s_z^2\sigma^2 \bs{\widetilde{U}} \Theta \bs{\widetilde{U}}^T +s_z^2 U \Lambda U^T + s_z^2 \sigma^2 I_N + O_{\prec}\big(\frac{1}{\sqrt{N}} \big) \\
    & = \bs{\widetilde{U}} \Theta \mc{Z} \Pi U^T + \bs{\hat{\sigma}}_z^2 \bs{\widetilde{U}} \Theta \bs{\widetilde{U}}^T + U \Pi U^T + \bs{\hat\sigma}^2_z I_N + O_{\prec}\big(\frac{1}{\sqrt{N}}\big).
\end{align*}
It follows that
\begin{align*}
\frac{s_z^2}{N} \tr [W_\Theta \Sigma] &=  \frac{1}{N} \tr[\Theta \mc{Z}\Pi  U^T \bs{\widetilde{U}}] + \bs{\hat{\sigma}}_z^2  \tr[\Theta] + \frac{1}{N}\tr [\Pi] + \bs{\hat{\sigma}}^2_z  + O_p \big(1/\sqrt{N}\big)\\
& = \frac{1}{N} \tr[\Theta (\Pi \mc{Z}^2 + \bs{\hat{\sigma}}_z^2 I_K)]  + \frac{1}{N}\tr[\Pi] + \bs{\hat{\sigma}}_z^2+ O_p(1/\sqrt{N})\\
&= \Delta_\Theta + O_p(1/\sqrt{N}).
\end{align*}
Second, consider $\Phi_\Theta$. 
\begin{align*}
    &s_z^4 W_\Theta \Sigma W_\Theta \Sigma\\
    & = s_z^4 \bs{\widetilde{U}} \Theta \mc{Z} \Lambda \mc{Z} \Theta \mc{Z} \Lambda U^T + s_z^4 \sigma^2 \bs{\widetilde{U}} \Theta \mc{Z} \Lambda \mc{Z} \Theta \bs{\widetilde{U}}^T + s_z^4 \bs{\widetilde{U}}\Theta \mc{Z}\Lambda^2 U^T + s_z^4 \sigma^2 \bs{\widetilde{U}}\Theta \mc{Z}\Lambda U^T \\
    &\quad + s_z^4 \sigma^2 \bs{\widetilde{U}} \Theta^2 \mc{Z}\Lambda U^T + s_z^4 \sigma^4 \bs{\widetilde{U}} \Theta^2 \bs{\widetilde{U}}^T + s_z^4 \sigma^2 \bs{\widetilde{U}}\Theta \mc{Z}\Lambda U^T + s_z^4 \sigma^4 \bs{\widetilde{U}}\Theta\bs{\widetilde{U}}^T \\
    &\quad + s_z^4 U \Lambda \mc{Z} \Theta \mc{Z}\Lambda U^T + s_z^4 \sigma^2 U\Lambda \mc{Z} \Theta \bs{\widetilde{U}}^T + s_z^4 U\Lambda^2 U^T + s_z^4 \sigma^2 U\Lambda U^T \\
    &\quad + s_z^4\sigma^2 \bs{\widetilde{U}} \Theta \mc{Z} \Lambda U^T +  s_z^4\sigma^4 \bs{\widetilde{U}} \Theta \bs{\widetilde{U}}^T +s_z^4\sigma^2 U \Lambda U^T + s_z^4 \sigma^4 I_N + O_{\prec}(1/\sqrt{N})\\
    & =\bs{\widetilde{U}} \Theta \mc{Z} \Pi \mc{Z} \Theta \mc{Z} \Pi U^T + \bs{\hat{\sigma}}^2_z\bs{\widetilde{U}} \Theta \mc{Z} \Pi \mc{Z} \Theta \bs{\widetilde{U}}^T + \bs{\widetilde{U}}\Theta \mc{Z}\Pi^2 U^T +  \bs{\hat{\sigma}}^2_z \bs{\widetilde{U}}\Theta \mc{Z}\Pi U^T\\
    &\quad +  \bs{\hat{\sigma}}^2_z \bs{\widetilde{U}} \Theta^2 \mc{Z} \Pi U^T +  \bs{\hat{\sigma}}^4_z \bs{\widetilde{U}} \Theta^2 \bs{\widetilde{U}}^T +  \bs{\hat{\sigma}}^2_z \bs{\widetilde{U}}\Theta \mc{Z}\Pi U^T + \bs{\hat{\sigma}}^4_z \bs{\widetilde{U}}\Theta\bs{\widetilde{U}}^T\\
    & \quad +  U \Pi \mc{Z} \Theta \mc{Z} \Pi U^T + \bs{\hat{\sigma}}^2_z U \Pi \mc{Z} \Theta \bs{\widetilde{U}}^T +  U\Pi^2 U^T + \bs{\hat{\sigma}}^2_z U \Pi U^T\\
    &\quad + \bs{\hat{\sigma}}^2_z \bs{\widetilde{U}} \Theta \mc{Z} \Pi U^T +  \bs{\hat{\sigma}}^4_z \bs{\widetilde{U}} \Theta \bs{\widetilde{U}}^T + \bs{\hat{\sigma}}^2_z U \Pi U^T +  \bs{\hat{\sigma}}^4_z I_N + O_{\prec}(1/\sqrt{N}).
\end{align*}
It follows that 
\begin{align*}
    \frac{s_z^4}{N} \tr [W_\Theta \Sigma W_\Theta \Sigma ]& =\frac{1}{N} \tr\big[\Theta^2 (\mc{Z}^2 \Pi + \hat\sigma_z^2 I_K )^2  \big]\\
    &+ \frac{2}{N} \tr\big[ \Theta ( \mc{Z}^2 \Pi^2 + 2\bs{\hat{\sigma}}_z^2 \mc{Z}^2 \Pi +\bs{\hat{\sigma}}_z^4 I_K ) \big] \\
    &+ \frac{1}{N}\tr[ (\Pi + \bs{\hat{\sigma}}_z^2 I_K )^2] + (1-K/N)\bs{\hat{\sigma}}_z^4 + O_p(1/\sqrt{N})\\
    & = \Phi_\Theta + O_p(1/\sqrt{N}).
\end{align*}

As for $G_\Theta$, note that we can write 
\[ \bs{\bar{X}} = X + s_x \Sigma^{1/2} \bs{Q} D^{1/2},\]
where $\bs{Q}$ is $N\times p$ matrix of i.i.d. $N(0,1)$ entries. Then, we decompose
\begin{align*} 
\frac{1}{N} \bs{\bar{X}}^T W_\Theta \bs{\bar{X}} &=\frac{1}{N} X^T W_\Theta X + \frac{1}{N} s_x X^T W_\Theta \Sigma^{1/2} \bs{Q} D^{1/2} + \frac{1}{N} s_x D^{1/2}\bs{Q}^T \Sigma^{1/2} W_\Theta X \\
&\quad + \frac{1}{N}  s_x^2 D^{1/2} \bs{Q}^T \Sigma^{1/2}W_\Theta \Sigma^{1/2} \bs{Q} D^{1/2}. 
\end{align*}
Using Lemma \ref{lemma:concentration_quadratic_form}, 
\begin{align*}
   & \frac{1}{N} s_x X^T W_\Theta \Sigma^{1/2} \bs{Q} D^{1/2} = O_\prec(1/\sqrt{N}),\\
   & \frac{s_x^2}{N} \bs{Q}^T \Sigma^{1/2} W_\Theta \Sigma^{1/2}\bs{Q}  = \frac{s_x^2}{N} \tr[W_\Theta \Sigma] I_p + O_{\prec}(1/\sqrt{N}). 
\end{align*}
Together with the convergence of $\alpha$ shown in Lemma \ref{lemma:converge_alpha}, it follows then
\begin{align*}
G_\Theta &= \frac{1}{N}\bs{\bar{X}}^T W_\Theta \bs{\bar{X}} - \alpha \Delta_\Theta D \\
   & =\frac{1}{N} \bs{\bar{X}}^T W_\Theta \bs{\bar{X}}  - \frac{s_x^2}{N}\tr [W_\Theta \Sigma ] D  + O_{\prec}(1/\sqrt{N}) \\ 
   & = \frac{1}{N}X^T W_\Theta X + O_{\prec}(1/\sqrt{N}).
\end{align*}

Lastly, we discuss $\Omega_\Theta$. We first analyze the projection of $X$ onto the leading spike directions. Following Result (3) of Theorem \ref{thm:asymptotics_eigenvalues_eigenvectors},  for each $k=1,\dots,K$, we can express 
\[\bs{\widetilde{U}}_k = U_k \mc{Z}_k + R_k,\] 
where $R_k$ is a $N$-dimensional random vector such that
\[ R_k = \sqrt{\frac{ 1- \zeta^2(\pi_k/\sigma_z^2, \gamma)}{N-K}} U_{\perp} w_k + r_k,\]
with $w_k$ being a standard normal random vector $w_k\sim N(0, I_{N-K})$ and $r_k$ is a negligible vector in the sense that for any $l >0$, 
\[\mE \|r_k\|^l_2 =  o(1).\] 
The representative is valid because a uniformly distributed random vector $\bs{s}_k$ on the unit sphere $\mathbb{S}^{N-K-1}$ can be re-expressed as 
\[ \bs{s}_k = w_k /\|w_k\|.\]
Note that $(N-K)^{-1/2} \|w_k\|_2 \stackrel{p}{\longrightarrow}1$. The residual $r_k$ comes from replacing $\|w_k\|_2$ by $\sqrt{N-K}$, the remainder term $N^{-1/2} \bs{\omega}_k U \bs{q}_k$ and the distance between $\bs{\omega}_k$ and $\zeta (\pi_k/\sigma_z^2, \gamma)$.

For each $j=1,\dots, p$, consider the projection of $X_j$ onto $\bs{\widetilde{U}}_k$ 
\begin{align*}
&\frac{1}{\sqrt{N}} \bs{\widetilde{U}}_k^T X_j = \frac{1}{\sqrt{N}}  \mc{Z} U^T_k X_j + \frac{1}{\sqrt{N}} R^T_k X_j \\
&=  \frac{1}{\sqrt{N}}  \mc{Z} U^T_k X_j + \frac{\sqrt{1-\zeta^2(\pi_k/\sigma_z^2,\gamma)}}{\sqrt{N(N-K)}} w^T_k X_j +\frac{1}{\sqrt{N}} r_k^T X_j .
\end{align*}
Since $N^{-1/2}\|X\|_2 <\infty$,  clearly,
\[ \frac{1}{\sqrt{N}} |r_k^T X_j| \leq \frac{1}{\sqrt{N}} \|r_k\|_2 \|X\|_2 = o_p(1).\]
Moreover, by Lemma \ref{lemma:concentration_quadratic_form},
\[ \mE \Big| N^{-1} w_k^T X_j\Big|^2 = N^{-2}\mE |w_k^T X_jX^T_j w_k| =O(N^{-2}\|X^TX\|_2)= O(N^{-1}).\]
It indicates that for each $j$ and $k$,
\[ \frac{1}{\sqrt{N}} \bs{\widetilde{U}}_k^T X_j = \frac{1}{\sqrt{N}} \mc{Z} U_k^T X_j + o_p(1).\]

Recall that 
\[ \widetilde{W} = \bs{\widetilde{U}} \mc{Z}^{-2} \Pi( \mc{Z}^{2} \Theta+  I_K )^2\bs{\widetilde{U}}^T + \bs{\hat{\sigma}}_z^2 W_\Theta^2 =  \bs{\widetilde{U}} \Pi \Big(  \mc{Z}^2 \Theta^2 + 2 \Theta + \mc{Z}^{-2} \Big)\bs{\widetilde{U}}^T + \bs{\hat{\sigma}}_z^2 W_\Theta^2. \]
We aim to show that  $N^{-1} X^T \widetilde{W} X$ concentrates around $ s_z^2 N^{-1} X^T W_\Theta \Sigma W_\Theta X$. For convenience, we collect $R = (R_1,\dots, R_K)$. Following our analysis,
\[ \frac{1}{\sqrt{N}} R^T X = o_{\prec}(1). \]
We can decompose  
\begin{align*}
    &s_z^2 N^{-1}X^T W_\Theta \Sigma W_\Theta X\\
    &= N^{-1}X^T (\bs{\widetilde{U}} \Theta \bs{\widetilde{U}}^T + I_N)  U\Pi U^T (\bs{\widetilde{U}}\Theta \bs{\widetilde{U}}^T + I_N)  X +\sigma_z^2  N^{-1}X^T W_{\Theta}^2  X\\
    &= N^{-1} X^T  \bs{\widetilde{U}} \Theta^2 \mc{Z}^2 \Pi \bs{\widetilde{U}}^T  X + N^{-1} X^T \bs{\widetilde{U}} \Theta \mc{Z} \Pi U^T  X \\
    &\qquad + N^{-1} X^T U \Theta \mc{Z} \Pi \bs{\widetilde{U}}^T  X+ N^{-1} X^T  U \Pi U^T X + \sigma_z^2 N^{-1} X^T  W_{\Theta}^2 X+ o_{\prec}(1)\\
    &=  N^{-1} X^T  \bs{\widetilde{U}} \Theta^2 \mc{Z}^2 \Pi \bs{\widetilde{U}}^T X  + N^{-1} X^T  \bs{\widetilde{U}} \Theta \Pi \Big(\bs{\widetilde{U}}^T  - R^T \Big) X \\
    &\qquad +  N^{-1} X^T  (\bs{\widetilde{U}} - R) \Theta \Pi \bs{\widetilde{U}}^T X +  N^{-1}X^T  (\bs{\widetilde{U}} - R)  \Pi \mc{Z}^{-2}  (\bs{\widetilde{U}}^T - R^T) X \\
    &\qquad + \sigma_z^2 N^{-1} X^T W_{\Theta}^2 X+ o_{\prec}(1)\\
    & = N^{-1} X^T \widetilde{W} X  + (\sigma_z^2-  \bs{\hat\sigma}_z^{2} ) N^{-1} X^T W_\Theta^2 X + N^{-1} X^T R \mc{Z}^{-2}\Pi R^T  X \\
    &\qquad - \Big(N^{-1}X^T \bs{\widetilde{U}} (\Theta + \mc{Z}^{-2} ) \Pi R^T X + N^{-1}X^T R \Pi (\Theta + \mc{Z}^{-2} ) \bs{\widetilde{U}}^T X \Big)+ o_{\prec}(1)\\
    & = N^{-1} X^T \widetilde{W} X + o_{\prec}(1).
\end{align*}

Next, we quantify the difference between $N^{-1}\bs{\bar{X}}^T \widetilde{W} \bs{\bar{X}}$ and $N^{-1}X^T \widetilde{W} X$. We decompose 
\begin{align*} 
\frac{1}{N} &\bs{\bar{X}}^T \widetilde{W} \bs{\bar{X}} =\frac{1}{N} X^T \widetilde{W} X + \frac{1}{N} s_x X^T \widetilde{W} \Sigma^{1/2} \bs{Q} D^{1/2} \\
&\quad + \frac{1}{N} s_x D^{1/2}\bs{Q}^T \Sigma^{1/2} \widetilde{W}  X + \frac{1}{N}  s_x^2 D^{1/2} \bs{Q}^T \Sigma^{1/2} \widetilde{W} \Sigma^{1/2} \bs{Q} D^{1/2} \\
& = \frac{1}{N} X^T \widetilde{W} X + \frac{1}{N}  s_x^2 D^{1/2} \bs{Q}^T \Sigma^{1/2} \widetilde{W} \Sigma^{1/2} \bs{Q} D^{1/2} + o_{\prec}(1)
\end{align*}
Here, we again use Lemma \ref{lemma:concentration_quadratic_form} to obtain the bound 
\[ \frac{1}{N} s_x X^T \widetilde{W} \Sigma^{1/2} \bs{Q} D^{1/2}  = O_{\prec}(1/\sqrt{N}). \]

We discuss the limiting behavior of $N^{-1} s_x^2 \bs{Q}^T \Sigma^{1/2} \widetilde{W} \Sigma^{1/2}\bs{Q}$. By Lemma \ref{lemma:concentration_quadratic_form}, 
\begin{align*}
&\frac{1}{N} s_x^2 \bs{Q}^T \Sigma^{1/2} \widetilde{W} \Sigma^{1/2}\bs{Q} = \frac{1}{N} s_x^2\tr(\widetilde{W}\Sigma) I_p + o_{\prec}(1)\\
&=N^{-1} (s_x^2 /s_z^2) \tr \Big(\widetilde{W} (U \Pi U^T + \bs{\hat\sigma}_z^2 I_N ) \Big) I_p + o_{\prec}(1)\\
&=\alpha \Psi_\Theta + o_{\prec}(1), 
\end{align*}
where
\begin{align*}
\Psi_\Theta& = \bs{\hat{\sigma}}_z^4 [1+ N^{-1}\tr (2\Theta + \Theta^2) ] + \frac{1}{N} \tr[\Pi^2 (\mc{Z}^2 \Theta + I_K)^2] \\
&+ \bs{\hat{\sigma}}_z^2 N^{-1}\tr[\Pi (\mc{Z}^2 \Theta^2 + 2\Theta + \mc{Z}^{-2})]+\bs{\hat{\sigma}}_z^2 N^{-1} \tr [ \Pi \{ I_K + \mc{Z}^2 (2\Theta + \Theta^2)\}].
\end{align*}
Then, 
\[\Omega_\Theta = \frac{1}{N}\bs{\bar{X}}^T \widetilde{W} \bs{\bar{X}} - \alpha \Psi_\Theta = \frac{1}{N}X^T \widetilde{W}X  + o_{\prec}(1) = \frac{s_z^2}{N} X^T W_\Theta \Sigma W_\Theta X + o_{\prec}(1).  \]

\subsection{Proof of Theorem \ref{thm:consistency_s_hat}}
\label{subsec:proof_theorem_consistency_s_hat}

Let $\hat{\beta}(s) = \hat{\beta}(s, W)$ be the TLS estimator with the weight matrix $W =W_\Theta$ and a general value of $s\in[0,\Upsilon]$. Using the closed-form formula of the TLS estimator shown in Lemma \ref{lemma:closed_form_TLS},
\[ \hat{\beta}(s) = \Big[  \frac{1}{N} \bs{\bar{X}}^T W \bs{\bar{X}} - \lambda(s^2) D \Big]^{-1} \frac{1}{N} \bs{\bar{X}}^T W\bs{Y},\] 
where $\lambda(s^2)$ is the smallest eigenvalue of the matrix
\[A(s^2) = \frac{1}{N} \begin{bmatrix}
    D^{-1/2} & 0 \\ 0 & s 
\end{bmatrix} 
\begin{bmatrix}
    \bs{\bar{X}}^T \\ \bs{Y}^T  
\end{bmatrix} W 
\begin{bmatrix}
    \bs{\bar{X}} & \bs{Y} 
\end{bmatrix}
\begin{bmatrix}
    D^{-1/2} & 0 \\ 0 & s
\end{bmatrix}.
\]
We study the behavior of $A(s^2)$. Notice that we can express $\bs{\bar{X}}$ and $\bs{Y}$ as 
\[ \bs{\bar{X}} = X + s_x \Sigma^{1/2} \bs{Q} D^{1/2}  \qquad \bs{Y} = X\beta + \Sigma^{1/2} \bs{e}.\]
Here, $\bs{Q} = [\bs{Q}_1 , \dots, \bs{Q}_p] \stackrel{iid}{\sim} N(0, I_N)$, $\bs{e} \sim N(0, I_N)$, and $\bs{Q}$ is independent of $\bs{e}$.  

Following Lemma \ref{lemma:consistency_Delta_Phi} and Lemma \ref{lemma:concentration_quadratic_form}, we can conclude that 
\[ A(s^2) = \begin{bmatrix}
    -D^{-1/2} \\ s \beta^T 
\end{bmatrix} \frac{1}{N} X^T W X 
\begin{bmatrix}
    -D^{-1/2} & s\beta 
\end{bmatrix}
+ \begin{bmatrix}
    s_x^2 \Delta I_p & 0 \\ 0 & s^2 \Delta 
\end{bmatrix} + O_{\prec}(1/\sqrt{N}).
\]
Here, $\Delta = \Delta_\Theta$ and $O_{\prec}(1/\sqrt{N})$ is a residual matrix whose spectral norm is of stochastic order $O_p(1/\sqrt{N})$ uniformly for $s\in [0, \Upsilon]$.  
Denote the sum of the leading two terms on the right-hand side to be $A_0(s^2)$. Further, denote the smallest eigenvalue of $A_0(s^2)$ to be $\lambda_0(s^2)$. 
We have that uniformly on $s\in [0, \Upsilon]$,
\[ \lambda(s^2)  - \lambda_0 (s^2) = O_p(1/\sqrt{N}). \]

It is straightforward to verify that when $N$ is sufficiently large, for any $s\in [0, \Upsilon]$, 
\[ \lambda_0(s^2) = s^2 \Delta.\]

Summarizing these results, we conclude that when $N$ is sufficiently large, uniformly on $[0,\Upsilon]$, 
\[\Big[\frac{1}{N} \bs{\bar{X}}^T W \bs{\bar{X}} - \lambda(s^2) D \Big]^{-1}   =  \Big[ \frac{1}{N} X^T W X - (s^2 - s_x^2) \Delta D \Big]^{-1} + O_{\prec}(1/\sqrt{N}).\]

On the other hand, consider the behavior of $N^{-1} \bs{\bar{X}}^T W \bs{Y}$. 
\begin{align*}
    \frac{1}{N} \bs{\bar{X}}^T W \bs{Y} &= \frac{1}{N} X^T W X \beta + \frac{1}{N} X^T W \Sigma^{1/2}\bs{e} + \frac{s_x}{N} D^{1/2}\bs{Q}^T\Sigma^{1/2} W X \beta + \frac{s_x}{N} D^{1/2}\bs{Q}^T \Sigma^{1/2} W \bs{e}.
\end{align*}
Again, using Lemma \ref{lemma:concentration_quadratic_form}, we can show that 
\[ \frac{1}{N} X^T W \Sigma^{1/2}\bs{e} = O_{\prec}(1/\sqrt{N}).\]
\[ \frac{s_x}{N} D^{1/2}\bs{Q}^T\Sigma^{1/2} W X \beta = O_{\prec}(1/\sqrt{N}).\]
\[ \frac{s_x}{N} D^{1/2}\bs{Q}^T \Sigma^{1/2} W \bs{e} = O_{\prec}(1/\sqrt{N}).\]
It follows then 
\[\frac{1}{N} \bs{\bar{X}}^T W \bs{Y} = \frac{1}{N}X^T W X \beta + O_{\prec}(1/\sqrt{N}).\]

Summarizing these results, we obtain the following lemma. 

\begin{lemma}
    \label{lemma:quantify_beta}
   Fix $W = W_\Theta$. Under our assumptions, for any $s\in [0, \Upsilon]$, define
    \[ \beta(s) = \Big[ \frac{1}{N}X^T W X - (s^2 - s_x^2) \Delta_\Theta D \Big]^{-1} \frac{1}{N} X^T W X\beta.\]
    We have 
    \[ \sup_{s\in [0,\Upsilon]} \|\hat{\beta}(s^2) - \beta(s^2) \|_2 = O_p(1/\sqrt{N}).\]
\end{lemma}

To show the convergence of $\hat{s}_x (W_\Theta)$, consider the loss function $L(s) = L(s, W_\Theta)$. First of all, under our assumptions, 
\[ \alpha \Delta_\Theta = \frac{s_x^2}{N} \tr (W_\Theta \Sigma) + O_p(1/\sqrt{N}).\]
Consider $\big\|W^{1/2} ( \bs{Y} - \bs{\bar{X}} \hat{\beta}(s)) \big\|_2^2$. We decompose it as 
\begin{align*}
    &W^{1/2} \Big(\bs{Y} - \bs{\bar{X}}^T \hat{\beta}(s) \Big)  = W^{1/2}\Big[ \bs{Y} - \bs{\bar{X}} \beta(s) \Big] + W^{1/2}\Big[ \bs{\bar{X}} \beta(s) - \bs{\bar{X}}\hat{\beta}(s) \Big]\\
    & = W^{1/2} \Big[  \bs{e} - s_x  \Sigma^{1/2} \bs{Q} D^{1/2} \beta(s) \Big] + W^{1/2}\bs{\bar{X}} \Big[  \beta(s) - \hat{\beta}(s)\Big] + W^{1/2} X \Big[ \beta - \beta(s) \Big] .
\end{align*}
Using Lemma \ref{lemma:concentration_quadratic_form}, the first term on the right-hand side is such that uniformly on $s\in [0, \Upsilon]$, 
\[\frac{1}{N} \Big\| W^{1/2} \Big(\Sigma^{1/2}\bs{e} - s_x \Sigma^{1/2} \bs{Q} D^{1/2} \beta(s)\Big) \Big\|_2^2 = \big(1 + s_x^2 \beta^T(s)  D \beta(s)\big) \frac{1}{N}\tr (W\Sigma) + O_p(1/\sqrt{N}). \]
The second term on the right-hand side is such that uniformly on $s\in [0,\Upsilon]$, 
\[\frac{1}{N} \Big\| W^{1/2} \bs{\bar{X}} [\beta(s) -\hat{\beta}(s)]\Big\|_2^2 = O_p(N^{-1}),\]
which is deduced from Lemma \ref{lemma:quantify_beta}.

Together, we conclude that uniformly on $s\in [0, \Upsilon]$,
\begin{align*}
    &\frac{1}{N} \Big\| W^{1/2} \Big(\bs{Y} - \bs{\bar{X}} \hat{\beta}(s) \Big)  \Big\|_2^2\\ = 
    &\big(1+s_x^2 \beta^T(s)  D\beta(s)\big) \frac{1}{N}\tr(W\Sigma)
    + \frac{1}{N}[\beta(s) -\beta]^T X^T WX [\beta(s) -\beta] + O_p(N^{-1/2}).
\end{align*}

Next, we show the convergence of the loss function $L(s)$. Define
\[ h_n(s) = \frac{s^2}{N} \frac{\|W^{1/2} (\bs{Y} - \bs{\bar{X}}\hat\beta(s) ) \|_2^2}{ 1+ s^2 \hat{\beta}^T(s) D \hat{\beta}(s)}  - \alpha \Delta_\Theta.\]
Define
{\small
\[ h(s) = \frac{s^2}{N} \frac{  \big(1+ s_x^2 \beta^T(s) D \beta(s) \big)N^{-1}\tr[W\Sigma] + N^{-1} [\beta(s) -\beta ]^T X^T W X [\beta(s) -\beta] }{  1+ s^2 \beta^T(s) D \beta(s) } - \frac{s_x^2}{N}\tr[W\Sigma].\]
}

Due to the characterization of the limiting behavior of $\|W^{1/2} (\bs{Y}- \bs{\bar{X}} \hat{\beta}(s))\|_2^2$, we conclude that uniformly on $s\in[0,\Upsilon]$,
\[ h_n(s) - h(s) \stackrel{p}{\longrightarrow}0.\]

Next, we consider the properties of $h(s)$. Clearly, $h(s_x) = 0$. For all sufficiently large $N$, $s_x$ is the unique root to $h(s_x) = 0$.

Let $\hat{s}_x$ be a minimizer of $h_n(s)$ on $[0,\Upsilon]$. Due to the convergence of $h_n(s)$ to $h(s)$, we can conclude that 
\[ \hat{s}_x - s_x \stackrel{p}{\longrightarrow} 0.\]

It remains to verify that the convergence rate is $O_p(1/\sqrt{N})$. To this end, note that there exists a neighborhood of $s_x$ such that within the neighborhood, $|h'(s)|> C>0$ for some constant $C$. On the other hand, since $h(\Upsilon)>0$ for all sufficiently large $N$ when $ \Upsilon> s_x$, we conclude that 
\[ \mathbb{P} \Big(  \frac{\Upsilon^2}{N}\frac{\|W^{1/2} (\bs{Y} -\bs{\bar{X}}\hat\beta(\Upsilon) ) \|_2^2}{1+ \Upsilon^2 \hat{\beta}^T(\Upsilon) D\hat{\beta}(\Upsilon) } >\alpha \Delta_\Theta \Big) \longrightarrow 1,\quad \text{as }N\to \infty. \]
Therefore, with high probability, the solution $\hat{s}_x$ on $[0, \Upsilon]$ is obtained at  
\[  h_n(\hat{s}_x) =  0.\]
It follows then
\[ h(\hat{s}_x) - h(s_x) = h'(s^*) (\hat{s}_x - s_x) = h_n(\hat{s}_x) - h(s_x) + O_p(N^{-1/2}) = O_p(N^{-1/2}), \]
for some $s^*$ in between $\hat{s}_x$ and $s_x$. While $|h'(s^*)|>C>0$, we conclude then
\[ \hat{s}_x - s_x = O_p(N^{-1/2}).\]
Lastly, the convergence of $\hat{s}_z$ follows from the fact that $\alpha = s_x^2/s_z^2 + O_p(1/\sqrt{N\log N})$.

\subsection{Proof of Theorem \ref{thm:normality}}\label{subsec:proof_normality}
Recall that Theorem \ref{thm:normality} states the asymptotic normality of the TLS estimator $\hat{\beta}(\hat{s}_x, W_\Theta)$ with the estimated value $\hat{s}_x$. The proof proceeds in two steps. We first show that the TLS estimator $\hat{\beta}(s_x, W_\Theta)$ with the true value $s_x$ is asymptotically normal. Secondly, we quantify the difference between $\hat{\beta}(\hat{s}_x, W_\Theta)$ and $\hat{\beta}(s_x, W_\Theta)$. 

Recall the formulation of $\bs{\bar{X}}$ and $\bs{Y}$ as 
\[\bs{\bar{X}} = X + s_x \Sigma^{1/2} \bs{Q} D^{1/2} \qquad \text{and}\qquad \bs{Y} = X\beta + \Sigma^{1/2}\bs{e}.\]
Consider the eigendecomposition of $W^{1/2}_\Theta \Sigma W^{1/2}_\Theta = \Gamma A \Gamma^T$, where $A = \operatorname{Diag}(a_1, a_2, \dots, a_N)$ is the diagonal matrix of eigenvalues, and $\Gamma$ is the corresponding matrix of eigenvectors. Let 
\begin{align*}
    &\bs{Y}^* = \Gamma^T W^{1/2}_\Theta \bs{Y},  \qquad X^* = s_x^{-1} \Gamma^T W^{1/2}_\Theta X D^{-1/2}, \qquad \bs{e}^* = \Gamma^T W^{1/2}_\Theta \Sigma^{1/2}\bs{e},\\
    &\bs{\bar{X}}^* = s_x^{-1} \Gamma^T W^{1/2}_\Theta \bs{\bar{X}} D^{-1/2}, \qquad \bs{V}^* =  \Gamma^T W^{1/2}_\Theta \Sigma^{1/2} \bs{Q},\qquad \beta^* = s_x D^{1/2}\beta. 
\end{align*}
With the definition, we have 
\[ \bs{Y}^* = X^* \beta^* + \bs{e}^*\qquad \text{and} \qquad \bs{\bar{X}}^* = X^* + \bs{V}^*.\]
\[ \frac{\| W^{1/2}_\Theta ( \bs{Y} - \bs{\bar{X}}\beta )  \|_2^2}{  1+ s_x^2 \beta^T D\beta} = \frac{\| \bs{Y}^* - \bs{\bar{X}}^*\beta^* \|_2^2 }{ 1+ {\beta^*}^T  \beta^*}.\]
Clearly, if $\hat\beta^*$ is the minimizer of the right-hand side, the TLS estimator is $\hat{\beta}(s_x) = s_x^{-1} D^{-1/2}\hat{\beta}^*$. 

Taking the derivative of the objective function with respect to $\beta^*$, we find that the minimizer solves the score equation
\[  S(\hat{\beta}^*) = \frac{ (\bs{\bar{X}}^*)^T (\bs{Y}^* -\bs{\bar{X}}^* \hat{\beta}^*)}{N} + \hat{\beta}^* \frac{ \|\bs{Y}^* - \bs{\bar{X}}^* \hat{\beta}^* \|_2^2}{ N (1+ (\hat{\beta}^*)^T  \hat{\beta}^*)} =0, \]
where $S(\beta) =  (S_1(\beta), S_2(\beta), \dots, S_p(\beta))^T$ is a $p$-dimensional vector.

By Taylor's theorem, there exists a series of $\tilde{\beta}_j^*$ on the line segment between $\hat{\beta}_j^*$ and $\beta_j^*$ for $j=1,2,\dots, p$, such that 
\[ S(\hat{\beta}^*) = S(\beta^*) + H (\hat{\beta}^* - \beta^*) = 0, \]
where $H = \nabla S(\tilde{\beta}^*)$ is the $p\times p $ Hessian matrix at $\tilde{\beta}^* = (\tilde{\beta}_1^*, \tilde{\beta}_2^*,  \dots, \tilde{\beta}_p^* )^T$. 

It follows that 
\begin{align*}
H \sqrt{N} (\hat{\beta}^*  -\beta^*) &= -\sqrt{N} S(\beta^*) \\
& = -\frac{1}{\sqrt{N}} \sum_{i=1}^N \Big\{ (\bs{e}_i^* -  (\bs{v}_i^*)^T \beta^*)(x_i^* + \bs{v}_i^* )  + \beta^* \frac{ (\bs{e}_i^* - (\bs{v}_i^*)^T \beta^*)^2  }{ 1+  (\beta^*)^T\beta^*}   \Big\}\\
& \coloneqq -\frac{1}{\sqrt{N}} \sum_{i=1}^N M_i(\beta^*),
\end{align*}
where for $i=1,\dots, N$, $\bs{e}_i^*$ is the $i$th element of $\bs{e}^*$, $\bs{v}_i^*$ is the $i$th row vector of the matrix $\bs{V}^*$, and $x_i^*$ is the $i$th row vector of the matrix $X^*$.

From the normality of $\bs{e}$ and $\bs{Q}$, it is straightforward that $\bs{e}_i^* \mid W_\Theta  \sim N(0, a_i)$, $\bs{v}_i^*\mid W_\Theta \sim N(0, a_i I_p)$, and $M_i(\beta^*) \mid W_\Theta$ are mutually independent vectors with finite (conditional) covariance matrices: 
\begin{align*}
     &\operatorname{cov}(M_i(\beta^*) \mid W_\Theta)\\
     &=  a_i  x_i^* (x_i^*)^T ( 1+ (\beta^*)^T \beta^*) + a_i^2 (I_p  + I_p (\beta^*)^T \beta^* + 2 \beta^* (\beta^*)^T) - 3 a_i^2 \beta^* (\beta^*)^T\\
     & = [ a_i x_i^* (x_i^*)^T + a_i^2 (I_p + \beta^* (\beta^*)^T)^{-1} ] ( 1+ (\beta^*)^T\beta^*). 
\end{align*}

It follows that 
\begin{align*}
    \mathfrak{M} \coloneqq \frac{1}{N} \sum_{i=1}^N \operatorname{cov}( M_i(\beta^*) \mid W_\Theta) = \Big\{  \frac{1}{N} s_x^{-2} D^{-1/2} X^T W_\Theta \Sigma W_\Theta X D^{-1/2}\\
    + \frac{1}{N} \tr\Big[ W_\Theta \Sigma W_\Theta \Sigma\Big] (I_p + s_x^2 D^{1/2} \beta \beta^T D^{1/2})^{-1}\Big\} (1+ s_x^2 \beta^T D\beta).  
\end{align*}
Moreover, it is straightforward to show that $\mE (M_i*(\beta^*) \mid W_\Theta) = 0$. 

Using the Lindeberg-Feller central limit theorem, we obtain that for any vector $a\in \mathbb{R}^p$ and any $t\in\mathbb{R}$, 
\[   \mathbb{P}\Big( \frac{1}{\sqrt{N a^T \mathfrak{M} a }} \sum_{i=1}^N a^T M_i(\beta^*) \leq t \Big) \longrightarrow \Phi(t), \quad \text{as } N\to\infty, \]
where $\Phi(\cdot)$ denotes the cumulative distribution function of the standard normal distribution $N(0,1)$. Here, the Lindeberg condition is satisfied due to the existence of higher-order moments of $\bs{e}_i\mid W_\Theta$ and $\bs{v}_i^* \mid W_\Theta$ and the boundedness of  $\mathfrak{M}$. 

All together, we conclude that 
\[ \frac{1}{\sqrt{N  a^T  \mathfrak{M} a }} \sum_{i=1}^N a^T M_i(\beta^*) \stackrel{d}{\longrightarrow} N(0,1). \]

Moreover, by Lemma \ref{lemma:consistency_Delta_Phi},
\begin{align*}
    \mathfrak{M} = (1  + s_x^2 \beta^T D \beta)\Big( s_x^{-2}s_z^{-2} D^{-1/2}\Omega_\Theta D^{-1/2} + s_z^{-4} \Phi_\Theta ( I_p + s_x^2 D^{1/2} \beta\beta^T D^{1/2} )^{-1}\Big) + o_{\prec}(1). 
\end{align*}

Next, consider the Hessian at $\tilde{\beta}^*$. At any value $\beta_0$, 
\begin{align*}
    \nabla S(\beta_0) &= -\frac{ (\bs{\bar{X}}^*)^T \bs{\bar{X}}^*}{N} + \frac{\|\bs{Y}^*  -\bs{\bar{X}}^* \beta_0 \|_2^2  }{N ( 1+ \beta^T_0\beta_0 )} I_p + \beta_0 \frac{\partial \|  \bs{Y}^* - \bs{\bar{X}}^*\beta_0 \|_2^2/ \{  N(1+\beta_0^T \beta_0) \}  }{\partial \beta_0^T }\\
    & = -\frac{(\bs{\bar{X}}^*)^T \bs{\bar{X}}^*}{N} + \frac{\| \bs{Y}^* - \bs{\bar{X}}^* \beta_0 \|_2^2 }{ N (1+\beta^T_0\beta_0)}I_p + \frac{\beta_0^T \{S(\beta_0) \}^T }{1+\beta^T_0\beta_0}.
\end{align*}

Since as $N\to\infty$ and $\beta_0 \to \hat{\beta}^*$, we have $S(\beta_0) \to 0$. Therefore, the Hessian at $\tilde{\beta}^*$ is such that 
\begin{align*}
    H &= -\frac{ (\bs{\bar{X}}^*)^T \bs{\bar{X}}^*  }{N} + \frac{ \|\bs{Y}^* - \bs{\bar{X}}^* \tilde{\beta}^* \|_2^2 }{N (1 + (\tilde{\beta}^*)^T \tilde{\beta}^* )} + o_{\prec}(1)\\
      &= -\frac{1}{N} s_x^{-2} D^{-1/2}X^T W_\Theta  X D^{-1/2} - \frac{1}{N} \tr[ W_\Theta \Sigma] I_p + \frac{1}{N} \tr[W_\Theta \Sigma] I_p \\
      &\qquad + \frac{ (\tilde{\beta}^* -\beta^*)^T N^{-1} (X^*)^T W_\Theta \Sigma W_\Theta X^* (\tilde{\beta}^* -\beta^*)   }{1+ (\tilde{\beta}^*)^T\tilde{\beta}^* } I_p + o_{\prec}(1)\\
      & = -s_x^{-2} D^{-1/2}G_\Theta D^{-1/2} + o_{\prec}(1).
\end{align*}

In summary, due to Slutsky's theorem, we conclude that 
\[ \sqrt{N} \Big\{ s_x^4 D^{1/2} G_\Theta^{-1} D^{1/2} \mathfrak{M} D^{1/2} G_\Theta^{-1} D^{1/2}\Big\}^{-1/2} (\hat{\beta}^* -\beta^*) \stackrel{d}{\longrightarrow} N(0,I_p).\] 

Using the convergence of $\hat{s}_x$ and $\hat{s}_z$ shown in Theorem \ref{thm:consistency_s_hat}, we conclude that 
\[  \sqrt{N} \Xi_\Theta^{-1/2} (\hat{\beta}(s_x) - \beta) \stackrel{d}{\longrightarrow}N(0, I_p).\]
It complete the proof of asymptotic normality of the TLS estimator  $\hat{\beta}(s_x)$ using the true multiplier $s_x$.

It remain to quantify the difference induced by the replacement of $s_x$ by $\hat{s}_x$. Recall that $\hat{s}_x = s_x+ O_p(1/\sqrt{N})$. Using Lemma \ref{lemma:closed_form_TLS},
the difference is such that 
\begin{align*}
    &\sqrt{N} [ \hat{\beta}(\hat{s}_x)  - \hat{\beta}(s_x)]  \\
    & =\sqrt{N} [\lambda(\hat{s}_x) - \lambda(s_x)] \Big[ \frac{1}{N} \bs{\bar{X}}^T W_\Theta \bs{\bar{X}} - \lambda (\hat{s}_x) D \Big]^{-1} D \Big[\frac{1}{N} \bs{\bar{X}}^T W_\Theta \bs{\bar{X}} - \lambda(s_x) D \Big]^{-1} \frac{1}{N}\bs{\bar{X}}^T W_\Theta \bs{Y}.  
\end{align*}

It is straightforward that in a neighborhood of $s_x$, we have $\lambda'(s_x) = 1$. Therefore, 
\begin{align*}
    \sqrt{N} [\lambda(\hat{s}) -\lambda (s_x)] = O_p(\sqrt{N} [\hat{s}_x - s_x]), \\
    \sqrt{N}[\hat{\beta}(\hat{s}_x) -\hat{\beta}(s_x)] = O_p(\|D\|_2) = O_p(1/\log N). 
\end{align*}
It follows then 
\[ \sqrt{N} \Xi_\Theta^{-1/2} (\hat{\beta}(\hat{s}_x) - \beta) \stackrel{d}{\longrightarrow} N(0,I_p).\]

\subsection{Proof of Theorem \ref{thm:adequacy_test}}\label{subsec:proof_adequacy}

Consider the decomposition
\[ \hat{\epsilon} = \bs{Y} -\bs{\bar{X}} \hat{\beta} = X\beta + \Sigma^{1/2}\bs{e} - X\hat\beta  - s_x \Sigma^{1/2} \bs{Q}D^{1/2}\hat{\beta} .\]
It follows then
$$
\begin{aligned}
\frac{1}{N} & \hat{\epsilon}^T W_\Theta \hat{\epsilon}-\frac{1}{N} \operatorname{tr}[W_\Theta \Sigma]-s_x^2 \hat{\beta}^T D \hat{\beta} \frac{1}{N} \operatorname{tr}[W_\Theta \Sigma] \\
= & \frac{s_x^2}{N} \hat{\beta}^T D^{1 / 2}\left[\bs{Q}^T \Sigma^{1 / 2} W_\Theta \Sigma^{1 / 2} \bs{Q} -\operatorname{tr}[W_\Theta \Sigma] I_p\right] D^{1 / 2} \hat{\beta} \\
& +\frac{1}{N}\left[\bs{e}^T \Sigma^{1 / 2} W_\Theta \Sigma^{1 / 2} \bs{e} -\operatorname{tr}[W_\Theta \Sigma]\right] \\
& -\frac{2 s_X}{N} \bs{e}^T \Sigma^{1 / 2} W_\Theta \Sigma^{1 / 2} \bs{Q} D^{1 / 2} \hat{\beta} \\
& +\frac{1}{N}(\hat{\beta}-\beta)^T X^T W_\Theta X(\hat{\beta}-\beta) \\
& +\frac{2}{N} \bs{e}^T \Sigma^{1 / 2} W_\Theta X (\beta-\hat{\beta}) \\
& +\frac{2 s_x}{N} \hat{\beta}^T D^{1 / 2} \bs{Q}^T \Sigma^{1 / 2} W_\Theta X(\hat{\beta}-\beta)
\end{aligned}
$$

Since $N^{-1} X^T W_\Theta X = O_{\prec}(1)$ and from Theorem \ref{thm:normality}, $\hat{\beta} -\beta = O_p(N^{-1/2})$, we have 
\[\frac{1}{N}(\hat{\beta} -\beta)^T X^T W_\Theta X (\hat{\beta} - \beta) = O_p(N^{-1}).\]
Also, due to Lemma \ref{lemma:concentration_quadratic_form}, 
\[ \frac{2}{N} \bs{e}^T \Sigma W_\Theta X X^T W_\Theta \Sigma^{1/2}\bs{e}  = O_p(1). \]
It follows that 
\[ \frac{1}{N} \bs{e}^T \Sigma^{1/2} W_\Theta X (\hat{\beta} -\beta) = O_p(N^{-1}).\]
Due to analogous arguments, 
\[ \frac{s_x}{N} \hat{\beta}^T D^{1/2} \bs{Q}^T \Sigma^{1/2} W_\Theta  X (\hat{\beta} -\beta) = O_p(N^{-1}).\]
Moreover, due to Lemma \ref{lemma:consistency_Delta_Phi}, 
\[ \frac{1}{N} \bs{Q}^T \Sigma^{1/2} W_\Theta \Sigma^{1/2} \bs{Q} - \frac{1}{N} \tr[W_\Theta \Sigma] I_p = O_{\prec}(N^{-1/2}).\]
It follows then 
\begin{align*}
    \frac{s_x^2}{N} \hat{\beta}^T D^{1/2} \Big[  \bs{Q}^T \Sigma^{1/2} W_\Theta \Sigma^{1/2} \bs{Q} -\tr[W_\Theta \Sigma ] I_p \Big]D^{1/2} \hat{\beta} \\
    - \frac{s_x^2}{N} \beta^T D^{1/2} \Big[ \bs{Q}^T \Sigma^{1/2} W_\Theta \Sigma^{1/2} \bs{Q} - \tr[W_\Theta \Sigma] I_p \Big]D^{1/2}\beta = O_p(N^{-1}).
\end{align*}

Also, 
\[ \frac{1}{N}\bs{e}^T \Sigma^{1/2} W_\Theta \Sigma^{1/2} \bs{Q} D^{1/2}\hat{\beta} - \frac{1}{N} \bs{e}^T \Sigma^{1/2} W_\Theta \Sigma^{1/2} \bs{Q} D^{1/2}\beta = O_p(N^{-1}). \]
We conclude that 
\begin{align*}
   & \frac{1}{N} \hat{\epsilon}^T W_\Theta \hat{\epsilon} - (1+s_x^2 \hat{\beta}^T D\hat{\beta}) \frac{1}{N} \tr[W_\Theta \Sigma]\\
    & = \frac{1}{N} [ \bs{e} -s_x \bs{Q} D^{1/2} \beta]^T \Sigma^{1/2}W_\Theta \Sigma^{1/2} [\bs{e} - s_x \bs{Q}D^{1/2}\beta]\\
    &\quad  - \frac{1}{N} \mE [ \bs{e} -s_x \bs{Q} D^{1/2} \beta]^T \Sigma^{1/2}W_\Theta \Sigma^{1/2} [\bs{e} - s_x \bs{Q}D^{1/2}\beta] + o_p(N^{-1/2}).  
\end{align*}
Clearly, $\bs{e} - s_x \bs{Q} D^{1/2}\beta $ is a $N\times 1$ vector whose entries are iid $N(0, 1+ s_x^2 \beta^T D \beta)$. Using Lindeberg-Fuller central limit theorem, 
\[ \sqrt{N} \frac{N^{-1} \hat{\epsilon}^T W_\Theta \hat{\epsilon} - (1+ s_x^2\beta^T D\beta) N^{-1}\tr[W_\Theta \Sigma]}{ \sqrt{2 (1+ s_x^2 \beta^T D\beta)^2 N^{-1}\tr[W_\Theta \Sigma W_\Theta \Sigma] } }  \stackrel{d}{\longrightarrow}N(0,1).\]

On the other hand, we have 
\[ \frac{s_z^2}{N}\tr[W_\Theta \Sigma ] = \Delta_\Theta +O_p(N^{-1/2}).\]
\[ \frac{s_z^4}{N} \tr[W_\Theta \Sigma W_\Theta \Sigma] = \Phi_\Theta + o_p(1).\]
\[\hat{s}_x - s_x = O_p(N^{-1/2}).\]
\[ \hat{\beta} -\beta = O_p(N^{-1/2}).\]
Therefore, using Slutsky's theorem, we conclude that 
\[ \frac{Q(\hat{\epsilon})}{\sqrt{V_Q}} \stackrel{d}{\longrightarrow}N(0,1),  \quad \text{where}\quad V_\Theta = \frac{2}{N} (1+\hat{s}_x^2 \hat{\beta}^TD\hat{\beta})^2 \Phi_{\Theta}. \]

\section{Detailed Simulation Results}\label{sec:add_results}

This section presents the complete simulation results that complement the selected findings reported in the main text. The additional results include both regression coefficients, $\beta_1$ and $\beta_2$, for $K=3$ and $K=5$, under Gaussian and heavy-tailed $t$ errors. For each configuration, we report the empirical bias and RMSE of the point estimators, together with the empirical coverage probability and average length of the 90\% confidence intervals.

Tables~\ref{tab:bias_rmse_case_normal_beta1_K3}--\ref{tab:bias_rmse_case_t_beta2_K5}
summarize the empirical bias and RMSE, and
Tables~\ref{tab:coverage_length_case_normal_beta1_K3}--\ref{tab:coverage_length_case_t_beta2_K5}
summarize the corresponding confidence interval performance. The results are reported across dimensions, ensemble sizes, signal-to-noise ratios, spike detection probabilities, and variance inflation levels. Overall, the additional simulations are consistent with the findings in the main text.

\begin{table}[htbp]
\scriptsize
\setlength{\tabcolsep}{2.4pt}
\renewcommand{\arraystretch}{0.8}
\centering
\caption{Empirical bias (RMSE) of point estimators $\hat{\beta}_1$ under Normal errors with $K=3$.}
\label{tab:bias_rmse_case_normal_beta1_K3}
\begin{tabular}{@{}llrrrr@{\hspace{7pt}}rrrr@{}}
\toprule
 &  & \multicolumn{4}{c}{$N=260$} & \multicolumn{4}{c}{$N=520$} \\
\cmidrule(lr){3-6}\cmidrule(lr){7-10}
 &  & \multicolumn{2}{c}{SDP = 0.2} & \multicolumn{2}{c}{SDP = 0.5} & \multicolumn{2}{c}{SDP = 0.2} & \multicolumn{2}{c}{SDP = 0.5} \\
\cmidrule(lr){3-4}\cmidrule(lr){5-6}\cmidrule(lr){7-8}\cmidrule(lr){9-10}
$m$ & Method & SNR = 2 & SNR = 5 & SNR = 2 & SNR = 5 & SNR = 2 & SNR = 5 & SNR = 2 & SNR = 5 \\
\midrule
\multicolumn{10}{c}{Case 1: $s^2_x = s^2_z=1$} \\
\multirow{4}{*}{100} & SOF & -0.002 (0.088) & -0.000 (0.057) & 0.003 (0.121) & 0.003 (0.072) & 0.004 (0.070) & -0.001 (0.043) & 0.001 (0.089) & -0.000 (0.055)
\\
 & ROF & -0.002 (0.093) & -0.000 (0.060) & 0.003 (0.127) & 0.003 (0.076) & 0.004 (0.073) & -0.001 (0.046) & 0.001 (0.092) & -0.001 (0.057)
\\
 & EOF & -0.010 (0.089) & -0.004 (0.058) & -0.009 (0.121) & -0.002 (0.072) & -0.002 (0.070) & -0.004 (0.044) & -0.007 (0.089) & -0.004 (0.055)
\\
 & Oracle & -0.002 (0.085) & -0.001 (0.054) & 0.000 (0.111) & 0.003 (0.067) & 0.003 (0.064) & -0.001 (0.040) & -0.000 (0.079) & -0.002 (0.050)
\\
\bottomrule
\end{tabular}
\end{table}

\begin{table}[htbp]
\scriptsize
\setlength{\tabcolsep}{2.4pt}
\renewcommand{\arraystretch}{0.8}
\centering
\caption{Empirical bias (RMSE) of point estimators $\hat{\beta}_1$ under Normal errors with $K=5$.}
\label{tab:bias_rmse_case_normal_beta1_K5}
\begin{tabular}{@{}llrrrr@{\hspace{7pt}}rrrr@{}}
\toprule
 &  & \multicolumn{4}{c}{$N=260$} & \multicolumn{4}{c}{$N=520$} \\
\cmidrule(lr){3-6}\cmidrule(lr){7-10}
 &  & \multicolumn{2}{c}{SDP = 0.2} & \multicolumn{2}{c}{SDP = 0.5} & \multicolumn{2}{c}{SDP = 0.2} & \multicolumn{2}{c}{SDP = 0.5} \\
\cmidrule(lr){3-4}\cmidrule(lr){5-6}\cmidrule(lr){7-8}\cmidrule(lr){9-10}
$m$ & Method & SNR = 2 & SNR = 5 & SNR = 2 & SNR = 5 & SNR = 2 & SNR = 5 & SNR = 2 & SNR = 5 \\
\midrule
\multicolumn{10}{c}{Case 1: $s^2_x = s^2_z=1$} \\
\multirow{4}{*}{100} & SOF & -0.003 (0.091) & -0.004 (0.054) & 0.004 (0.120) & 0.002 (0.075) & 0.001 (0.070) & -0.000 (0.043) & -0.007 (0.091) & -0.001 (0.056)
\\
 & ROF & -0.003 (0.096) & -0.004 (0.058) & 0.002 (0.126) & 0.003 (0.080) & 0.002 (0.073) & -0.000 (0.046) & -0.006 (0.094) & -0.000 (0.058)
\\
 & EOF & -0.017 (0.092) & -0.009 (0.055) & -0.018 (0.120) & -0.006 (0.076) & -0.009 (0.071) & -0.004 (0.044) & -0.022 (0.093) & -0.007 (0.057)
\\
 & Oracle & -0.002 (0.085) & -0.004 (0.050) & 0.002 (0.108) & 0.003 (0.068) & -0.000 (0.065) & 0.001 (0.040) & -0.004 (0.080) & -0.002 (0.047)
\\
\bottomrule
\end{tabular}
\end{table}

\begin{table}[htbp]
\scriptsize
\setlength{\tabcolsep}{2.4pt}
\renewcommand{\arraystretch}{0.8}
\centering
\caption{Empirical bias (RMSE) of point estimators $\hat{\beta}_1$ under T errors with $K=3$.}
\label{tab:bias_rmse_case_t_beta1_K3}
\begin{tabular}{@{}llrrrr@{\hspace{7pt}}rrrr@{}}
\toprule
 &  & \multicolumn{4}{c}{$N=260$} & \multicolumn{4}{c}{$N=520$} \\
\cmidrule(lr){3-6}\cmidrule(lr){7-10}
 &  & \multicolumn{2}{c}{SDP = 0.2} & \multicolumn{2}{c}{SDP = 0.5} & \multicolumn{2}{c}{SDP = 0.2} & \multicolumn{2}{c}{SDP = 0.5} \\
\cmidrule(lr){3-4}\cmidrule(lr){5-6}\cmidrule(lr){7-8}\cmidrule(lr){9-10}
$m$ & Method & SNR = 2 & SNR = 5 & SNR = 2 & SNR = 5 & SNR = 2 & SNR = 5 & SNR = 2 & SNR = 5 \\
\midrule
\multicolumn{10}{c}{Case 1: $s^2_x = s^2_z=1$} \\
\multirow{4}{*}{100} & SOF & -0.003 (0.095) & 0.000 (0.058) & -0.001 (0.119) & -0.001 (0.077) & 0.000 (0.071) & -0.001 (0.043) & -0.001 (0.090) & 0.000 (0.055)
\\
 & ROF & -0.001 (0.101) & 0.001 (0.062) & 0.001 (0.124) & -0.000 (0.079) & 0.001 (0.074) & 0.000 (0.045) & 0.000 (0.091) & 0.000 (0.057)
\\
 & EOF & -0.012 (0.096) & -0.003 (0.058) & -0.015 (0.120) & -0.007 (0.077) & -0.007 (0.072) & -0.004 (0.044) & -0.012 (0.090) & -0.004 (0.056)
\\
 & Oracle & -0.002 (0.090) & 0.000 (0.054) & 0.001 (0.112) & -0.001 (0.069) & -0.000 (0.065) & 0.000 (0.039) & -0.001 (0.081) & -0.001 (0.047)
\\
\bottomrule
\end{tabular}
\end{table}

\begin{table}[htbp]
\scriptsize
\setlength{\tabcolsep}{2.4pt}
\renewcommand{\arraystretch}{0.8}
\centering
\caption{Empirical bias (RMSE) of point estimators $\hat{\beta}_1$ under T errors with $K=5$.}
\label{tab:bias_rmse_case_t_beta1_K5}
\begin{tabular}{@{}llrrrr@{\hspace{7pt}}rrrr@{}}
\toprule
 &  & \multicolumn{4}{c}{$N=260$} & \multicolumn{4}{c}{$N=520$} \\
\cmidrule(lr){3-6}\cmidrule(lr){7-10}
 &  & \multicolumn{2}{c}{SDP = 0.2} & \multicolumn{2}{c}{SDP = 0.5} & \multicolumn{2}{c}{SDP = 0.2} & \multicolumn{2}{c}{SDP = 0.5} \\
\cmidrule(lr){3-4}\cmidrule(lr){5-6}\cmidrule(lr){7-8}\cmidrule(lr){9-10}
$m$ & Method & SNR = 2 & SNR = 5 & SNR = 2 & SNR = 5 & SNR = 2 & SNR = 5 & SNR = 2 & SNR = 5 \\
\midrule
\multicolumn{10}{c}{Case 1: $s^2_x = s^2_z=1$} \\
\multirow{4}{*}{100} & SOF & -0.001 (0.092) & -0.001 (0.056) & 0.001 (0.121) & 0.004 (0.074) & 0.000 (0.071) & -0.001 (0.044) & -0.003 (0.092) & -0.001 (0.057)
\\
 & ROF & 0.001 (0.097) & 0.000 (0.060) & 0.003 (0.126) & 0.005 (0.078) & 0.001 (0.073) & 0.000 (0.045) & 0.000 (0.094) & -0.000 (0.060)
\\
 & EOF & -0.015 (0.093) & -0.006 (0.057) & -0.020 (0.121) & -0.005 (0.074) & -0.011 (0.072) & -0.005 (0.044) & -0.018 (0.093) & -0.008 (0.058)
\\
 & Oracle & 0.001 (0.098) & -0.001 (0.052) & 0.003 (0.111) & 0.003 (0.068) & 0.000 (0.065) & 0.000 (0.039) & 0.002 (0.080) & -0.002 (0.047)
\\
\bottomrule
\end{tabular}
\end{table}

\begin{table}[htbp]
\scriptsize
\setlength{\tabcolsep}{2.4pt}
\renewcommand{\arraystretch}{0.8}
\centering
\caption{Empirical bias (RMSE) of point estimators $\hat{\beta}_2$ under Normal errors with $K=3$.}
\label{tab:bias_rmse_case_normal_beta2_K3}
\begin{tabular}{@{}llrrrr@{\hspace{7pt}}rrrr@{}}
\toprule
 &  & \multicolumn{4}{c}{$N=260$} & \multicolumn{4}{c}{$N=520$} \\
\cmidrule(lr){3-6}\cmidrule(lr){7-10}
 &  & \multicolumn{2}{c}{SDP = 0.2} & \multicolumn{2}{c}{SDP = 0.5} & \multicolumn{2}{c}{SDP = 0.2} & \multicolumn{2}{c}{SDP = 0.5} \\
\cmidrule(lr){3-4}\cmidrule(lr){5-6}\cmidrule(lr){7-8}\cmidrule(lr){9-10}
$m$ & Method & SNR = 2 & SNR = 5 & SNR = 2 & SNR = 5 & SNR = 2 & SNR = 5 & SNR = 2 & SNR = 5 \\
\midrule
\multicolumn{10}{c}{Case 1: $s^2_x = s^2_z=1$} \\
\multirow{4}{*}{100} & SOF & -0.000 (0.062) & -0.000 (0.039) & 0.002 (0.083) & 0.005 (0.052) & 0.000 (0.048) & 0.000 (0.032) & -0.004 (0.064) & -0.003 (0.039)
\\
 & ROF & -0.001 (0.065) & -0.000 (0.041) & 0.003 (0.086) & 0.005 (0.053) & 0.000 (0.050) & -0.000 (0.033) & -0.004 (0.065) & -0.003 (0.040)
\\
 & EOF & -0.004 (0.063) & -0.002 (0.039) & -0.005 (0.083) & 0.002 (0.052) & -0.003 (0.048) & -0.001 (0.032) & -0.009 (0.064) & -0.005 (0.039)
\\
 & Oracle & -0.001 (0.059) & -0.000 (0.037) & 0.001 (0.076) & 0.004 (0.046) & 0.000 (0.044) & -0.000 (0.029) & -0.004 (0.054) & -0.002 (0.034)
\\
\bottomrule
\end{tabular}
\end{table}

\begin{table}[htbp]
\scriptsize
\setlength{\tabcolsep}{2.4pt}
\renewcommand{\arraystretch}{0.8}
\centering
\caption{Empirical bias (RMSE) of point estimators $\hat{\beta}_2$ under Normal errors with $K=5$.}
\label{tab:bias_rmse_case_normal_beta2_K5}
\begin{tabular}{@{}llrrrr@{\hspace{7pt}}rrrr@{}}
\toprule
 &  & \multicolumn{4}{c}{$N=260$} & \multicolumn{4}{c}{$N=520$} \\
\cmidrule(lr){3-6}\cmidrule(lr){7-10}
 &  & \multicolumn{2}{c}{SDP = 0.2} & \multicolumn{2}{c}{SDP = 0.5} & \multicolumn{2}{c}{SDP = 0.2} & \multicolumn{2}{c}{SDP = 0.5} \\
\cmidrule(lr){3-4}\cmidrule(lr){5-6}\cmidrule(lr){7-8}\cmidrule(lr){9-10}
$m$ & Method & SNR = 2 & SNR = 5 & SNR = 2 & SNR = 5 & SNR = 2 & SNR = 5 & SNR = 2 & SNR = 5 \\
\midrule
\multicolumn{10}{c}{Case 1: $s^2_x = s^2_z=1$} \\
\multirow{4}{*}{100} & SOF & 0.005 (0.062) & -0.001 (0.040) & -0.001 (0.083) & -0.002 (0.049) & -0.002 (0.049) & 0.001 (0.030) & 0.001 (0.061) & -0.002 (0.040)
\\
 & ROF & 0.005 (0.066) & -0.000 (0.043) & -0.001 (0.086) & -0.001 (0.052) & -0.002 (0.051) & 0.001 (0.032) & 0.002 (0.062) & -0.001 (0.040)
\\
 & EOF & -0.003 (0.062) & -0.004 (0.040) & -0.013 (0.083) & -0.006 (0.050) & -0.007 (0.050) & -0.001 (0.031) & -0.007 (0.062) & -0.005 (0.040)
\\
 & Oracle & 0.005 (0.057) & -0.001 (0.037) & -0.001 (0.076) & -0.003 (0.045) & -0.001 (0.044) & 0.001 (0.027) & -0.001 (0.052) & -0.002 (0.034)
\\
\bottomrule
\end{tabular}
\end{table}

\begin{table}[htbp]
\scriptsize
\setlength{\tabcolsep}{2.4pt}
\renewcommand{\arraystretch}{0.8}
\centering
\caption{Empirical bias (RMSE) of point estimators $\hat{\beta}_2$ under T errors with $K=3$.}
\label{tab:bias_rmse_case_t_beta2_K3}
\begin{tabular}{@{}llrrrr@{\hspace{7pt}}rrrr@{}}
\toprule
 &  & \multicolumn{4}{c}{$N=260$} & \multicolumn{4}{c}{$N=520$} \\
\cmidrule(lr){3-6}\cmidrule(lr){7-10}
 &  & \multicolumn{2}{c}{SDP = 0.2} & \multicolumn{2}{c}{SDP = 0.5} & \multicolumn{2}{c}{SDP = 0.2} & \multicolumn{2}{c}{SDP = 0.5} \\
\cmidrule(lr){3-4}\cmidrule(lr){5-6}\cmidrule(lr){7-8}\cmidrule(lr){9-10}
$m$ & Method & SNR = 2 & SNR = 5 & SNR = 2 & SNR = 5 & SNR = 2 & SNR = 5 & SNR = 2 & SNR = 5 \\
\midrule
\multicolumn{10}{c}{Case 1: $s^2_x = s^2_z=1$} \\
\multirow{4}{*}{100} & SOF & -0.001 (0.066) & -0.001 (0.040) & 0.005 (0.085) & -0.001 (0.056) & -0.002 (0.050) & 0.001 (0.030) & -0.003 (0.063) & -0.002 (0.037)
\\
 & ROF & 0.001 (0.068) & -0.001 (0.042) & 0.007 (0.088) & -0.000 (0.056) & -0.001 (0.052) & 0.002 (0.031) & -0.002 (0.064) & -0.001 (0.037)
\\
 & EOF & -0.006 (0.066) & -0.003 (0.040) & -0.003 (0.084) & -0.004 (0.056) & -0.006 (0.051) & -0.000 (0.030) & -0.009 (0.064) & -0.004 (0.037)
\\
 & Oracle & -0.000 (0.062) & -0.001 (0.037) & 0.004 (0.076) & -0.001 (0.047) & -0.000 (0.046) & 0.000 (0.027) & -0.003 (0.056) & -0.001 (0.032)
\\
\bottomrule
\end{tabular}
\end{table}

\begin{table}[htbp]
\scriptsize
\setlength{\tabcolsep}{2.4pt}
\renewcommand{\arraystretch}{0.8}
\centering
\caption{Empirical bias (RMSE) of point estimators $\hat{\beta}_2$ under T errors with $K=5$.}
\label{tab:bias_rmse_case_t_beta2_K5}
\begin{tabular}{@{}llrrrr@{\hspace{7pt}}rrrr@{}}
\toprule
 &  & \multicolumn{4}{c}{$N=260$} & \multicolumn{4}{c}{$N=520$} \\
\cmidrule(lr){3-6}\cmidrule(lr){7-10}
 &  & \multicolumn{2}{c}{SDP = 0.2} & \multicolumn{2}{c}{SDP = 0.5} & \multicolumn{2}{c}{SDP = 0.2} & \multicolumn{2}{c}{SDP = 0.5} \\
\cmidrule(lr){3-4}\cmidrule(lr){5-6}\cmidrule(lr){7-8}\cmidrule(lr){9-10}
$m$ & Method & SNR = 2 & SNR = 5 & SNR = 2 & SNR = 5 & SNR = 2 & SNR = 5 & SNR = 2 & SNR = 5 \\
\midrule
\multicolumn{10}{c}{Case 1: $s^2_x = s^2_z=1$} \\
\multirow{4}{*}{100} & SOF & 0.001 (0.063) & 0.000 (0.040) & -0.001 (0.082) & -0.001 (0.051) & -0.001 (0.051) & -0.000 (0.031) & 0.003 (0.062) & -0.001 (0.040)
\\
 & ROF & 0.003 (0.066) & 0.001 (0.043) & 0.001 (0.085) & -0.000 (0.053) & -0.000 (0.052) & -0.000 (0.032) & 0.005 (0.063) & -0.000 (0.041)
\\
 & EOF & -0.007 (0.063) & -0.003 (0.041) & -0.013 (0.083) & -0.006 (0.051) & -0.007 (0.051) & -0.003 (0.032) & -0.006 (0.062) & -0.005 (0.041)
\\
 & Oracle & 0.001 (0.060) & 0.000 (0.038) & 0.001 (0.072) & 0.000 (0.045) & 0.000 (0.045) & -0.000 (0.028) & 0.000 (0.052) & -0.001 (0.033)
\\
\bottomrule
\end{tabular}
\end{table}

\begin{table}[htbp]
\scriptsize
\setlength{\tabcolsep}{2.4pt}
\renewcommand{\arraystretch}{0.8}
\centering
\caption{Empirical coverage rate $\times 100\%$ (average length) of 90\% confidence intervals for $\beta_1$ under Normal errors with $K=3$.}
\label{tab:coverage_length_case_normal_beta1_K3}
\begin{tabular}{@{}llrrrr@{\hspace{7pt}}rrrr@{}}
\toprule
 &  & \multicolumn{4}{c}{$N=260$} & \multicolumn{4}{c}{$N=520$} \\
\cmidrule(lr){3-6}\cmidrule(lr){7-10}
 &  & \multicolumn{2}{c}{SDP = 0.2} & \multicolumn{2}{c}{SDP = 0.5} & \multicolumn{2}{c}{SDP = 0.2} & \multicolumn{2}{c}{SDP = 0.5} \\
\cmidrule(lr){3-4}\cmidrule(lr){5-6}\cmidrule(lr){7-8}\cmidrule(lr){9-10}
$m$ & Method & SNR = 2 & SNR = 5 & SNR = 2 & SNR = 5 & SNR = 2 & SNR = 5 & SNR = 2 & SNR = 5 \\
\midrule
\multicolumn{10}{c}{Case 1: $s^2_x = s^2_z=1$} \\
\multirow{4}{*}{100} & SOF & 92.1 (0.303) & 89.6 (0.188) & 88.3 (0.387) & 89.8 (0.238) & 89.5 (0.229) & 91.1 (0.143) & 88.7 (0.294) & 89.7 (0.181)
\\
 & ROF & 91.5 (0.319) & 90.3 (0.198) & 87.8 (0.403) & 89.7 (0.249) & 88.9 (0.237) & 90.2 (0.148) & 89.1 (0.302) & 89.6 (0.186)
\\
 & EOF & 89.6 (0.289) & 88.0 (0.179) & 82.9 (0.340) & 85.5 (0.209) & 84.3 (0.202) & 85.5 (0.126) & 81.0 (0.235) & 82.0 (0.144)
\\
 & Oracle & 92.4 (0.289) & 90.0 (0.179) & 88.5 (0.359) & 89.5 (0.219) & 89.2 (0.209) & 89.0 (0.129) & 90.3 (0.259) & 88.8 (0.158)
\\
\bottomrule
\end{tabular}
\end{table}

\begin{table}[htbp]
\scriptsize
\setlength{\tabcolsep}{2.4pt}
\renewcommand{\arraystretch}{0.8}
\centering
\caption{Empirical coverage rate $\times 100\%$ (average length) of 90\% confidence intervals for $\beta_1$ under Normal errors with $K=5$.}
\label{tab:coverage_length_case_normal_beta1_K5}
\begin{tabular}{@{}llrrrr@{\hspace{7pt}}rrrr@{}}
\toprule
 &  & \multicolumn{4}{c}{$N=260$} & \multicolumn{4}{c}{$N=520$} \\
\cmidrule(lr){3-6}\cmidrule(lr){7-10}
 &  & \multicolumn{2}{c}{SDP = 0.2} & \multicolumn{2}{c}{SDP = 0.5} & \multicolumn{2}{c}{SDP = 0.2} & \multicolumn{2}{c}{SDP = 0.5} \\
\cmidrule(lr){3-4}\cmidrule(lr){5-6}\cmidrule(lr){7-8}\cmidrule(lr){9-10}
$m$ & Method & SNR = 2 & SNR = 5 & SNR = 2 & SNR = 5 & SNR = 2 & SNR = 5 & SNR = 2 & SNR = 5 \\
\midrule
\multicolumn{10}{c}{Case 1: $s^2_x = s^2_z=1$} \\
\multirow{4}{*}{100} & SOF & 88.8 (0.292) & 91.5 (0.181) & 89.0 (0.378) & 87.2 (0.232) & 90.4 (0.227) & 89.6 (0.140) & 88.9 (0.295) & 89.6 (0.183)
\\
 & ROF & 89.3 (0.310) & 90.9 (0.193) & 88.6 (0.401) & 87.5 (0.245) & 89.1 (0.235) & 89.5 (0.145) & 89.2 (0.303) & 90.2 (0.188)
\\
 & EOF & 87.7 (0.289) & 90.5 (0.180) & 84.5 (0.343) & 83.9 (0.211) & 85.0 (0.203) & 85.3 (0.126) & 80.6 (0.236) & 80.0 (0.146)
\\
 & Oracle & 90.3 (0.279) & 91.2 (0.173) & 89.1 (0.349) & 88.1 (0.213) & 89.8 (0.205) & 89.1 (0.127) & 89.5 (0.256) & 91.0 (0.156)
\\
\bottomrule
\end{tabular}
\end{table}

\begin{table}[htbp]
\scriptsize
\setlength{\tabcolsep}{2.4pt}
\renewcommand{\arraystretch}{0.8}
\centering
\caption{Empirical coverage rate $\times 100\%$ (average length) of 90\% confidence intervals for $\beta_1$ under T errors with $K=3$.}
\label{tab:coverage_length_case_t_beta1_K3}
\begin{tabular}{@{}llrrrr@{\hspace{7pt}}rrrr@{}}
\toprule
 &  & \multicolumn{4}{c}{$N=260$} & \multicolumn{4}{c}{$N=520$} \\
\cmidrule(lr){3-6}\cmidrule(lr){7-10}
 &  & \multicolumn{2}{c}{SDP = 0.2} & \multicolumn{2}{c}{SDP = 0.5} & \multicolumn{2}{c}{SDP = 0.2} & \multicolumn{2}{c}{SDP = 0.5} \\
\cmidrule(lr){3-4}\cmidrule(lr){5-6}\cmidrule(lr){7-8}\cmidrule(lr){9-10}
$m$ & Method & SNR = 2 & SNR = 5 & SNR = 2 & SNR = 5 & SNR = 2 & SNR = 5 & SNR = 2 & SNR = 5 \\
\midrule
\multicolumn{10}{c}{Case 1: $s^2_x = s^2_z=1$} \\
\multirow{4}{*}{100} & SOF & 89.4 (0.301) & 90.5 (0.187) & 89.1 (0.382) & 88.9 (0.236) & 89.1 (0.227) & 90.2 (0.140) & 87.4 (0.286) & 90.3 (0.177)
\\
 & ROF & 87.6 (0.313) & 88.1 (0.194) & 89.0 (0.394) & 88.5 (0.244) & 89.4 (0.234) & 90.8 (0.145) & 88.1 (0.297) & 90.2 (0.184)
\\
 & EOF & 86.6 (0.286) & 87.8 (0.178) & 84.9 (0.336) & 84.4 (0.207) & 84.9 (0.201) & 86.2 (0.124) & 79.1 (0.232) & 81.9 (0.144)
\\
 & Oracle & 91.7 (0.289) & 90.4 (0.179) & 88.6 (0.359) & 88.5 (0.219) & 90.0 (0.209) & 91.1 (0.129) & 88.2 (0.259) & 91.2 (0.158)
\\
\bottomrule
\end{tabular}
\end{table}

\begin{table}[htbp]
\scriptsize
\setlength{\tabcolsep}{2.4pt}
\renewcommand{\arraystretch}{0.8}
\centering
\caption{Empirical coverage rate $\times 100\%$ (average length) of 90\% confidence intervals for $\beta_1$ under T errors with $K=5$.}
\label{tab:coverage_length_case_t_beta1_K5}
\begin{tabular}{@{}llrrrr@{\hspace{7pt}}rrrr@{}}
\toprule
 &  & \multicolumn{4}{c}{$N=260$} & \multicolumn{4}{c}{$N=520$} \\
\cmidrule(lr){3-6}\cmidrule(lr){7-10}
 &  & \multicolumn{2}{c}{SDP = 0.2} & \multicolumn{2}{c}{SDP = 0.5} & \multicolumn{2}{c}{SDP = 0.2} & \multicolumn{2}{c}{SDP = 0.5} \\
\cmidrule(lr){3-4}\cmidrule(lr){5-6}\cmidrule(lr){7-8}\cmidrule(lr){9-10}
$m$ & Method & SNR = 2 & SNR = 5 & SNR = 2 & SNR = 5 & SNR = 2 & SNR = 5 & SNR = 2 & SNR = 5 \\
\midrule
\multicolumn{10}{c}{Case 1: $s^2_x = s^2_z=1$} \\
\multirow{4}{*}{100} & SOF & 89.3 (0.300) & 89.3 (0.186) & 88.6 (0.390) & 87.9 (0.235) & 88.5 (0.225) & 87.6 (0.140) & 89.0 (0.288) & 89.2 (0.178)
\\
 & ROF & 88.5 (0.306) & 89.1 (0.190) & 88.2 (0.394) & 86.8 (0.241) & 89.8 (0.232) & 87.8 (0.145) & 89.4 (0.296) & 87.4 (0.185)
\\
 & EOF & 88.2 (0.287) & 88.9 (0.179) & 84.0 (0.340) & 84.8 (0.209) & 84.5 (0.201) & 82.4 (0.125) & 80.0 (0.233) & 79.7 (0.144)
\\
 & Oracle & 90.1 (0.280) & 90.2 (0.173) & 88.6 (0.349) & 89.0 (0.213) & 88.9 (0.205) & 89.8 (0.127) & 89.0 (0.256) & 90.9 (0.156)
\\
\bottomrule
\end{tabular}
\end{table}

\begin{table}[htbp]
\scriptsize
\setlength{\tabcolsep}{2.4pt}
\renewcommand{\arraystretch}{0.8}
\centering
\caption{Empirical coverage rate $\times 100\%$ (average length) of 90\% confidence intervals for $\beta_2$ under Normal errors with $K=3$.}
\label{tab:coverage_length_case_normal_beta2_K3}
\begin{tabular}{@{}llrrrr@{\hspace{7pt}}rrrr@{}}
\toprule
 &  & \multicolumn{4}{c}{$N=260$} & \multicolumn{4}{c}{$N=520$} \\
\cmidrule(lr){3-6}\cmidrule(lr){7-10}
 &  & \multicolumn{2}{c}{SDP = 0.2} & \multicolumn{2}{c}{SDP = 0.5} & \multicolumn{2}{c}{SDP = 0.2} & \multicolumn{2}{c}{SDP = 0.5} \\
\cmidrule(lr){3-4}\cmidrule(lr){5-6}\cmidrule(lr){7-8}\cmidrule(lr){9-10}
$m$ & Method & SNR = 2 & SNR = 5 & SNR = 2 & SNR = 5 & SNR = 2 & SNR = 5 & SNR = 2 & SNR = 5 \\
\midrule
\multicolumn{10}{c}{Case 1: $s^2_x = s^2_z=1$} \\
\multirow{4}{*}{100} & SOF & 90.7 (0.211) & 91.8 (0.132) & 89.2 (0.269) & 89.2 (0.168) & 90.5 (0.160) & 88.4 (0.100) & 89.0 (0.204) & 88.5 (0.127)
\\
 & ROF & 91.8 (0.223) & 92.1 (0.140) & 89.4 (0.280) & 90.0 (0.174) & 90.1 (0.165) & 88.4 (0.103) & 87.9 (0.206) & 89.0 (0.129)
\\
 & EOF & 88.5 (0.201) & 90.4 (0.126) & 85.0 (0.233) & 83.7 (0.145) & 86.3 (0.140) & 83.5 (0.088) & 79.9 (0.162) & 80.7 (0.100)
\\
 & Oracle & 91.9 (0.201) & 90.8 (0.125) & 89.3 (0.246) & 90.6 (0.153) & 90.3 (0.144) & 87.9 (0.090) & 89.6 (0.177) & 88.8 (0.110)
\\
\bottomrule
\end{tabular}
\end{table}

\begin{table}[htbp]
\scriptsize
\setlength{\tabcolsep}{2.4pt}
\renewcommand{\arraystretch}{0.8}
\centering
\caption{Empirical coverage rate $\times 100\%$ (average length) of 90\% confidence intervals for $\beta_2$ under Normal errors with $K=5$.}
\label{tab:coverage_length_case_normal_beta2_K5}
\begin{tabular}{@{}llrrrr@{\hspace{7pt}}rrrr@{}}
\toprule
 &  & \multicolumn{4}{c}{$N=260$} & \multicolumn{4}{c}{$N=520$} \\
\cmidrule(lr){3-6}\cmidrule(lr){7-10}
 &  & \multicolumn{2}{c}{SDP = 0.2} & \multicolumn{2}{c}{SDP = 0.5} & \multicolumn{2}{c}{SDP = 0.2} & \multicolumn{2}{c}{SDP = 0.5} \\
\cmidrule(lr){3-4}\cmidrule(lr){5-6}\cmidrule(lr){7-8}\cmidrule(lr){9-10}
$m$ & Method & SNR = 2 & SNR = 5 & SNR = 2 & SNR = 5 & SNR = 2 & SNR = 5 & SNR = 2 & SNR = 5 \\
\midrule
\multicolumn{10}{c}{Case 1: $s^2_x = s^2_z=1$} \\
\multirow{4}{*}{100} & SOF & 90.0 (0.203) & 89.1 (0.127) & 88.1 (0.262) & 90.0 (0.164) & 89.0 (0.158) & 89.7 (0.099) & 90.2 (0.205) & 90.0 (0.128)
\\
 & ROF & 90.4 (0.217) & 89.6 (0.136) & 88.9 (0.278) & 90.0 (0.173) & 90.1 (0.163) & 89.4 (0.103) & 90.7 (0.208) & 88.9 (0.130)
\\
 & EOF & 89.8 (0.201) & 88.7 (0.126) & 84.7 (0.235) & 85.1 (0.147) & 84.4 (0.141) & 85.8 (0.088) & 80.9 (0.162) & 79.2 (0.101)
\\
 & Oracle & 91.7 (0.194) & 89.8 (0.121) & 89.5 (0.239) & 89.4 (0.149) & 90.0 (0.142) & 90.5 (0.089) & 90.6 (0.175) & 89.0 (0.109)
\\
\bottomrule
\end{tabular}
\end{table}

\begin{table}[htbp]
\scriptsize
\setlength{\tabcolsep}{2.4pt}
\renewcommand{\arraystretch}{0.8}
\centering
\caption{Empirical coverage rate $\times 100\%$ (average length) of 90\% confidence intervals for $\beta_2$ under T errors with $K=3$.}
\label{tab:coverage_length_case_t_beta2_K3}
\begin{tabular}{@{}llrrrr@{\hspace{7pt}}rrrr@{}}
\toprule
 &  & \multicolumn{4}{c}{$N=260$} & \multicolumn{4}{c}{$N=520$} \\
\cmidrule(lr){3-6}\cmidrule(lr){7-10}
 &  & \multicolumn{2}{c}{SDP = 0.2} & \multicolumn{2}{c}{SDP = 0.5} & \multicolumn{2}{c}{SDP = 0.2} & \multicolumn{2}{c}{SDP = 0.5} \\
\cmidrule(lr){3-4}\cmidrule(lr){5-6}\cmidrule(lr){7-8}\cmidrule(lr){9-10}
$m$ & Method & SNR = 2 & SNR = 5 & SNR = 2 & SNR = 5 & SNR = 2 & SNR = 5 & SNR = 2 & SNR = 5 \\
\midrule
\multicolumn{10}{c}{Case 1: $s^2_x = s^2_z=1$} \\
\multirow{4}{*}{100} & SOF & 90.0 (0.210) & 88.9 (0.131) & 88.7 (0.264) & 88.4 (0.165) & 89.2 (0.158) & 90.3 (0.098) & 88.0 (0.196) & 91.2 (0.123)
\\
 & ROF & 88.3 (0.218) & 88.7 (0.137) & 87.7 (0.271) & 89.7 (0.170) & 89.5 (0.163) & 89.8 (0.102) & 89.4 (0.203) & 90.8 (0.127)
\\
 & EOF & 87.9 (0.199) & 86.7 (0.125) & 83.8 (0.231) & 83.5 (0.144) & 84.0 (0.139) & 85.4 (0.087) & 81.1 (0.160) & 83.5 (0.100)
\\
 & Oracle & 89.8 (0.201) & 91.2 (0.125) & 89.3 (0.246) & 90.6 (0.153) & 89.4 (0.144) & 91.0 (0.090) & 89.6 (0.177) & 90.8 (0.110)
\\
\bottomrule
\end{tabular}
\end{table}

\begin{table}[htbp]
\scriptsize
\setlength{\tabcolsep}{2.4pt}
\renewcommand{\arraystretch}{0.8}
\centering
\caption{Empirical coverage rate $\times 100\%$ (average length) of 90\% confidence intervals for $\beta_2$ under T errors with $K=5$.}
\label{tab:coverage_length_case_t_beta2_K5}
\begin{tabular}{@{}llrrrr@{\hspace{7pt}}rrrr@{}}
\toprule
 &  & \multicolumn{4}{c}{$N=260$} & \multicolumn{4}{c}{$N=520$} \\
\cmidrule(lr){3-6}\cmidrule(lr){7-10}
 &  & \multicolumn{2}{c}{SDP = 0.2} & \multicolumn{2}{c}{SDP = 0.5} & \multicolumn{2}{c}{SDP = 0.2} & \multicolumn{2}{c}{SDP = 0.5} \\
\cmidrule(lr){3-4}\cmidrule(lr){5-6}\cmidrule(lr){7-8}\cmidrule(lr){9-10}
$m$ & Method & SNR = 2 & SNR = 5 & SNR = 2 & SNR = 5 & SNR = 2 & SNR = 5 & SNR = 2 & SNR = 5 \\
\midrule
\multicolumn{10}{c}{Case 1: $s^2_x = s^2_z=1$} \\
\multirow{4}{*}{100} & SOF & 89.9 (0.212) & 88.6 (0.134) & 89.8 (0.271) & 89.3 (0.166) & 87.2 (0.156) & 87.3 (0.098) & 89.4 (0.199) & 88.6 (0.125)
\\
 & ROF & 90.8 (0.214) & 86.7 (0.133) & 89.6 (0.272) & 89.7 (0.168) & 87.1 (0.161) & 88.6 (0.101) & 88.9 (0.204) & 88.1 (0.128)
\\
 & EOF & 89.0 (0.200) & 87.7 (0.125) & 84.7 (0.234) & 85.7 (0.145) & 81.8 (0.140) & 84.0 (0.088) & 80.4 (0.161) & 80.8 (0.100)
\\
 & Oracle & 89.4 (0.194) & 89.1 (0.121) & 90.0 (0.239) & 90.6 (0.149) & 88.7 (0.142) & 88.8 (0.089) & 90.2 (0.175) & 91.0 (0.109)
\\
\bottomrule
\end{tabular}
\end{table}

\end{document}